\documentclass[12pt,a4paper]{article}

\usepackage[utf8]{inputenc}

\usepackage{amscd}
\usepackage{amsmath,amssymb,amsthm,mathrsfs, mathtools}
\usepackage[margin=2.5cm]{geometry}

\usepackage{tikz} 

\usetikzlibrary{matrix,arrows,calc}

\usepackage{parskip}
\usepackage{enumerate}
\usepackage{enumitem}
\setlist[enumerate, 1]{label=(\roman*)}
\setlist[itemize]{leftmargin=1.5em}
\setlist[description]{leftmargin=1em}

\usepackage{bm}

\usepackage{overpic}
\usepackage{contour}
\contourlength{2pt}
\usepackage{graphicx}
\usepackage{float}
\graphicspath{ {figures/} }
\DeclareGraphicsExtensions{.pdf, .jpg, .jpeg, .png}
\usepackage{xcolor}
\definecolor{darkgreen}{RGB}{0, 102, 0}

\usepackage[hidelinks, backref=page]{hyperref}  
\usepackage[format=plain,labelsep=period,justification=justified,font=small,labelfont=bf, margin=1em]{caption}

\makeatletter 
\newcommand\nobreakpar{\par\nobreak\@afterheading} 
\makeatother

\input{custommacros}
\numberwithin{equation}{section}

\usepackage{amsthm}

\usepackage{thmtools}
\declaretheorem[style=plain,name=Theorem,qed={\tiny$\blacksquare$},numberwithin=section]{theorem}

\declaretheorem[style=definition,name=Definition,sibling=theorem,qed={\tiny$\blacksquare$}]{definition}
\declaretheorem[style=plain,name=Lemma,sibling=theorem,qed={\tiny$\blacksquare$}]{lemma}

\declaretheorem[style=definition,name=Remark,sibling=theorem,qed={\tiny$\blacksquare$}]{remark}
\declaretheorem[style=definition,name=Question,sibling=theorem,qed={\tiny$\blacksquare$}]{question}

\declaretheorem[style=plain,name=Proposition,sibling=theorem,qed={\tiny$\blacksquare$}]{proposition}

\usepackage[breakable]{tcolorbox}
\declaretheorem[style=remark,name=Coordinates,sibling=theorem]{coordinates}
\tcolorboxenvironment{coordinates}{colframe=blue!05!white,breakable,before skip=10pt,after skip=10pt,colback=blue!05!white}

\makeatletter
\newcommand{\HScaled}{%
  \mathpalette\@HScaled{}%
}
\newcommand*{\@HScaled}{%
  \scalebox{0.35}{$\square\m@th$}%
}
\makeatother

\title{Principal binets}

\author{
	Niklas C. Affolter\thanks{Institut für Diskrete Mathematik und Geometrie, TU Wien, Austria; and
	Institut für Mathematik, TU Berlin, Germany. 
   \textit{E-mail address}: \texttt{affolter@posteo.net}},
   Jan Techter\thanks{Institut für Mathematik, TU Berlin, Germany. 
    \textit{E-mail address}: \texttt{techter@math.tu-berlin.de}} 
    \bigskip\bigskip\\
}

\date{\today}

\begin{document}

\maketitle

\noindent
\textbf{Abstract.}
Conjugate line parametrizations of surfaces were first discretized  almost a century ago as quad meshes with planar faces.
With the recent development of discrete differential geometry, two discretizations of principal curvature line parametrizations were discovered:
circular nets and conical nets, both of which are special cases of discrete conjugate nets.
Subsequently, circular and conical nets were given a unified description as isotropic line congruences in the Lie quadric.
We propose a generalization by considering polar pairs of line congruences in the ambient space of the Lie quadric.
These correspond to pairs of discrete conjugate nets with orthogonal edges, which we call principal binets,
a new and more general discretization of principal curvature line parametrizations.
We also introduce two new discretizations of orthogonal and Gauß-orthogonal parametrizations.
All our discretizations are subject to the transformation group principle,
which means that they satisfy the corresponding Lie, Möbius, or Laguerre invariance respectively, in analogy to the smooth theory.
Finally, we show that they satisfy the consistency principle, which means that our definitions generalize to higher dimensional square lattices.
Our work expands on recent work by Dellinger on checkerboard patterns.

\newpage

\setcounter{tocdepth}{1}
\tableofcontents

\newpage

\section{Introduction} \label{sec:introduction}

Discrete differential geometry is an area of mathematics that aims to discretize differential geometry in a way that is \emph{structure preserving} \cite{ddgbook}. This means that the goal is to discretize sets of definitions and theorems simultaneously and consistently with each other. In this paper we deal with discretizations of parametrized surfaces.

Let us first characterize the parametrizations that are relevant to us.
\begin{enumerate}
	\item An \emph{orthogonal parametrization} is such that the first fundamental form is diagonal.
	\item A \emph{conjugate line parametrization} is such that the second fundamental form is diagonal.
	\item Another type of parametrization that we encounter is such that the third fundamental form is diagonal. We call this case a \emph{Gauß-orthogonal parametrization}.
\end{enumerate}

Due to the linear dependence of the three fundamental forms of a parametrization,
a conjugate line parametrization is orthogonal if and only if it is Gauß-orthogonal. Moreover, a conjugate line parametrization that is orthogonal (or equivalently Gauß-orthogonal) is called a \emph{(principal) curvature line parametrization}.

We follow the approach to consider \emph{discrete nets} as discrete analogues of parametrized surfaces.
Discrete nets are maps from $\Z^2$ (or a more general quad graph) to $\R^3$ (or some other ambient space).
A commonly used discretization of conjugate line parametrizations
is given by \emph{discrete conjugate nets}, which are discrete nets,
such that the image of each quad is contained in a plane \cite{sauerqnet, dsqnet}.
There are two well established discretizations of curvature line parametrizations.
First, there are \emph{circular nets}, which are discrete conjugate nets such that the image of each quad is contained in a circle \cite{cdscircular}.
Secondly, there are \emph{conical nets}, which are discrete conjugate nets such that around each vertex the four planes are in contact with a cone of revolution \cite{lpwywconical,bsorganizing,pwconical}.

\begin{figure}[tb]
  \centering
  \includegraphics[width=0.32\textwidth]{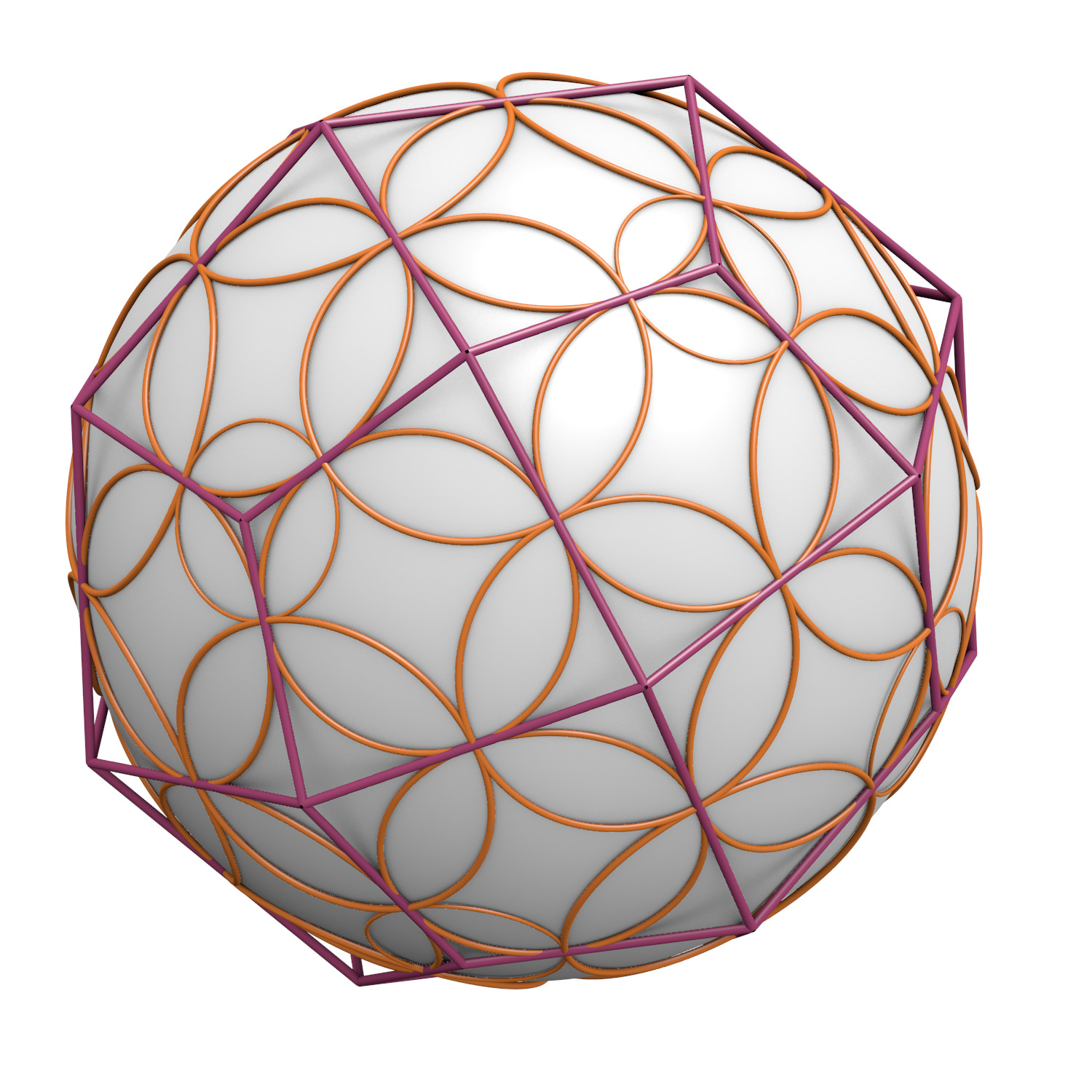}
  \includegraphics[width=0.32\textwidth]{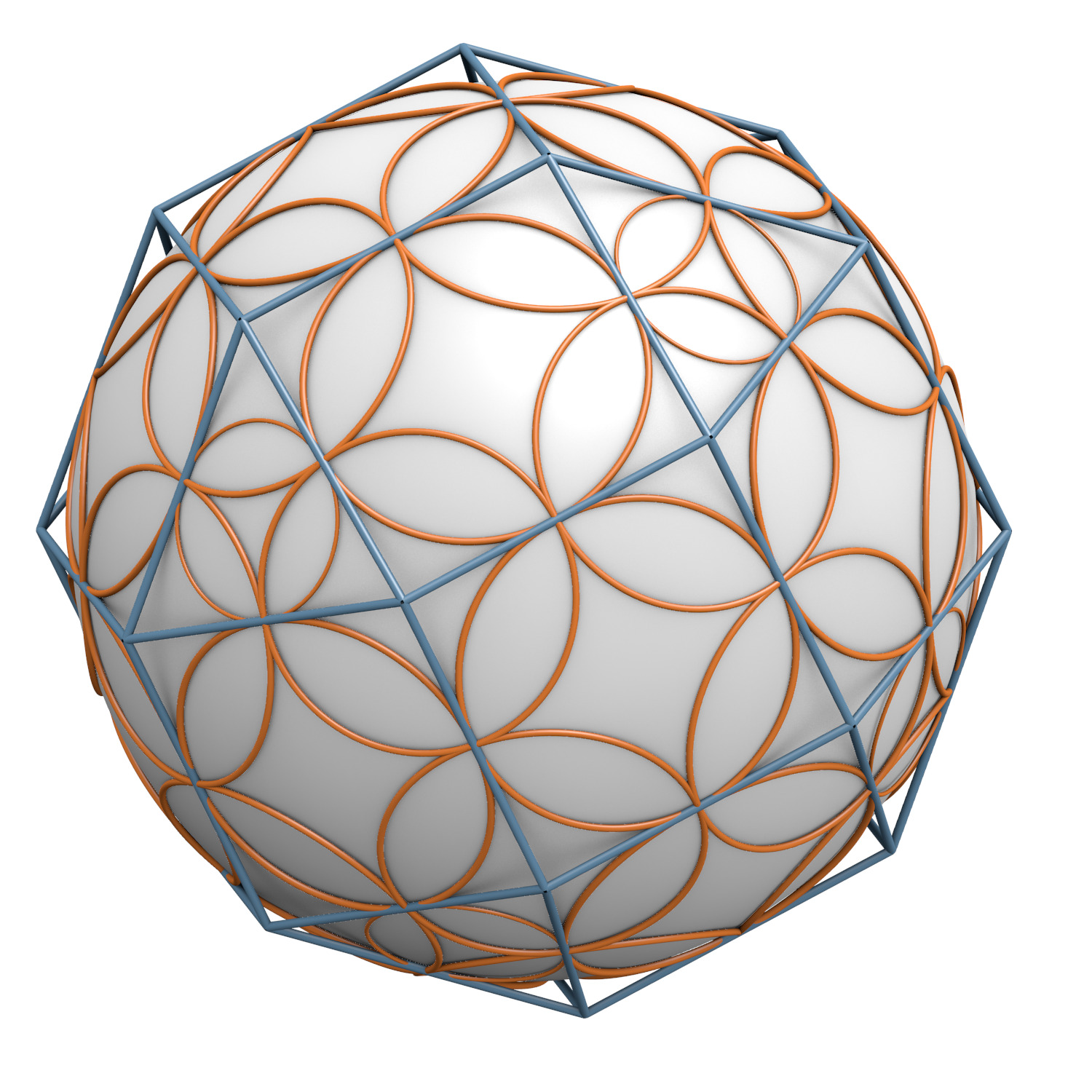}
  \includegraphics[width=0.33\textwidth]{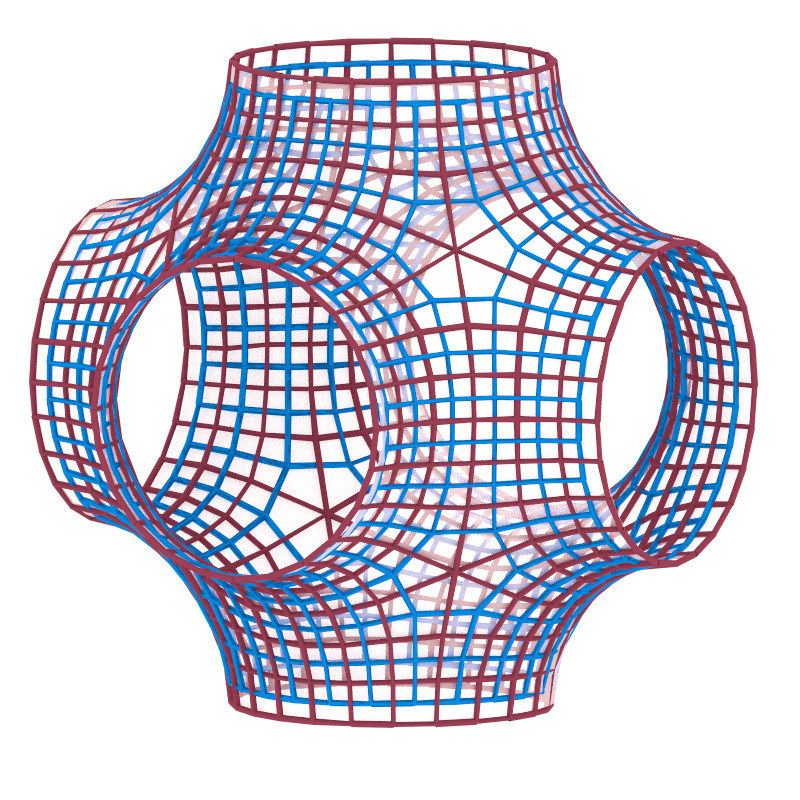}
  \caption{
    Left/middle: Two dual Koebe polyhedra.
    Right: Pair of discrete minimal surfaces with dual combinatorics.
  }
  \label{fig:examples}
\end{figure}
There is no commonly used discretization of general orthogonal parametrizations in terms of single discrete nets in the literature. However, there are several interesting examples of discrete surfaces that arise naturally as pairs of nets with dual combinatorics.
We give three examples, which constitute the main motivation for our approach.
\begin{enumerate}
\item \label{itm:examplekoebe}
  Koebe polyhedra are discrete conjugate nets with edges that are tangent to the unit sphere.
  For each Koebe polyhedron there is a dual Koebe polyhedron with polar edges (see Figure~\ref{fig:examples}, left/middle). 
  Koebe polyhedra may be interpreted as discretizations of the Gauß map of a minimal surface.
  Consequently, in \cite{bhssminimal} discrete minimal surfaces are constructed from Koebe polyhedra (see Figure~\ref{fig:examples}, right).
  Moreover, in \cite{bhsring} the authors consider a generalization of Koebe polyhedra, which also come in pairs and are discretizations of the Gauß map of constant mean curvature surfaces \cite{bhscmc}.
\item
  \label{itm:exampleconfocal}
  Discrete confocal quadrics as introduced in \cite{bsstconfocali, bsstconfocalii}
  appear as pairs of higher dimensional discrete nets (see Figure~\ref{fig:discrete-confocal}, left).
  Taking a 2-dimensional slice of a 3-dimensional system of discrete confocal quadrics
  results in a pair of combinatorially dual discrete nets describing one discrete quadric \cite{hstellipsoid} (see Figure~\ref{fig:discrete-confocal}, right).
\item
  Starting with a circular net, there is a geometric construction
  by which one can obtain a corresponding conical net of dual combinatorics \cite{bsorganizing,pwconical} (see Figure~\ref{fig:circular-conical-binet}).
  Vice versa, starting with a conical net, a corresponding circular net can be constructed in a similar way.
  Together the circular and conical net form a pair of nets with dual combinatorics.
\end{enumerate}

In each of these examples the two discrete nets are discrete conjugate nets such that pairs of dual edges are orthogonal.
In this sense, they all constitute discretizations of curvature line parametrizations. Moreover, the discrete conditions for the discretization of a conjugate line parametrization and of an orthogonal parametrization can be clearly separated.

Motivated by the examples, we now develop a theory of discretizations of surface parametrizations using pairs of nets.
For this purpose, we denote by $D$ the union of vertices and faces of $\Z^2$, that is $D=\Z^2 \cup F(\Z^2)$. We call $d,d'\in D$ \emph{incident} if $d$ is a face that contains the vertex $d'$ or vice versa.
The idea we put forward in this paper is to consider the notion of \emph{binets}, which are maps
\begin{align}
	b: D=\Z^2 \cup F(\Z^2) \rightarrow \R^3.
\end{align}
The restriction of $D$ to vertices or faces is isomorphic to $\Z^2$, thus we consider the restriction of a binet to vertices or faces to be a discrete net.
For each edge $(v,v')$ of $\Z^2$ there is a unique dual edge $(f,f')$, and we call the pair of edge and dual edge a \emph{cross}.

A \emph{conjugate binet} is a binet such that the two restrictions to vertices and faces are discrete conjugate nets each.
Therefore, a conjugate binet is simply a pair of discrete conjugate nets.

An \emph{orthogonal binet} is a binet such that at each cross, the line $b(v) \vee b(v')$ is orthogonal to the line $b(f) \vee b(f')$.
Note that while there is no condition for conjugate binets that involves a combination of points from both vertices and faces of $\Z^2$,
this is the case for orthogonal binets.

Bobenko and Suris defined the so called \emph{transformation group principle} as a desired property for structure preserving discretizations \cite{bsorganizing}.
The principle states that discretizations should be invariant under the same group of transformations of the ambient space as in the corresponding smooth theory. In particular, it is well-known that
\begin{enumerate}
	\item conjugate parametrizations are invariant under projective transformations,
	\item orthogonal parametrizations are invariant under Möbius transformations,
	\item Gauß-orthogonal parametrizations are invariant under Laguerre transformations,
	\item and curvature line parametrizations are invariant under Lie transformations.
\end{enumerate}

And indeed, discrete conjugate nets are invariant under projective transformations and thereby so are conjugate binets.
Hence, the transformation principle is satisfied.
Orthogonal parametrizations are invariant under Möbius transformations.
However orthogonal binets are not invariant under applying a Möbius transformation to the points of the binet.
Yet we may achieve Möbius invariance for orthogonal binets in the following way.

Given a binet $b$ consider a map $b_\sp$ from $D$ to the set of spheres, such that the center of $b_\sp(d)$ is $b(d)$ for all $d \in D$, and such that $b_\sp(d)$ is orthogonal to $b_\sp(d')$ whenever $d,d'$ are incident. We call such a map an \emph{orthogonal sphere representation} of $b$.
We show that a binet has an orthogonal sphere representation if and only if $b$ is an orthogonal binet.
As every orthogonal binet comes with an orthogonal sphere representation, this allows us to apply a Möbius transformation to an orthogonal binet by applying it to an orthogonal sphere representation instead. In this sense, orthogonal binets are \emph{Möbius invariant}, which shows that they satisfy the transformation group principle.

In the projective model of Möbius geometry, $\R^3$ is represented by a quadric $\mobq$ of signature $\texttt{(++++-)}$ in $\RP^4$, which is called the \emph{Möbius quadric}.
Spheres in $\R^3$ correspond to points outside of $\mobq$. Thus, given an orthogonal binet $b$ and a sphere representation $b_\sp$, we can consider the \emph{Möbius lift}
\[
  b_\mobq: D \rightarrow \RP^4
\]
of $b$, which is such that $b_\mobq(d)$ is the point corresponding to the sphere $b_\sp(d)$ for all $d\in D$. The orthogonality of $b$ is reflected in the fact that $b_\mobq(d)$ is polar (with respect to $\mobq$) to $b_\mobq(d')$ whenever $d$ and $d'$ are incident.
We call a binet with that property a \emph{polar binet},
and show that Möbius lifts of orthogonal binets are in bijection with polar binets in the projective model of Möbius geometry.

Another type of map that we consider are \emph{bi*nets}. A bi*net is a map from $D$ to the space of planes of $\R^3$. We think of a bi*net as discretizing the tangent planes of a surface.
An \emph{orthogonal bi*net} is a bi*net $b$ such that at each cross, the line $b(v) \cap b(v')$ is orthogonal to the line $b(f) \cap b(f')$. This may be viewed as a discretization of a Gauß-orthogonal parametrization.

Given a bi*net $b$, consider a map $b_\ci$ from $D$ to the set of circles in the unit sphere $\unis$, such that the axis of $b_\ci(d)$ is orthogonal to $b(d)$ for all $d \in D$, and such that $b_\ci(d)$ is orthogonal to $b_\ci(d')$ whenever $d$ and $d'$ are incident. We call such a map an \emph{orthogonal circle representation} of $b$. We show that a bi*net has an orthogonal circle representation if and only if $b$ is an orthogonal bi*net. We may represent each circle $b_\ci(d)$ by a point $n(d)$ outside of $\unis$, and we call the resulting map the \emph{normal binet} of $b$.
We consider the normal binet to be a discretization of the normal map (or Gauß map) of a surface (in Gauß-orthogonal parametrization).
Normal binets are polar binets with respect to $\unis$.

In the projective model of \emph{Laguerre geometry}, oriented planes of $\R^3$ are represented by points on a quadric $\blac$ of signature $\texttt{(+++-0)}$ in $\RP^4$, which is called the \emph{Blaschke cylinder}.
Let $b$ be an orthogonal bi*net and $b_\ci$ an  orthogonal circle representation of $b$.
We show that $b(d)$ and $b_\ci(d)$ together define a unique point $b_\blac(d)$ in $\RP^4$, and we call the resulting map
\[
  b_\blac : D \rightarrow \RP^4
\]
the \emph{Laguerre lift} of $b$.
We show that $b_\blac$ is a polar binet with respect to $\blac$,
and that Laguerre lifts of orthogonal bi*nets are in bijection with polar binets in the projective model of Laguerre geomemtry.
Using the Laguerre lift, we have a way to apply Laguerre transformations to an orthogonal bi*net such that the result is again an orthogonal bi*net.
This is analogous to the smooth theory, as Gauß-orthogonal parametrizations are preserved by Laguerre transformations.

As a discretization of curvature line parametrizations we introduce \emph{principal binets},
which are binets that are both conjugate and orthogonal.
Note that we may always interpret the planes of a conjugate binet $b$ as a bi*net $\square b$.
In analogy to the smooth theory, we show that a conjugate binet $b$ is orthogonal if and only if its corresponding bi*net $\square b$ is orthogonal.
Thus, if $b$ is principal, there exists both a Möbius lift $b_\mobq$ of $b$ and a Laguerre lift of $\square b$ which we denote by $b_\blac$.

Both Möbius geometry and Laguerre geometry are subgeometries of Lie geometry.
In the projective model of Lie geometry, oriented spheres and oriented planes of $\R^3$ are represented as points on a quadric $\lieq$ of signature $\texttt{(++++--)}$ in $\RP^5$, which is called the \emph{Lie quadric}.
The projective model of Möbius geometry and the projective model of Laguerre geometry are each included in a hyperplane of $\RP^5$.
As a result, we may embed the Möbius lift $b_\mobq$ and the Laguerre lift $b_\blac$ into $\RP^5$.
We define the \emph{Lie lift} $b_\lieq$ of a principal binet $b$ by the lines joining corresponding points of $b_\mobq$ and $b_\blac$
\[
	b_\lieq: D \rightarrow \mathrm{Lines}(\RP^5), \qquad  d \mapsto b_\mobq(d) \vee b_\blac(d).
\]
We show that $b(d)$ is in the polar space of $b(d')$ whenever $d$ and $d'$ are incident.
As a consequence we prove that the restriction of the Lie lift to the vertices of $\Z^2$ is a \emph{discrete line congruence}, which means that adjacent lines intersect.
The same holds for the restriction to the faces $F(\Z^2)$. We say that the Lie lift is a \emph{polar line bicongruence} (with respect to $\lieq$).
We show Lie lifts of principal binets are in bijection with polar line bicongruences in the projective model of Lie geometry.
As a result, in analogy to the smooth theory, we are able to apply Lie transformations to principal binets such that the result is again a principle binet.

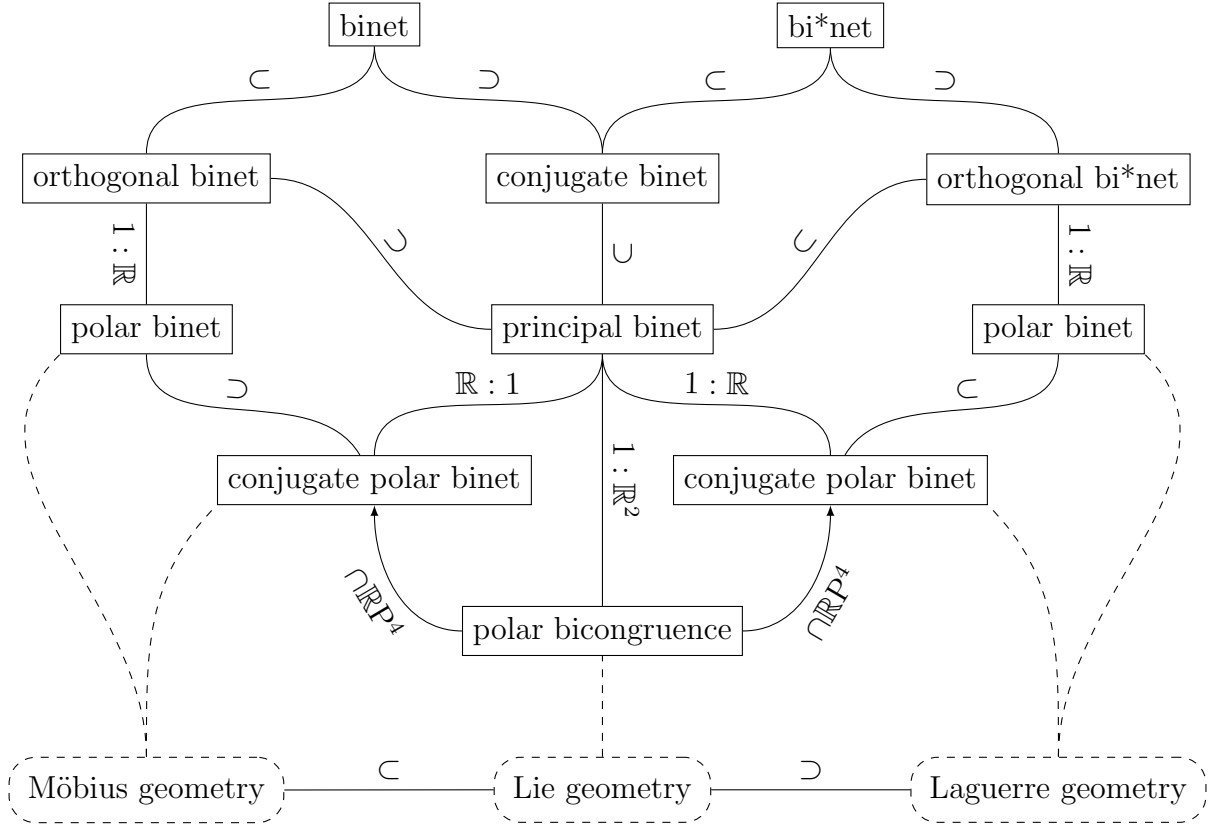
\begin{figure}[tb]
	\begin{tikzpicture}[every node/.style={anchor=north,draw, rectangle, align=center}]
		\node (bi) at (3,8) {binet};
		\node (bis) at (9,8) {bi*net};		
		\node (obi) at (0,6) {orthogonal binet};		
		\node (cbi) at (6,6) {conjugate binet};		
		\node (obis) at (12,6) {orthogonal bi*net};		
		\node (mpbi) at (0,4) {polar binet};		
		\node (lpbi) at (12,4) {polar binet};		
		\node (mcpbi) at (3,2) {conjugate polar binet};		
		\node (lcpbi) at (9,2) {conjugate polar binet};		
		\node (liebi) at (6,0) {polar bicongruence};			
		\node (pribi) at (6,4) {principal binet};	
		
		\node[rectangle,dashed,rounded corners=10,inner sep=7pt] (mobg) at (0,-2) {Möbius geometry};	
		\node[rectangle,dashed,rounded corners=10,inner sep=7pt] (lagg) at (12,-2) {Laguerre geometry};	
		\node[rectangle,dashed,rounded corners=10,inner sep=7pt] (lieg) at (6,-2) {Lie geometry};	
		
		\draw[-, sloped,out=270,in=90]
			(bi) edge[] node[above,draw=none] {$\subset$} (obi) edge[]  node[above,draw=none] {$\supset$} (cbi)
			(bis) edge[]  node[above,draw=none] {$\supset$} (obis) edge[]  node[above,draw=none] {$\subset$} (cbi)
			(obi) edge[]  node[below,draw=none] {$1:\R$} (mpbi) edge[out=0,in=180] node[above,draw=none] {$\supset$} (pribi)
			(obis) edge[]  node[above,draw=none] {$1:\R$} (lpbi) edge[out=180,in=0] node[above,draw=none] {$\subset$} (pribi)
			(cbi) edge[]  node[above,draw=none] {$\supset$} (pribi)
			(mpbi) edge[in=120] node[above,draw=none] {$\supset$} (mcpbi)
			(lpbi) edge[in=60] node[above,draw=none] {$\subset$} (lcpbi)
			(pribi) edge[]  node[above,draw=none] {$\R:1$} (mcpbi) edge[]  node[above,draw=none] {$1:\R$} (lcpbi) edge[] node[above,draw=none] {$1:\R^2$} (liebi)
			(mcpbi) edge[in=180,latex-] node[below,draw=none] {$\cap \RP^4$} (liebi)
			(lcpbi) edge[in=0,latex-] node[below,draw=none] {$\cap \RP^4$} (liebi)
			
			(mobg) edge[out=0,in=180] node[above,draw=none] {$\subset$} (lieg) edge[out=90,in=225, dashed] (mpbi.south west) edge[out=90,in=225, dashed] (mcpbi.south west)
			(lagg) edge[out=180,in=0] node[above,draw=none] {$\supset$} (lieg) edge[out=90,in=315, dashed] (lpbi.south east) edge[out=90,in=315, dashed] (lcpbi.south east)
			(lieg) edge[out=90,in=270,dashed] (liebi)
		;
		\draw[-]
		;
	\end{tikzpicture}
	
	\caption{Classes of binets and their relations.}
\end{figure}

Both circular nets and conical nets are discretizations of curvature line parametrizations, hence they should be Lie invariant.
However, circular nets are only Möbius invariant, while conical nets are only Laguerre invariant.
It was shown that this can be overcome by considering \emph{principal contact element nets},
which are generated from pairs of conical and circular nets \cite{bsorganizing, pwconical}, which we call \emph{circular-conical binets}.
The Lie lifts for principal contact element nets introduced in \cite{bsorganizing} are in bijection with \emph{discrete isotropic line congruences} in the Lie quadric, which are discrete line congruences such that the lines are contained in the Lie quadric.
We show that a circular-conical binet $b$ is a principal binet, and that in this case the restriction of our Lie lift $b_\lieq$ to the vertices
is a discrete isotropic line congruence, and thus recovers the previously known Lie lift.
In fact, whenever the restriction of a polar line bicongruence to the vertices is isotropic, this is the Lie lift of some circular-conical binet.
Therefore, our results show that it is not necessary to constrain the Lie lift to isotropic line congruences
in order to obtain a discretization of curvature line parametrizations.
Moreover, our description in Möbius geometry and in Laguerre geometry generalizes to discretizations of orthogonal and Gauß-orthogonal parametrizations.

Although Möbius lift, normal binet and Laguerre lift can be defined in an abstract manner, in practical terms it helps to have a coordinate description as well.
We use the standard coordinate framework of \cite{ddgbook} for the projective models of Möbius geometry (see Section~\ref{sec:moebius})
and Laguerre geometry (see Section~\ref{sec:lie}).
With these conventions, the Möbius lift $b_\mobq$ of an orthogonal binet $b$ is given by
\[
    b_\mobq(d) = [b(d) + e_0 + 2\rho(d) e_\infty ],
\]
where $\rho: D \rightarrow \R$ is a function that satisfies
\[
    \sca{b(d), b(d')} = \rho(d) + \rho(d'),
\]
for all incident $d, d' \in D$. The normal binet $n$ of an orthogonal bi*net $b$ is given by
\[
	n(d) = \sigma^{-1}(d)u(d),
\]
where $u(d)$ is a unit-normal vector to $b(d)$ and $\sigma: D \rightarrow \R$ is a function that satisfies
\[
    \sca{u(d), u(d')} = \sigma(d) \sigma(d'),
\]
for all incident $d,d' \in D$. Finally, the Laguerre lift $b_\blac$ of an orthogonal bi*net $b$ is given by
\[
	b_\blac(d) = [u(d), \sigma(d), h(d)] = [n(d), 1, \tfrac{h(d)}{\sigma(d)}],
\]
where $h: D \rightarrow \R$ is defined by
\[
  b(d) = \set{x \in \eucl}{\sca{u(d), x} + h(d) = 0}.
\]
Embedding Möbius geometry and Laguerre geometry into the projective model of Lie geometry (see Section~\ref{sec:lie})
the Lie lift $b_\lieq$ of a principal binet $b$ is given by the join of the Möbius lift and the Laguerre lift
\[
  b_\lieq(d) = b_\mobq(d) \vee b_\blac(d).
\]

In their seminal work \cite{bsorganizing}, Bobenko and Suris formulated the \emph{consistency principle}, stating that discretizations should have analogous integrability properties as their smooth counterparts.
In the case of curvature line parametrizations, consistency is expressed by the fact that such parametrizations do exist on domains in $\R^3$ and higher dimensions.
As a discrete analogue, we define binets and in particular principal binets on the vertices and faces of $\Z^N$ with $N \geq 3$,
We show that principal binets are a \emph{consistent reduction} of binets, meaning that if we require conjugacy and orthogonality on 2-dimensional initial data, this can be extended to a unique principal binet on all of $\Z^N$. In terms of the Lie lift, this is equivalent to showing that polar line bicomplexes are a consistent reduction of discrete line complexes, which are the $\Z^N$ analogue of discrete line congruences \cite{bslinecomplexes}.

During our investigations, Felix Dellinger has independently introduced the concept of \emph{checkerboard patterns} \cite{dellingercbpatterns}, which is closely related to binets. In particular, our definition of conjugate, orthogonal and principal binets are in bijection with the \emph{control nets} of conjugate, orthogonal and principal checkerboard patterns. Note that Dellinger has, in his language, also explored the projective invariance of conjugate checkerboard patterns and the Möbius invariance of orthogonal checkerboard patterns. Compared to Dellinger, we add the description of orthogonal bi*nets and their Laguerre invariance and the Lie lift and Lie invariance of principal binets. The results on consistency are also new. On the other hand, Dellinger has introduced the shape operator and fundamental forms for checkerboard patterns, something we would like to build upon in the future. Dellinger has also investigated K{\oe}nigs, isothermic and minimal checkerboard patterns.
\emph{K{\oe}nigs binets} will be part of a joint future publication \cite{adtkoenigsbinets}.
Note that the checkerboard approach is reminiscent of \emph{linear discrete complex analysis} \cite{BGcomplexanalysis} and more generally of \emph{Tutte embeddings} \cite{tutteembedding}.

\subsection{Future work}

In this paper we have shown how the binet approach unifies and generalizes discretizations of curvature line parametrizations. As mentioned, we are currently working on the topic of K{\oe}nigs binets, where we will show that K{\oe}nigs binets unify and generalize known discretizations of K{\oe}nigs nets.
This in turn may allow us to extend the results to isothermic and in particular constant mean curvature and minimal surfaces.

Another question that we have not investigated is whether there is some relation between binets and edge-constraint nets \cite{hsfwedgeconstraint} or the nets defined in \cite{hkyconstraint}, which in some sense could be seen as a dual notion of edge-constraint nets.

It is also known that the discretization of asymptotic parametrizations called \emph{A-nets} \cite{saueranet} corresponds to discrete isotropic line congruences in the Plücker quadric \cite{doliwaanetspluecker}. Thus, it would be interesting to investigate if there is a generalization of discrete A-nets to polar line bicongruences with respect to the Plücker quadric in analogy to the Lie lift of principal binets. 

Finally, besides isotropic line congruences, there are other interesting special cases of line congruences: linear line congruences \cite{bslinecomplexes}, tangential line congruences, Doliwa line congruences \cite{affolterthesis}. Do some of them correspond to meaningful special cases of principal binets?

\subsection{Structure of the paper}

We begin by introducing general binets in Section~\ref{sec:binets} and then proceed to introduce conjugate, polar and orthogonal binets in Section~\ref{sec:conjugatebinets}, \ref{sec:polarbinets} and \ref{sec:orthobinets}. We then recall the basics of Möbius geometry in Section~\ref{sec:moebius} in order to explain the Möbius lift of orthogonal binets in Section~\ref{sec:mobiuslift}. Subsequently, we present the dual story, that is we introduce bi*nets, conjugate bi*nets, and orthogonal bi*nets in Sections~\ref{sec:bistarnets}, \ref{sec:conjugate-bi-star-nets} and \ref{sec:orthobinstarnets} as well as normal binets in Section~\ref{sec:normalbinets}. Laguerre geometry is introduced in Section~\ref{sec:laguerre} and used in Section~\ref{sec:laguerrelift} for the Laguerre lift of orthogonal bi*nets. We introduce line bicongruences in Section~\ref{sec:linebi}, which prepares us for the synthesis of binets (the primal side) and bi*nets (the dual side): principal binets in Section~\ref{sec:principal}. On the geometric side, the synthesis of Möbius geometry and Laguerre geometry is Lie geometry as explained in Section~\ref{sec:lie}, which we use to introduce the Lie lift in Section~\ref{sec:lielift}. We also discuss the occurrence of principal curvature spheres in Section~\ref{sec:curvaturespheres}, and how our previous results specialize to circular and conical nets in Section~\ref{sec:circularconical}. In Section~\ref{sec:spaces} we briefly discuss the spaces of the various classes of binets. Finally, we discuss multi-dimensional consistency of principal binets in Section~\ref{sec:consistency}.

In spirit of the Klein-Erlangen program, our discussion of the results is based on a coordinate-free approach using projective models for the various geometries.
However, for practical purposes coordinate descriptions are essential, hence we have added coordinate boxes throughout the paper that are optional for the reader.

We give a one-page overview (cheat sheet) for our notation of the various spaces and maps in Appendix~\ref{sec:cheatsheet}.

\subsection*{Acknowledgements}
N.~C.~Affolter and J.~Techter were supported by the Deutsche Forschungsgemeinschaft (DFG) Collaborative Research Center TRR 109 ``Discretization in Geometry and Dynamics''. We would like to thank Alexander Bobenko, Felix Dellinger, Alexander Fairley, Christian Müller, Wolfgang Schief, Nina Smeenk and Boris Springborn for discussions. We would also like to thank Johannes Wallner for organizing the Fall School ``Discrete Geometry and Topology'' in Graz in 2016, where the authors first discussed this research project together.

Most of the images were created with \texttt{Krita}, \texttt{GeoGebra}, or  \texttt{pyddg} \cite{pyddg}.
We would like to thank
Shana Choukri for Figures~\ref{fig:orthogonal-sphere-represantation}, \ref{fig:cross-spheres}, \ref{fig:principal-binet-circles},
Oliver Gross for Figure~\ref{fig:focal-binet},
Stefan Sechelmann for Figure~\ref{fig:examples} (left and middle),
and Nina Smeenk for Figure~\ref{fig:examples} (right).

\section{Nets and binets} \label{sec:binets}

Traditionally, (discrete) nets are maps $\Z^2 \rightarrow \RP^n$,
which are viewed as discretizations of smooth parametrized surfaces \cite{ddgbook}.
Here $\RP^n$ denotes the $n$-dimensional real projective space.
We introduce binets as pairs of maps on the vertices \emph{and} the faces of $\Z^2$ (see Figure~\ref{fig:binet}).

\begin{figure}[H]
  \centering
  \begin{overpic}[width=0.6\textwidth]{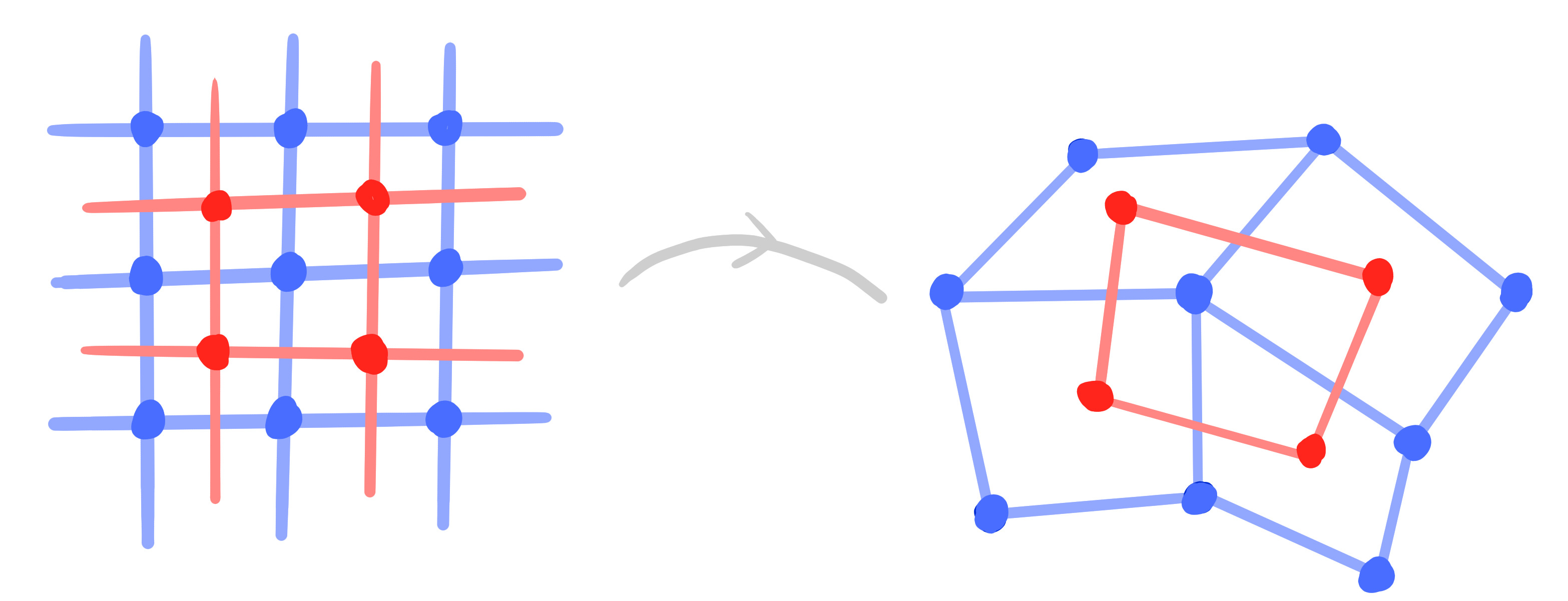}
    \put(3,5){$D$}
    \put(90,30){$\RP^n$}
  \end{overpic}
  \caption{A binet is a map $D = V \cup F \rightarrow \RP^n$.}
  \label{fig:binet}
\end{figure}

Consider $\Z^2$ as an infinite planar graph with
\begin{align*}
  \text{vertices} ~ &V \coloneqq V(\Z^2) = \Z^2,\\
  \text{edges} ~ &E \coloneqq E(\Z^2) \simeq \Z^2 \cup \Z^2,\\
  \text{faces} ~ &F \coloneqq F(\Z^2) \eqqcolon (\Z^2)^* \simeq \Z^2.
\end{align*}
We say two vertices $v, v' \in V$ (resp.\ two faces) are \emph{adjacent} if they share an edge, i.e.,
$
  (v, v') \in E.
$
We say that a vertex $v \in V$ and a face $f \in F$ are \emph{incident} if $v$ is a boundary vertex of $f$, and write
\[
  v \inc f.
\]
We additionally denote the vertices of the double graph by
\[
  D \coloneqq V \cup F.
\]
For $d,d' \in D$ the incidence $d \inc d'$ implies $d'\in V, d\in F$ or $d\in V, d'\in F$.
We also consider the incident pairs of adjacent vertices and faces, and denote them by
\[
  C \coloneqq \set{(v,f,v',f')}{v, v' \in V, \quad f, f' \in F, \quad v, v' \inc f, f'},
\]
which we call the \emph{crosses} of $D$ (see Figure~\ref{fig:cross}).
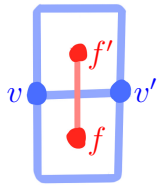
\begin{figure}[H]
  \centering
  \begin{overpic}[width=0.12\textwidth]{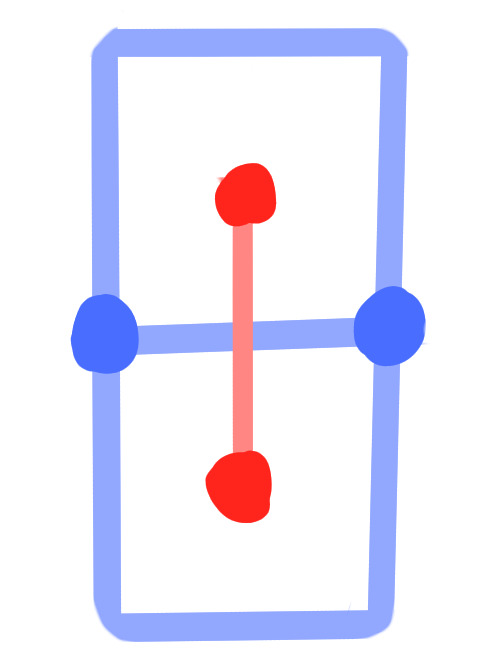}
    \put(0,45){$\color{blue}v$}
    \put(66,45){$\color{blue}v'$}
    \put(42,20){$\color{red}f$}
    \put(42,65){$\color{red}f'$}
  \end{overpic}
  \caption{A cross $(v,f,v',f') \in C$.}
  \label{fig:cross}
\end{figure}
\begin{remark}
  \label{rem:dual-graph}
  The faces and edges of $\Z^2$ can be identified with the vertices and edges of the dual graph $(\Z^2)^*$, respectively.
  In this way $D$ may be viewed as the join of $\Z^2$ and $(\Z^2)^*$,
  and $C$ as pairs of an edge of $\Z^2$ and its corresponding dual edge of $(\Z^2)^*$.
\end{remark}

We use the following definition of nets for our purposes. 
It includes a definition of regularity, which is a genericity assumption on triples of vertices,
which in the smooth case corresponds to the regularity of a parametrized surface. Throughout the article, we include definitions of regularity that allow for precise statements, and that reduce the number of special cases that need to be treated. In order to keep the text less technical, we will not differentiate between regular maps and non-regular nets in the text outside of definitions, lemmas and theorems.
\begin{definition}[Nets]\label{def:net}
  A \emph{net} is a map $p: V \rightarrow \RP^n$. A \emph{regular net} is a net 
  such that the points of any three vertices of each face
  span a plane.  
\end{definition}
The identification of the faces $F$ with the vertices of the dual graph (see Remark~\ref{rem:dual-graph})
yields an analogous definition for nets as maps $p : F \rightarrow \RP^n$.
We define binets as maps on $D$ that come from pairs of nets on $V$ and $F$.
\begin{definition}[Binets]\label{def:binet}
  A \emph{binet} is a map $b: D \rightarrow \RP^n$, such that the restrictions to $V$ and $F$ are nets. A \emph{regular binet} is a binet, such that 
  \begin{enumerate}
  \item
    the restrictions to $V$ and $F$ are regular nets,
  \item
    and for all incident $d,d'\in D$ the points $b(d),b(d')$ are distinct.\qedhere
  \end{enumerate}
\end{definition}

\section{Conjugate binets} \label{sec:conjugatebinets}

A conjugate net is also known as Q-net or quadrilateral lattice \cite{sauerqnet, ddgbook, dsqnet}.

\begin{definition}[Conjugate nets]
  A \emph{conjugate net} is a net $p: V \rightarrow \RP^n$, such that the image of each quad is contained in a plane. 
\end{definition}

By the identification of the faces $F$ with the vertices of the dual graph (see Remark~\ref{rem:dual-graph}), conjugate nets are defined on $F$ analogously.

\begin{definition}[Conjugate binets] \label{def:conjugatebinet}
  A binet $b$ such that the restrictions to $V$ and $F$ are conjugate nets is a \emph{conjugate binet}.
\end{definition}

A conjugate binet is essentially a pair of conjugate nets.

\subsection*{Invariance}
The definition of conjugate binets is invariant under projective transformations.

\section{Polar binets} \label{sec:polarbinets}

In the following, binets that consist of two nets which are related by polarity with respect to a quadric will play an important role (see Figure~\ref{fig:polar-binet}).
In particular, the lifts of orthogonal binets to Möbius geometry
and of orthogonal bi*nets to Laguerre geometry will be polar binets.

Recall that a quadric $\mathcal Q \subset \RP^n$ is given by a symmetric bilinear form $\sca{\cdot, \cdot}_{\mathcal{Q}}$ on $\R^{n+1}$
\[
  \mathcal{Q} = \set{[x] \in \RP^n}{\sca{x,x}_{\mathcal{Q}} = 0}.
\]
For a projective subspace $K \subset \RP^n$ its polar subspace is given by
\[
  K^\pol = \set{X \in \RP^n}{X \pol Y ~\text{for all}~ Y \in K},
\]
where we denote the polarity relation of two points $X = [x], Y = [y] \in \RP^n$ by
\[
  X \pol Y
  \quad\Leftrightarrow\quad
  \sca{x,y}_{\mathcal{Q}} = 0.
\]
\begin{definition}[Polar binets]
  Let $\mathcal Q \subset \RP^n$ be a quadric.
  A \emph{polar binet}  is a binet $b: D \rightarrow \RP^n$, such that
  \[
    b(d) \pol b(d') \qquad   \text{for all incident}~ d,d'\in D.\qedhere
  \]
\end{definition}
\begin{figure}[H]
  \centering
  \begin{overpic}[width=0.6\textwidth]{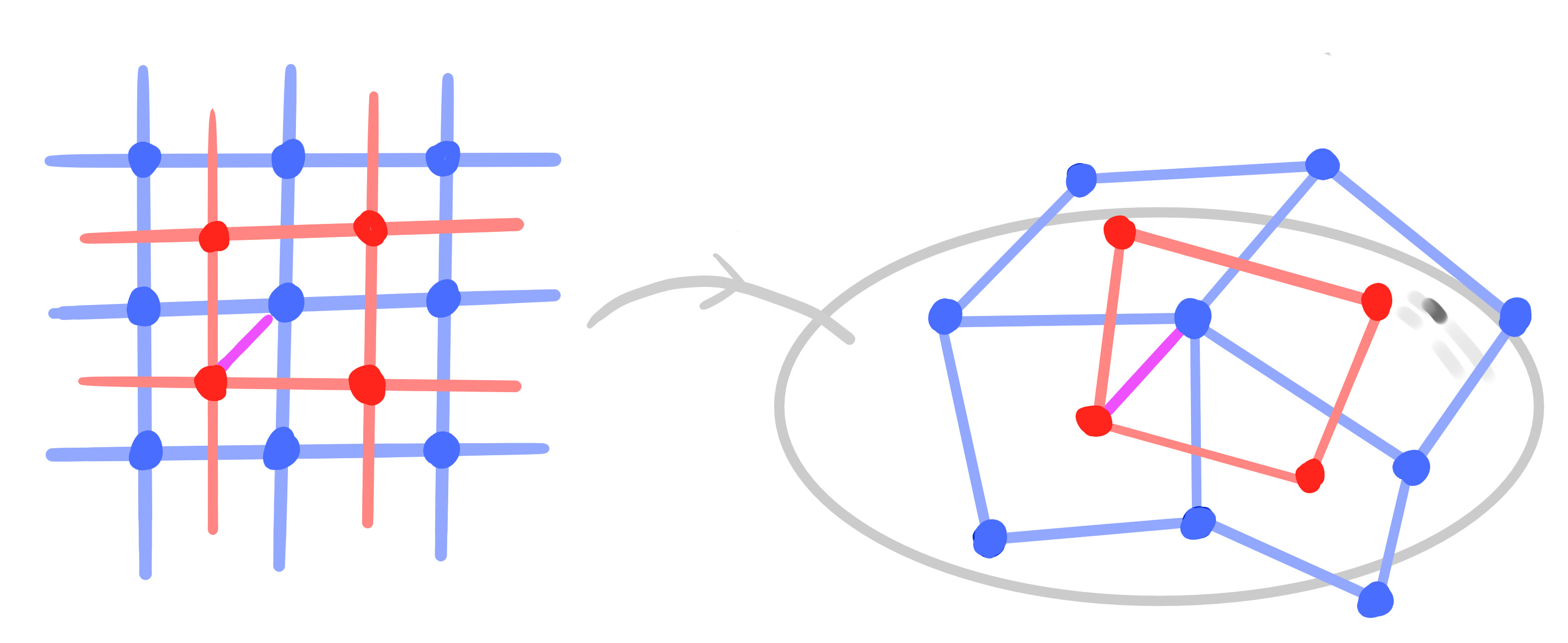}
    \put(19.5,23){$\color{blue}v$}
    \put(10,13){$\color{red}f$}
    \put(78,20){$\color{blue}b(v)$}
    \put(64,9.5){$\color{red}b(f)$}
    \put(95,5){$\mathcal{Q} \subset \RP^n$}
  \end{overpic}
  \caption{
    A polar binet $b : D \rightarrow \RP^n$ with respect to the quadric $\mathcal{Q}$.
    Two points $b(v)$ and $b(f)$ are polar for incident $v \in V$ and $f \in F$.
  }
  \label{fig:polar-binet}
\end{figure}
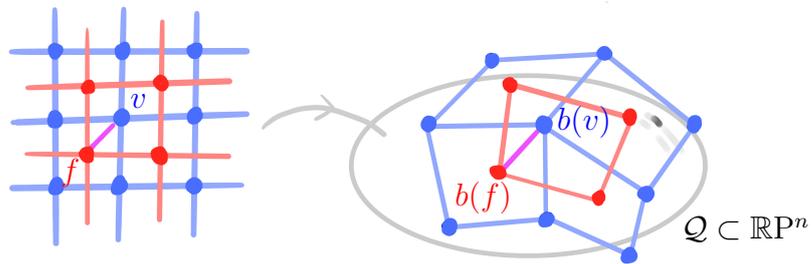

Equivalently polar binets can be characterized by the condition
that the two lines of an edge and its corresponding dual edge are related by polarity, i.e.,
\[
  (b(v) \vee b(v')) \pol (b(f) \vee b(f')) \qquad \text{for all crosses} ~ (v,f,v',f') \in C,
\]
where $X \vee Y$ denotes the line spanned by $X, Y \in \RP^n$.

\subsection*{Invariance}
The definition of polar binets is invariant under projective transformations that preserve the quadric $\mathcal Q$.

\section{Orthogonal binets}
\label{sec:orthobinets}

For binets in the Euclidean space
\[
  \eucl^n \simeq \R^n,
\]
we employ the following orthogonality condition on the crosses for $\Z^2$ (see Figure~\ref{fig:orthogonal-binet}),
which was introduced for discrete orthogonal coordinate systems in \cite{bsstconfocali, bsstconfocalii}.

\begin{definition}[Orthogonal binets]
  \label{def:orthogonal-binets}
  A binet $b: D \rightarrow \eucl^n$ is an \emph{orthogonal binet} if
  \[
    b(v) \vee b(v') ~\orth~ b(f) \vee b(f') \quad
    \text{for all crosses}~ (v,f,v',f') \in C,
  \]
  where $X \vee Y$ denotes the line spanned by $X, Y \in \eucl^n$ and $\orth$ denotes the Euclidean orthogonality.
\end{definition}
\begin{figure}[H]
  \centering
  \begin{overpic}[width=0.3\textwidth]{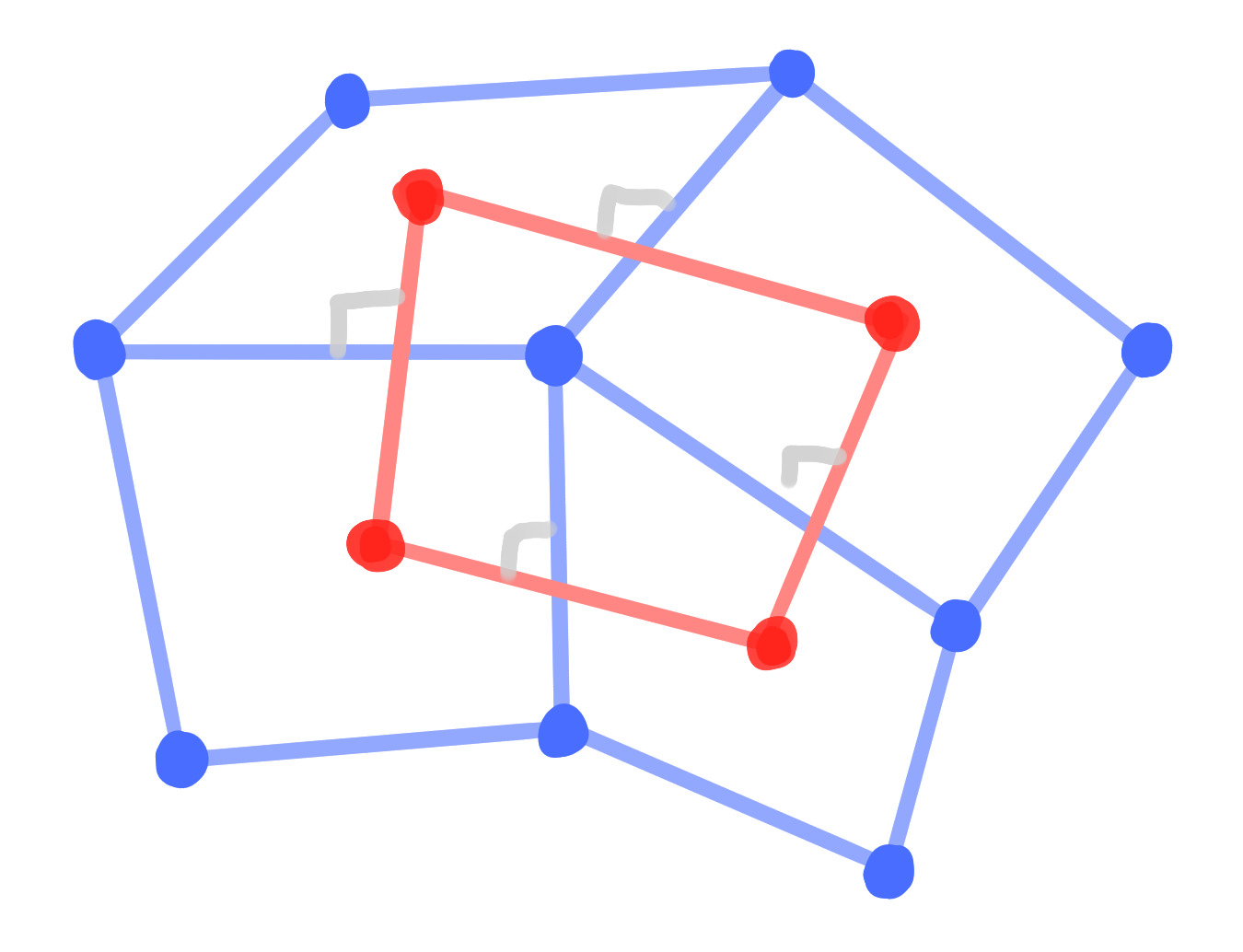}
    \put(36,8){$\color{blue}b(v)$}
    \put(48.5,47){$\color{blue}b(v')$}
    \put(13,27){$\color{red}b(f)$}
    \put(56,16){$\color{red}b(f')$}
  \end{overpic}
  \caption{
    An orthogonal binet $b : D \rightarrow \eucl^n$.
    Dual edges are orthogonal at every cross $(v,f,v',f') \in C$.
  }
  \label{fig:orthogonal-binet}
\end{figure}
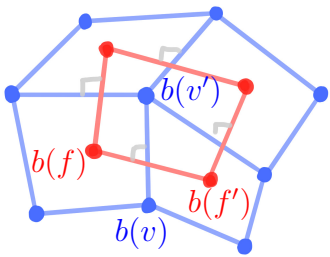
Equivalently, this orthogonality condition can be written as
\[
  \sca{b(v) - b(v'), b(f) - b(f')} = 0,
\]
where $\sca{\cdot,\cdot}$ denotes the Euclidean scalar product.
Note that the two lines $b(v) \vee b(v')$ and $b(f) \vee b(f')$ do not need to intersect.

For simplicity of the presentation and in regards to the lift to Lie geometry,
we restrict the following discussion of orthogonal binets to $n=3$.

\subsection*{Invariance}
The definition of orthogonal binets is invariant under similarity transformations.

In the smooth theory, an orthogonal parametrization of a surface is invariant under conformal transformations of the surface.
Thus, in particular, orthogonal parametrizations are invariant under conformal transformations of $\eucl^3$,
which are exactly the \emph{Möbius transformations}.
Therefore, Möbius invariance is a desired property of discretizations of orthogonal parametrizations.
Indeed, it is an advantage of circular nets (see Definition~\ref{def:circular-net}) that they are Möbius invariant.
Orthogonal binets are invariant under similarity transformations, yet a priori not under Möbius transformations.
However, in Section \ref{sec:mobiuslift} we introduce a Möbius lift of orthogonal binets, which does give us the ability to apply Möbius transformations to orthogonal binets such that we obtain an orthogonal binet again.

\section{Möbius geometry}
\label{sec:moebius}

First, recall the projective model of Euclidean geometry \cite{kleincomparative,kleinvorlesungen}.
Let $E^\infty \subset \RP^3$ be some plane, which we interpret as the \emph{plane at infinity}.
We identify
\[
  \eucl \coloneqq \eucl^3 \simeq \RP^3 \setminus E^\infty
\]
with the 3-dimensional \emph{Euclidean space}.
Consider a symmetric bilinear form $\sca{\cdot, \cdot}$ of signature $\texttt{(+++)}$ on $E^\infty$.
It defines the absolute conic $\euclq$ of Euclidean geometry in $E^\infty$, which has no real points.
Thus, the complexification of the absolute conic
in $\CP^3 \supset \RP^3$ should be considered in its place.
\begin{coordinates}
  \label{coords:euclidean}
  We choose the plane at infinity as
  \[
    E^\infty = \set{[x] \in \RP^3}{x_4 = 0},
  \]
  and obtain the Euclidean space
  \[
    \eucl = \set{[x] \in \RP^3}{x_4 \neq 0} \cong \R^3.
  \]
  We represent the absolute conic $\euclq$ by 
  \[
    \sca{x,x} = x_1^2 + x_2^2 + x_3^2,\qquad
    x \in \R^4, ~ x_4 = 0.
  \]
\end{coordinates}

Secondly, recall the projective model of Möbius geometry \cite{blaschkevl, hjmoebius}.
Let $\sca{\cdot, \cdot}_{\mobq}$ be a symmetric bilinear form of signature $\texttt{(++++-)}$, and
\[
  \mobq = \set{[x] \in \RP^4}{ \sca{x,x}_{\mobq} = 0}
\]
the corresponding quadric in $\RP^4$, which we call the \emph{Möbius quadric}.
\begin{coordinates}
  In homogeneous coordinates of $\RP^4$ we choose
  \[
    \sca{x,x}_{\mobq} = x_1^2 + x_2^2 + x_3^2 + x_4^2 - x_5^2.
  \]
  And thus, in affine coordinates $x_5 = 1$, the Möbius quadric is the unit sphere $\S^3$:
  \[
    x_1^2 + x_2^2 + x_3^2 + x_4^2 = 1.
  \]
\end{coordinates}

We now embed the 3-dimensional Euclidean space $\eucl$ into $\RP^4$ in the following way.
Let $\PB \in \mobq$ be a point on the Möbius quadric,
$S_\eucl \subset \RP^4$ a hyperplane that does not contain $\PB$,
and define $E^\infty \coloneqq S_{\eucl} \cap \PB^\pol$
as the intersection of $S_\eucl$ with the tangent hyperplane at $\PB$.
We identify
\[
  \eucl \simeq S_\eucl \setminus E^\infty,
\]
with the 3-dimensional Euclidean space.
The restriction of $\sca{\cdot,\cdot}_\mobq$ to $E^\infty$ has signature $\texttt{(+++)}$
and defines the absolute conic $\euclq$ of Euclidean geometry.

Consider the central projection to $S_\eucl$ with center $\PB$:
\[
  \pi_{\eucl} : \RP^4 \setminus \{ \PB \} \rightarrow S_\eucl,\quad
  X \mapsto (X \vee \PB) \cap S_\eucl.
\]
Its restriction to $\mobq \setminus \{ \PB \}$ is a bijection to $\eucl$,
which is usually called the \emph{stereographic projection}. The \emph{inverse stereographic projection is}
\begin{align}
	\pi_{\eucl}^{-1} : S_\eucl \rightarrow \mobq \setminus \{ \PB \},\quad
	X \mapsto (X \vee \PB) \cap (\mobq \setminus \{ \PB \}), \nonumber
\end{align}
which satisfies $\pi_{\eucl}^{-1} \circ \pi_{\eucl} = \id$ restricted to $\mobq \setminus \{ \PB \}$ and $\pi_{\eucl} \circ \pi_{\eucl}^{-1} = \id$ restricted to $\eucl$. 
\begin{coordinates}
  We choose
  \[
    \PB = [0,0,0,1,1], \qquad
    S_\eucl = \set{[x] \in \RP^4}{ x_4 = 0}.
  \]
  Then
  \[
    E^\infty = \set{[x] \in \RP^4}{ x_4 = x_5 = 0}, \qquad
    \eucl = \set{[x] \in \RP^4}{ x_4 = 0,\, x_5 \neq 0} \cong \R^3,
  \]
  and
  \[
    \sca{x,x} = x_1^2 + x_2^2 + x_3^2, \qquad x \in \R^5, ~ x_4 = x_5 = 0.
  \]
  Additionally, the stereographic projection is given by
  \[
    \pi_{\eucl}([x]) = [x_1, x_2, x_3, 0, x_5 - x_4].
  \]
\end{coordinates}

\subsection*{Spheres in Möbius geometry}
The \emph{outside} of the Möbius quadric is given by
\[
  \mobq^+ = \set{[x] \in \RP^4}{ \sca{x,x}_{\mobq} > 0}.
\]
Let $X \in \mobq^+$ be a point outside the Möbius quadric.
Then its polar hyperplane $X^\pol$ intersects $\mobq$ in a 2-dimensional quadric of signature $\texttt{(+++-)}$.
Under the stereographic projection this section becomes a Euclidean sphere
$
  \pi_\eucl( X^\pol \cap \mobq ) \subset \eucl
$
with center ${\pi_\eucl(X) \in \eucl}$ if $X \notin \PB^\pol$, or, a Euclidean plane if $X \in \PB^\pol$ (see Figure~\ref{fig:moebius-geometry}).

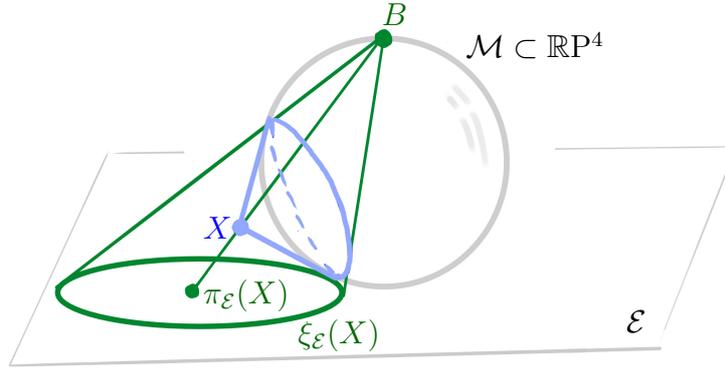
\begin{figure}[H]
  \centering
  \begin{overpic}[width=0.6\textwidth]{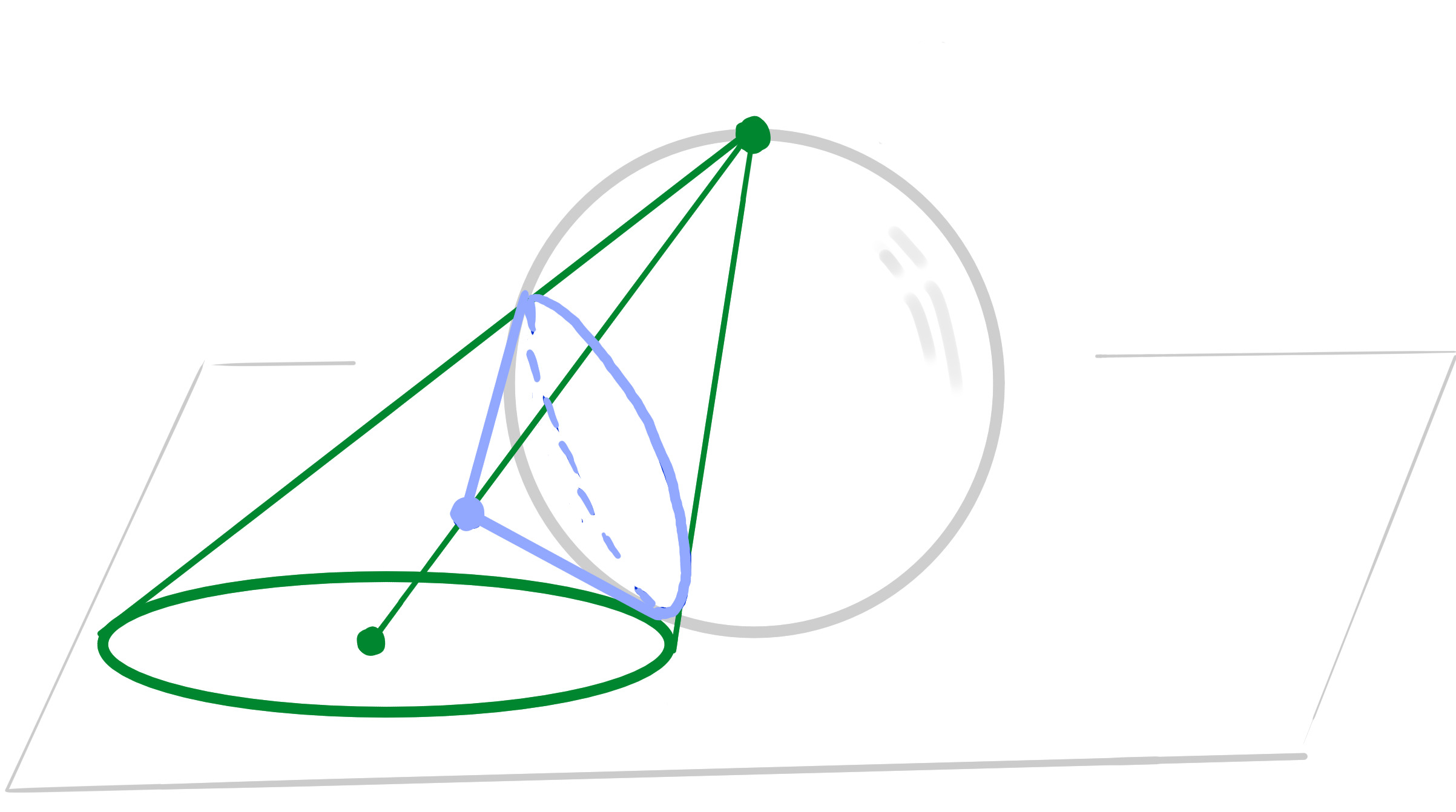}
    \put(85,6){$\eucl$}
    \put(63,44){$\mobq\subset\RP^4$}
    \put(51.5,48.5){\color{darkgreen}$\PB$}
    \put(27,19){\color{blue}$X$}
    \put(27,10){\color{darkgreen}$\pi_\eucl(X)$}
    \put(40,4.5){\color{darkgreen}$\xi_\eucl(X)$}
  \end{overpic}
  \caption{
    The representation of spheres in the projective model of Möbius geometry.
    A point $X \in \mobq^+$ corresponds to a sphere $\xi_\sp(X) \subset \eucl$ with center $\pi_\eucl(X) \in \eucl$ by stereographic projection.
  }
  \label{fig:moebius-geometry}
\end{figure}

Let $\sp$ denote the set of \emph{generalized Euclidean spheres}, that is the set of Euclidean spheres and planes
\[
  \mathcal{S} \coloneqq \text{Spheres}(\eucl).
\]
Then we obtain a map
\[
  \xi_\sp : \mobq^+ \rightarrow \sp, \quad
  X \mapsto \pi_\eucl( X^\pol \cap \mobq ).
\]
In this representation of spheres (and planes) the orthogonality of spheres is given by polarity with respect to the Möbius quadric (see Figure~\ref{fig:sphere-polarity}).
\begin{proposition}
  \label{prop:orthogonal-spheres}
  Two points $X, X' \in \mobq^+$ are polar with respect to $\mobq$ if and only if
  the two corresponding spheres $\xi_\sp(X)$ and $\xi_\sp(X')$ are orthogonal.
\end{proposition}

\begin{figure}[H]
  \centering
  \begin{overpic}[width=0.6\textwidth]{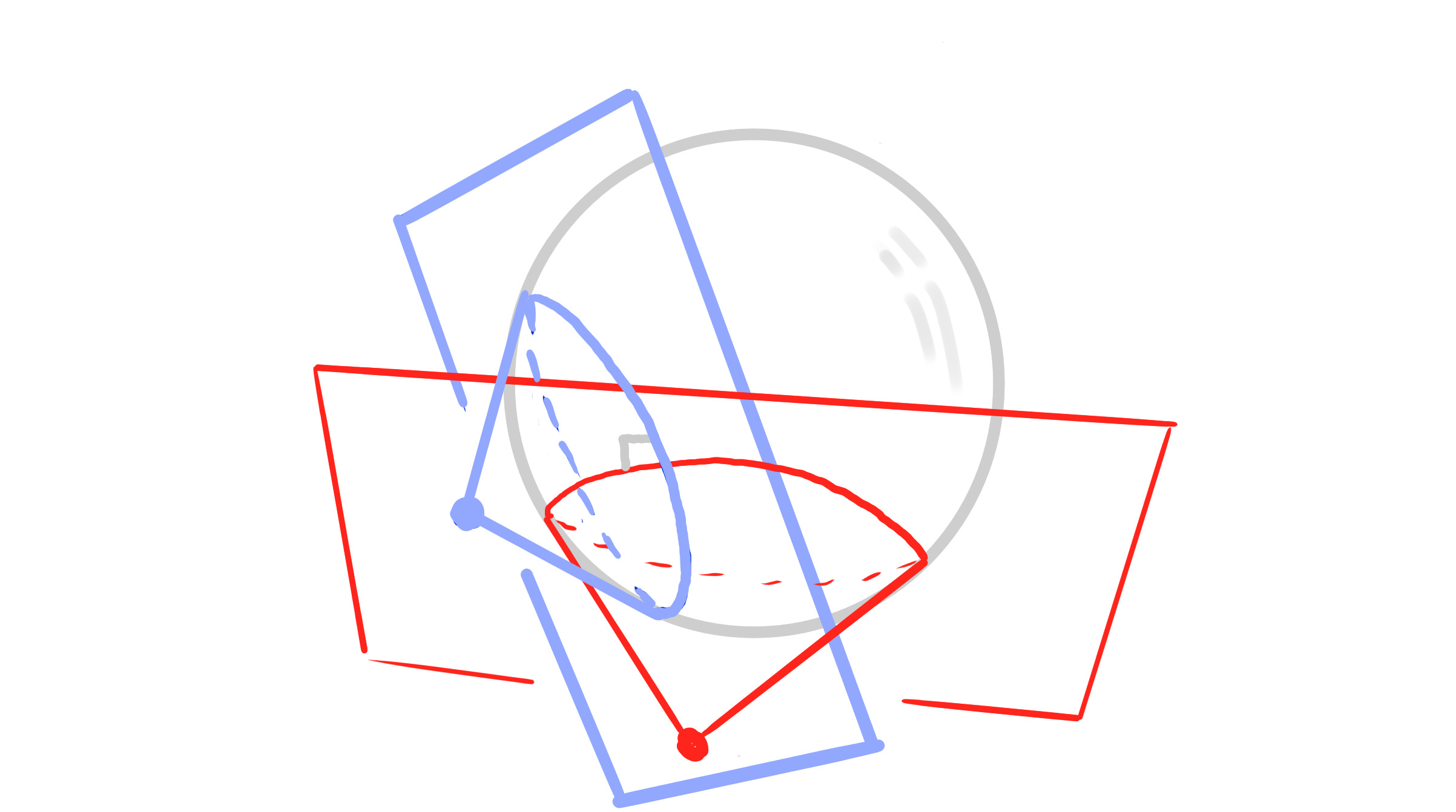}
    \put(27,19){\color{blue}$X$}
    \put(49.2,3){\color{red}$X'$}
    \put(32,46){\color{blue}$X^\pol$}
    \put(78,15){\color{red}$X'^\pol$}
    \put(63,44){$\mobq \subset \RP^4$}
  \end{overpic}
  \caption{
    Orthogonality of spheres in terms of polarity in the projective model of Möbius geometry.
    Two spheres $\xi_\sp(X)$ and $\xi_\sp(X')$ are orthogonal if and only if $X \pol X'$.
  }
  \label{fig:sphere-polarity}
\end{figure}
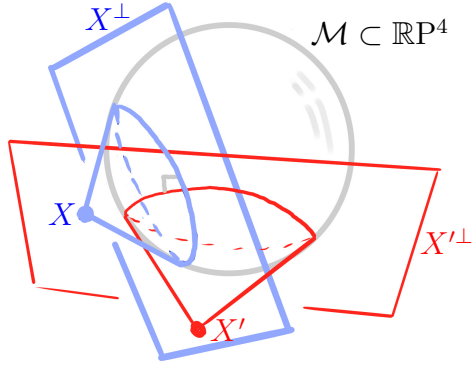

For a point $X \in \mobq^-$ inside the Möbius quadric,
the polar hyperplane does not intersect $\mobq$ in any real points.
Yet the projection $\pi_\eucl(X) \in \eucl$ still yields a real point in Euclidean space.
The point $X$, or its polar hyperplane, can be interpreted as an \emph{imaginary sphere} with real center and imaginary radius \cite{pssdarbouxwebs}.
In particular, for $X \in \mobq^-$ we choose for $\xi_\sp$ the \emph{real representative}
\begin{align}
	\xi_\sp(X) \coloneqq ((X \vee \PB) \cap X^\perp)^\perp.
\end{align}
We think of $\xi_\sp(X)$ as having \emph{imaginary radius}, see also Coordinates~\ref{coo:mobiusspheres}.
Note that the center of the real representative $\xi_\sp(X)$ is still $\pi_\eucl(X)$.

If we extend the set of spheres $\sp$ by the set of points and imaginary spheres,
the map $\xi_\sp$ can be extended to the entire space
\[
  \xi_\sp : \RP^4 \rightarrow \sp.
\]
The orthogonality of real spheres and imaginary spheres is then defined in such a way
that Proposition~\ref{prop:orthogonal-spheres} generalizes in the following way.
\begin{proposition}
  \label{prop:orthogonal-spheres-ext}
  Two points $X, X' \in \RP^4$ are polar with respect to $\mobq$ if and only if
  the two corresponding (possibly imaginary) spheres $\xi_\sp(X)$ and $\xi_\sp(X')$ are orthogonal.
\end{proposition}

\begin{remark}
  \label{rem:orthogonal-imaginary-sphere}
  Let $X \in \mobq^+, X' \in \mobq^-$.
  Then $\xi_\sp(X')$ is an imaginary sphere.
  We say the (real) sphere $\xi_\sp(X)$ is orthogonal to the (real representative) sphere $\xi_\sp(X')$
  if and only if the sphere $\xi_\sp(X)$ intersects the sphere $\xi_\sp(X')$ in a great circle on $\xi_\sp(X')$.
\end{remark}

\begin{coordinates}
	\label{coo:mobiusspheres}
  Two spheres in $\R^3$ with real centers $c, c' \in \R^3$ and real or imaginary radii, that is $r^2, r'^2 \in \R$, are orthogonal
  if and only if
  \[
    \abs{c - c'}^2 = r^2 + r'^2,
  \]
  or, with $\rho = \frac{1}{2}(\abs{c}^2 - r^2)$, if and only if
  \begin{equation}
    \label{eq:orthogonal-spheres}
    \sca{c, c'} = \rho + \rho'.
  \end{equation}
  For the lift to $\RP^4$ we introduce the basis of $\R^5$
  \[
    e_1,\quad
    e_2,\quad
    e_3,\quad
    e_\infty = \tfrac{1}{2}(e_5 + e_4),\quad
    e_0 = \tfrac{1}{2}(e_5 - e_5).
  \]
  Then a sphere with center $c \in \R^3$ and real or imaginary radius, that is $r^2 \in \R$, corresponds to the point $[x] \in \RP^4 \setminus B^\pol$ with
  \begin{equation}
    \label{eq:Möbius-lift-spheres}
    x = c + e_0 + \underbrace{(\abs{c}^2 - r^2)}_{= 2\rho} e_\infty = c + e_0 + 2\rho e_\infty.
  \end{equation}
  Vice versa, from a point $[x] \in \RP^4 \setminus B^\pol$
  the center and radius of the corresponding sphere $\xi_\sp([x])$ are recovered by
  \[
    c = \pi_\eucl([x]), \qquad
    r^2 = \frac{\sca{x,x}_\mobq}{4\sca{x, e_\infty}_\mobq^2}.
  \]
  Note that
  \[
    \begin{aligned}
      r^2 > 0 \quad &\Leftrightarrow \quad [x] \in \mobq^+,\\
      r^2 < 0 \quad &\Leftrightarrow \quad [x] \in \mobq^-.
    \end{aligned}
  \]
  For two points $[x]$, and $[x']$ the two corresponding spheres $\xi_\sp([x])$ and $\xi_\sp([x'])$ are orthogonal
  if and only if $\sca{x,x'}_\mobq = 0$, which is equivalent to \eqref{eq:orthogonal-spheres}.
\end{coordinates}

\begin{remark}
  In the same way that hyperplanes in $\RP^4$ can be identified with (possibly imaginary) spheres (and planes),
  2-dimensional planes in $\RP^4$ can be identified with (possibly imaginary) \emph{circles} (and \emph{lines}).
\end{remark}

When considering the lift of orthogonal binets to Möbius geometry,
the radical plane of two spheres will be a useful concept.
\begin{definition}[Radical planes]
  The \emph{radical plane} of two spheres $S, S' \subset \eucl$
  is the plane containing the centers of all spheres (with real or imaginary radius) orthogonal to $S$ and $S'$.
\end{definition}
In particular, we will make use of the following properties of radical planes.
\begin{proposition}
  \label{prop:radical-plane}
  Let $X, X' \in \RP^4$ be two points.
  Then the radical plane of the two spheres $\xi_\sp(X)$ and $\xi_\sp(X')$ is given by
  \[
    \pi_\eucl(X^\pol \cap X'^\pol).
  \]
  Furthermore, the radical plane is orthogonal to the line connecting the centers of the two spheres, that is
  \[
    \pi_\eucl(X) \vee \pi_\eucl(X') \quad\orth\quad \pi_\eucl(X^\pol \cap X'^\pol).\qedhere
  \]
\end{proposition}

\section{Möbius lift of orthogonal binets}
\label{sec:mobiuslift}

We use the projective model of Möbius geometry $\mobq \subset \RP^4$,
and embed the Euclidean space $\eucl \subset S_\eucl \subset \RP^4$ as described in Section~\ref{sec:moebius}.

We identified points in $\RP^4$ with spheres in $\eucl$ by means of the map $\xi_\sp : \RP^4 \rightarrow \sp$.
This immediately leads to a representation of a polar binet in $\RP^4$ in terms of orthogonal spheres,
which we call its \emph{orthogonal sphere representation} (see Figure~\ref{fig:orthogonal-sphere-represantation}).
\begin{figure}[H]
  \centering
  \includegraphics[width=0.32\textwidth, trim={800 0 300 0},clip]{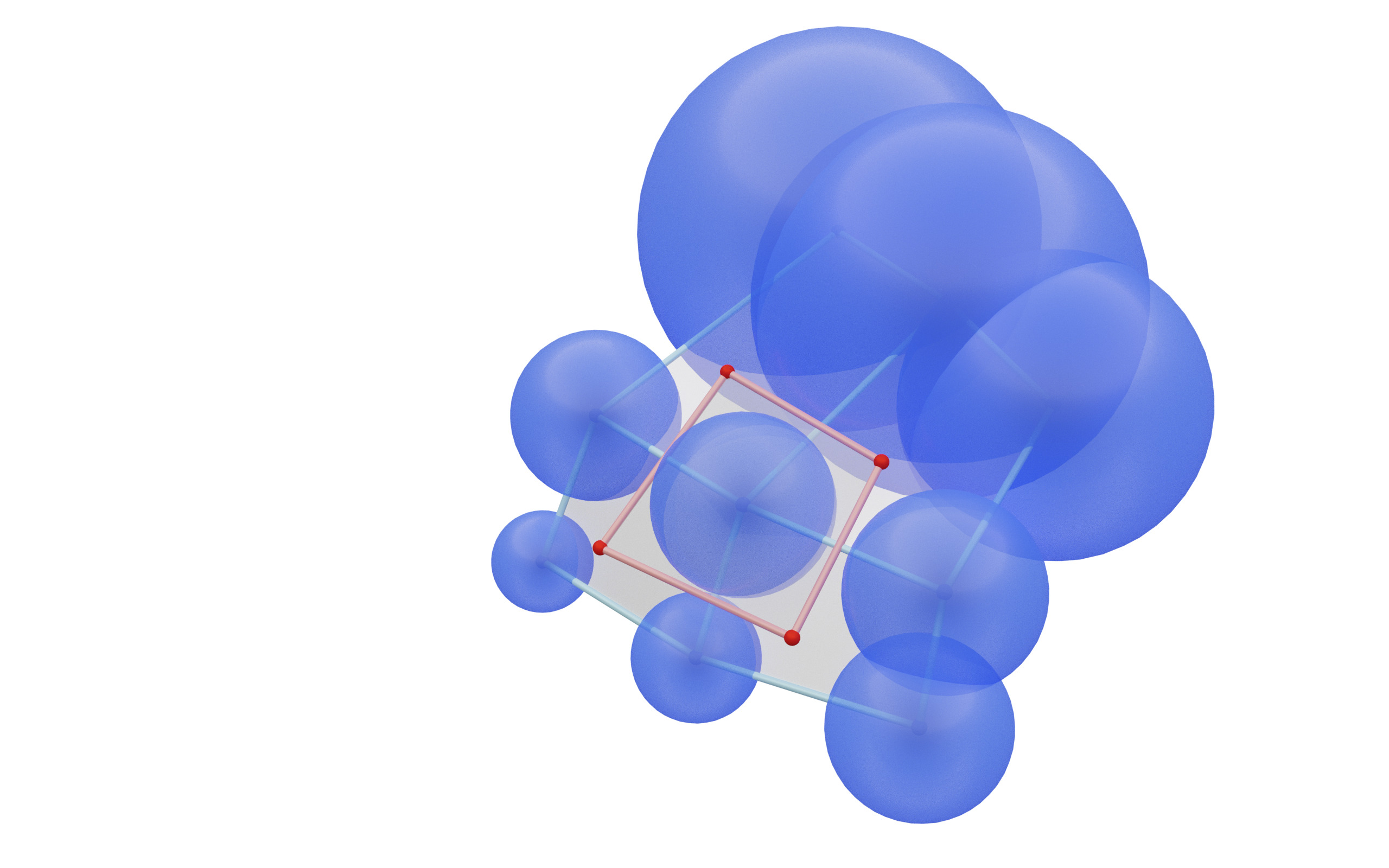}
  \includegraphics[width=0.32\textwidth, trim={800 0 300 0},clip]{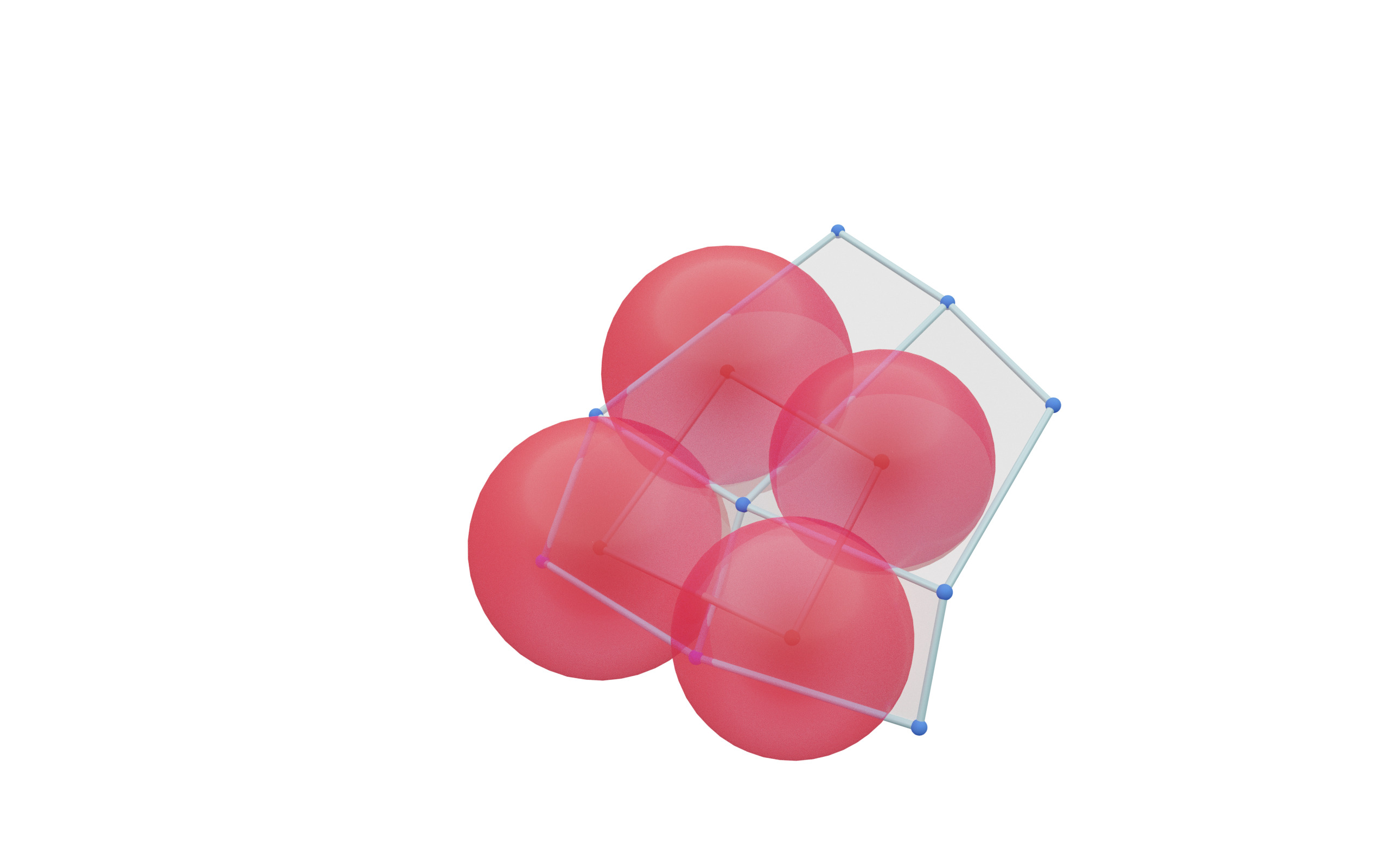}
  \includegraphics[width=0.32\textwidth, trim={800 0 300 0},clip]{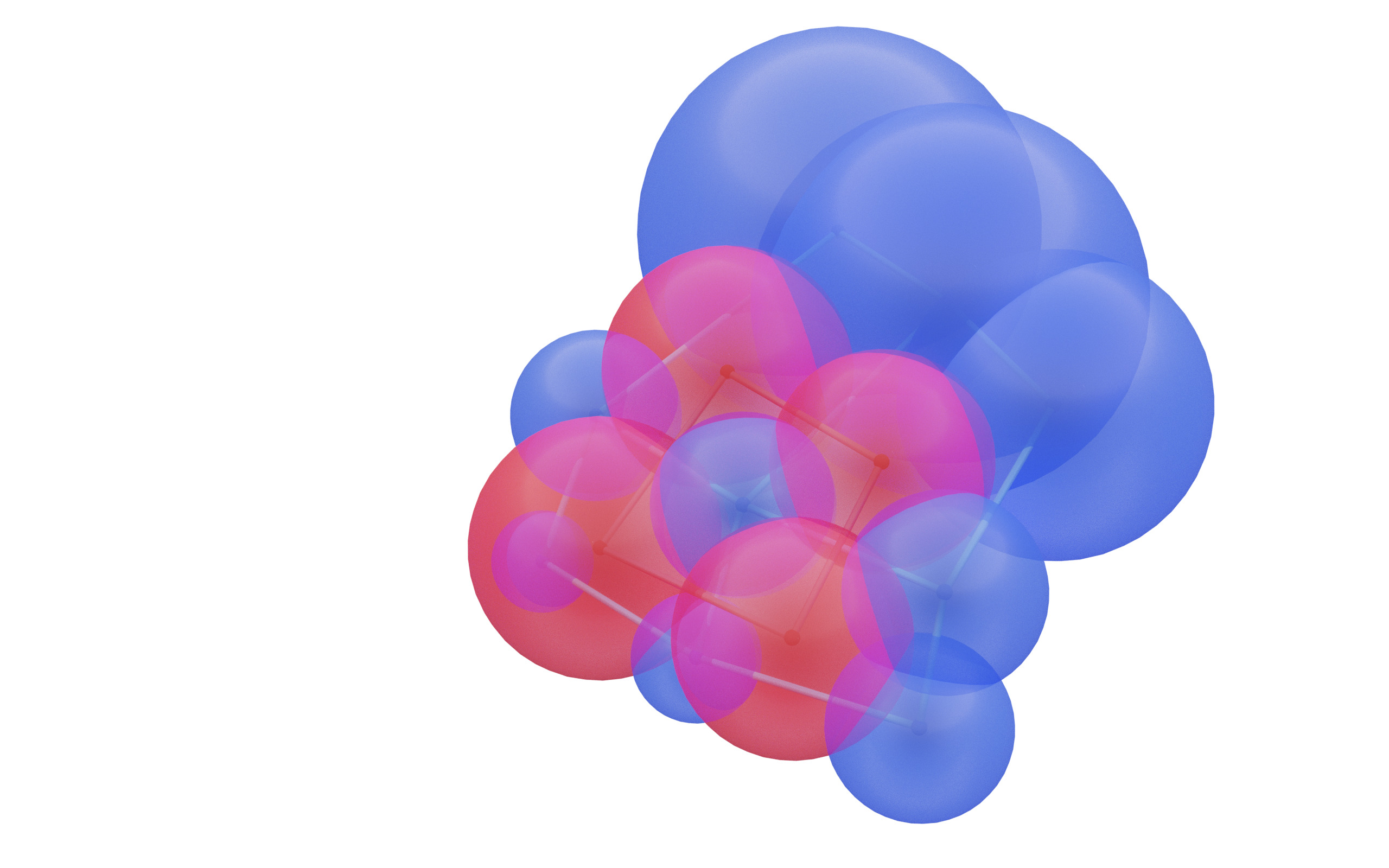}
  \caption{
    Orthogonal sphere representation of an orthogonal binet.
    Blue spheres are defined on $V$ and red spheres on $F$.
  }
  \label{fig:orthogonal-sphere-represantation}
\end{figure}
\begin{lemma}
  \label{lem:orthogonal-sphere-representation}
  Let $b_\mobq : D \rightarrow \RP^4 \setminus \PB^\pol$ be a polar binet with respect to the Möbius quadric $\mobq \subset \RP^4$.
  Let
  \[
    b_\sp \coloneqq \xi_\sp \circ b_\mobq
  \]
  be its orthogonal sphere representation, and
  \[
    b \coloneqq \pi_\eucl \circ b_\mobq
  \]
  the projection of $b_\mobq$ to the Euclidean space $\eucl$.
  Then
  \begin{enumerate}
  \item
    \label{lem:orthogonal-sphere-representation1}
    the two spheres $b_\sp(d)$ and $b_\sp(d')$ intersect orthogonally for all incident $d,d' \in D$,
  \item $b(d)$ is the center of $b_\sp(d)$ for all $d \in D$.\qedhere
  \end{enumerate}
\end{lemma}

\begin{proof}\
  \nobreakpar
  \begin{enumerate}
  \item
    Follows from Proposition~\ref{prop:orthogonal-spheres-ext}.
  \item
    As discussed in Section \ref{sec:moebius},
    for $X \in \RP^4$ the point $\pi_\eucl(X)$ is the center of the sphere $\xi_\sp(X)$.\qedhere
  \end{enumerate}
\end{proof}
The first key aspect that we want to emphasize in this article, is the close relation of polar binets in $\RP^4$ and orthogonal binets in $\eucl$. More specifically, the stereographic projection of a polar binet is an orthogonal binet. Conversely, an orthogonal binet can always be lifted to a polar binet. We formalize both statements in the following.
\begin{lemma}
  \label{lem:moebius-lift-projection}
  Let $b_\mobq : D \rightarrow \RP^4 \setminus \PB^\pol$ be a polar binet with respect to the Möbius quadric $\mobq \subset \RP^4$.
  Then its projection
  \[
    b \coloneqq \pi_\eucl \circ b_\mobq,
  \]
  to Euclidean space $\eucl$ is an orthogonal binet.
\end{lemma}
\begin{proof}
  Consider a cross $(v,f,v',f') \in C$.
  By Lemma~\ref{lem:orthogonal-sphere-representation}~\ref{lem:orthogonal-sphere-representation1},
  the two spheres $b_\sp(v)$ and $b_\sp(v')$ are orthogonal to both spheres $b_\sp(f)$ and $b_\sp(f')$.
  Therefore the line spanned by the two sphere centers $b(f) \vee b(f')$ is in the radical plane of $b_\sp(v)$ and $b_\sp(v')$.
  By Proposition~\ref{prop:radical-plane},
  this radical plane (and thus the line $b(f) \vee b(f')$) is orthogonal to the line $b(v) \vee b(v')$.
\end{proof}

We are now in a position to define the \emph{Möbius lift} of orthogonal binets (see Figure~\ref{fig:moebius-lift}).
\begin{definition}[Möbius lift]
	Let $b: D \rightarrow \eucl$ be an binet.
	Then a binet $b_\mobq: D \rightarrow \RP^4 \setminus \PB^\perp$ is called a \emph{Möbius lift} of $b$ if
	\begin{enumerate}
		\item $b_\mobq$ is a polar binet with respect to $\mobq$,
		\item $b$ is the projection of $b_\mobq$, that is, $\pi_\eucl(b_\mobq) = b$.\qedhere
	\end{enumerate}
\end{definition}
\begin{figure}[H]
  \centering
  \begin{overpic}[width=0.6\textwidth]{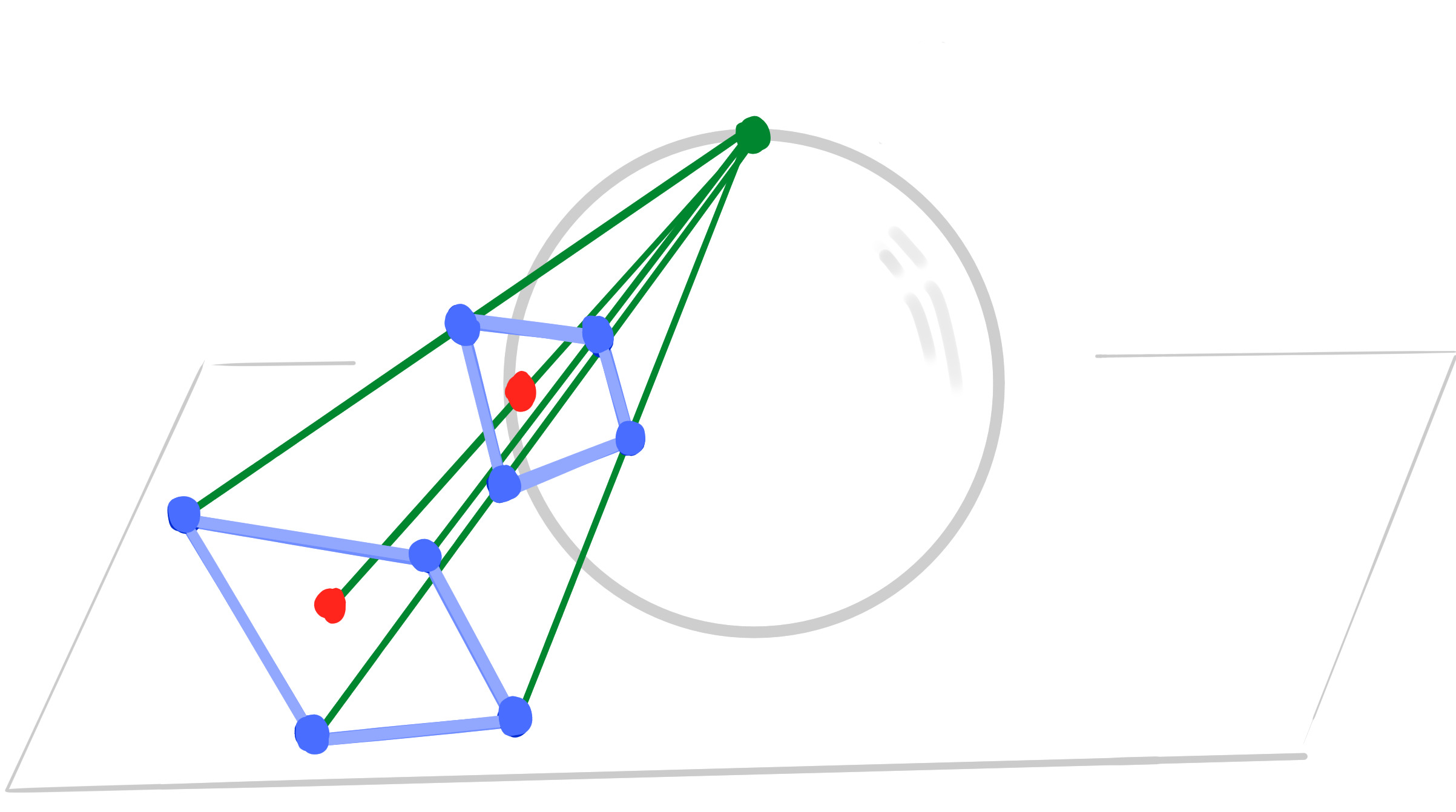}
    \put(85,6){$\eucl$}
    \put(63,44){$\mobq\subset\RP^4$}
    \put(51.5,48.5){\color{darkgreen}$\PB$}
    \put(37.2,4.5){\color{blue}$b(d)$}
    \put(45,24){\color{blue}$b_\mobq(d)$}
  \end{overpic}
  \caption{
    Möbius lift $b_\mobq$ of a binet $b$.
    A point $b(d) \in \eucl$ is lifted to a point $b_\mobq(d) \in \RP^4$ on the line $b(d) \vee \PB$.
  }
  \label{fig:moebius-lift}
\end{figure}
Lemma~\ref{lem:moebius-lift-projection} implies that the condition that $b$
is an orthogonal binet is necessary for the existence of a Möbius lift $b_\mobq$.
The following theorem shows that for regular binets it is also sufficient,
and thus guarantees the existence of a Möbius lift for an orthogonal binet.
\begin{theorem}\label{th:mobiusortho}
  Let $b: D \rightarrow \eucl$ be a regular binet.
  Then a Möbius lift $b_\mobq : D \rightarrow \RP^4 \setminus \PB^\perp$ of $b$ exists if and only if $b$ is a regular orthogonal binet.  
\end{theorem}

\begin{proof}\ \linebreak
  ($\Rightarrow$) 
  Follows from Lemma~\ref{lem:moebius-lift-projection}.
  \newline
  ($\Leftarrow$) 
  Assume $b$ is a regular orthogonal binet, and let us construct a corresponding Möbius lift $b_\mobq$.
  Because $b$ needs to be the stereographic projection of $b_\mobq$,
  for all $d \in D$ the point $b_\mobq(d)$ must lie on the line through $b(d)$ and $B$, that is
  \[
    b_\mobq(d) \in \ell(d) := b(d) \vee \PB.
  \]
  Note that for incident $d,d' \in D$ the two lines $\ell(d), \ell(d')$ cannot coincide since $b$ is regular. 
  Moreover, the point $b_\mobq(d)$ determines the point $b_\mobq(d')$ as the intersection of the line $\ell(d)$ with the polar hyperplane $b_\mobq(d)^\pol$:
  \[
    b_\mobq(d') = \ell(d) \cap b_\mobq(d)^\pol.
  \]
  Therefore, if for one $d\in D$ the lift $b_\mobq(d)$ is chosen arbitrarily on the line $\ell(d)$,
  then the whole image of $b_\mobq$ is determined.
  
  It remains to show that this construction is well-defined around crosses (compare Figure~\ref{fig:cross-spheres}).
  Consider a cross $(v,f,v',f') \in C$ and assume that the lifts of $v,f,f'$ satisfy $b_\mobq(v) \pol b_\mobq(f), b_\mobq(f')$,
  or equivalently, the corresponding sphere $b_\sp(v)$ is orthogonal to $b_\sp(f)$ and $b_\sp(f')$.
  Thus, $b(v)$ is in the radical plane $E$ of $b_\sp(f)$ and $b_\sp(f')$.
  By Proposition~\ref{prop:radical-plane}, the line $b(f) \vee b(f')$ is orthogonal to $E$.
  On the other hand, this line is also orthogonal to $b(v) \vee b(v')$ since $b$ is an orthogonal binet.
  Thus, the point $b(v')$ also lies in $E$.
  Hence, there is a sphere $b_\sp(v')$ with center $b(v')$ that is orthogonal to both $b_\sp(f)$ and $b_\sp(f')$,
  or equivalently, there is a well-defined lift $b_\mobq(v')$ which is polar to both $b_\mobq(f)$ and $b_\mobq(f')$.
\end{proof}
\begin{figure}[H]
  \centering
  \begin{overpic}[width=0.32\textwidth]{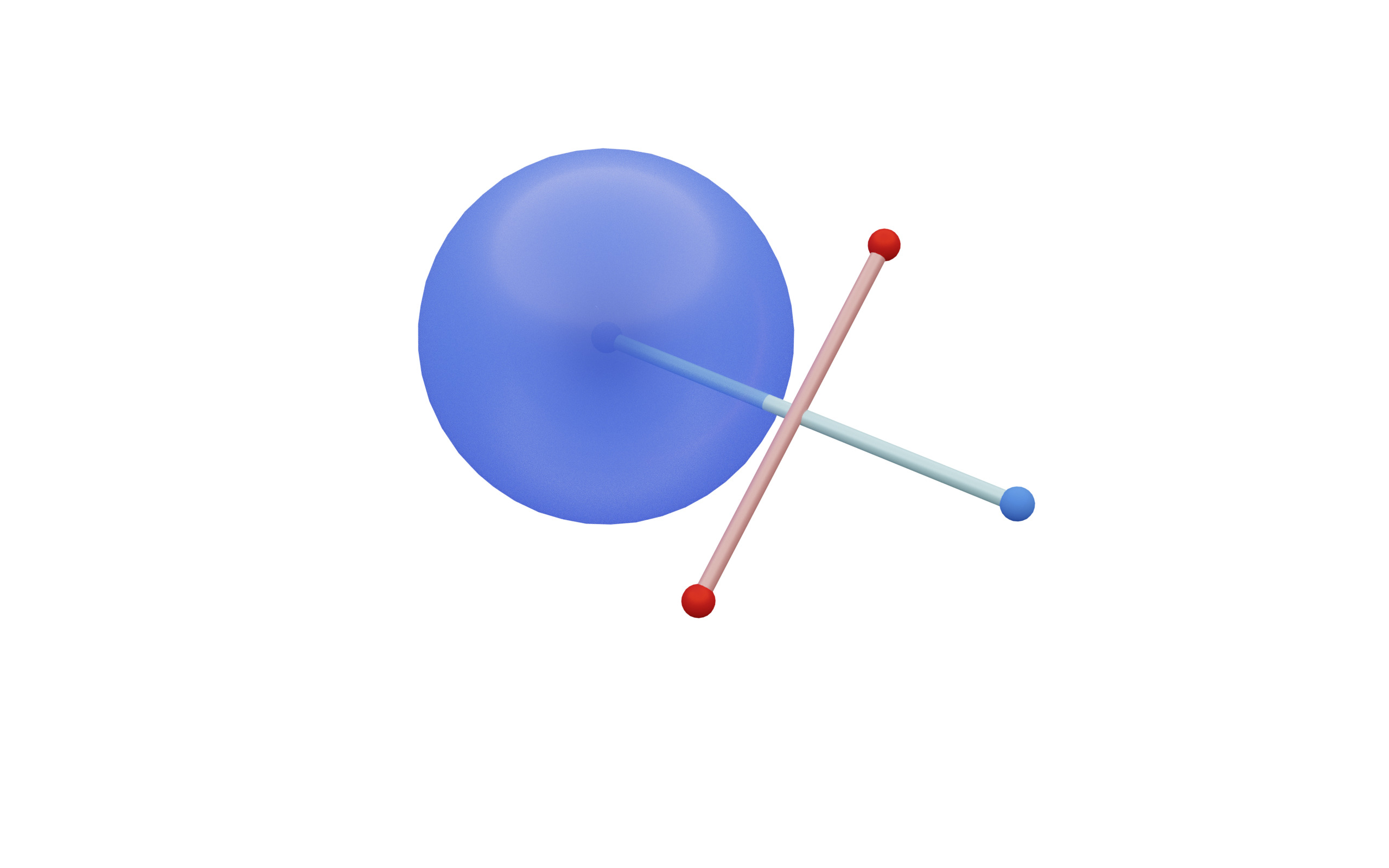}
    \put(11,38){$\color{blue}b_\sp(v)$}
  \end{overpic}
  \begin{overpic}[width=0.32\textwidth]{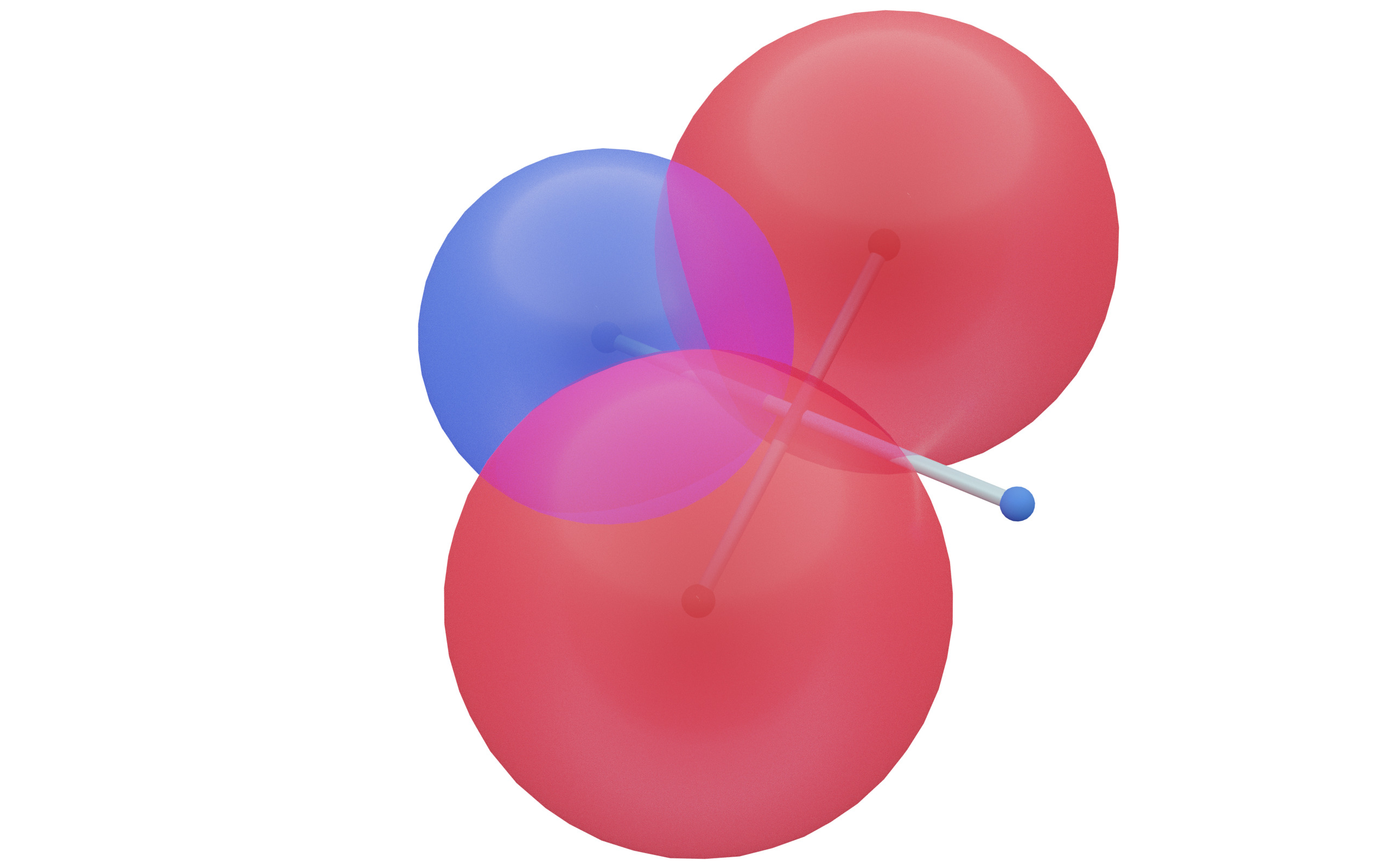}
    \put(70,12){$\color{red}b_\sp(f)$}
    \put(81,45){$\color{red}b_\sp(f')$}
  \end{overpic}
  \begin{overpic}[width=0.32\textwidth]{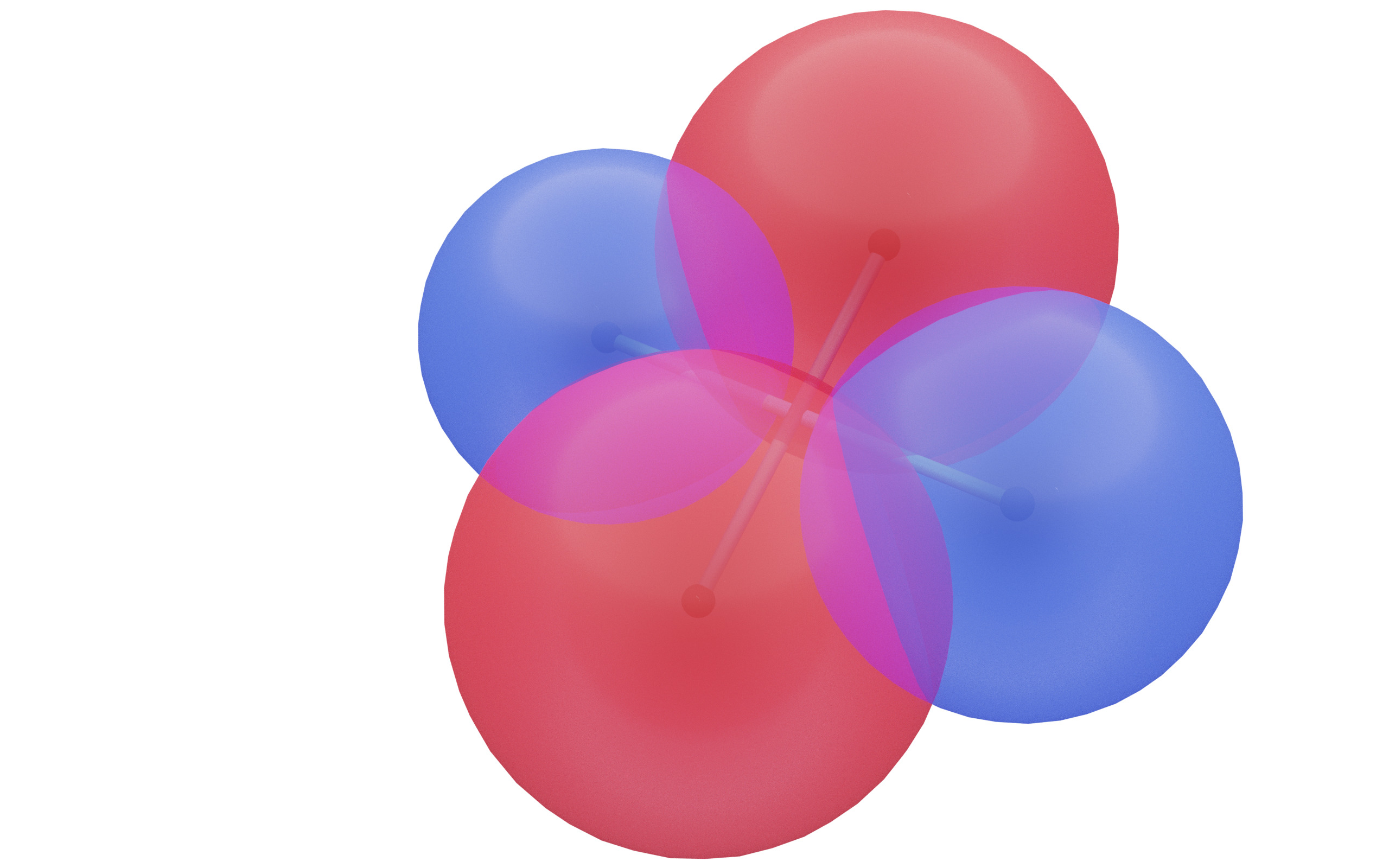}
    \put(90.5,22){$\color{blue}b_\sp(v')$}
  \end{overpic}
  \caption{
    Constructing four spheres of an orthogonal sphere representation of an orthogonal binet around a cross $(v, f, v', f') \in C$.
    The radius for $b_\sp(v)$ can be chosen arbitrarily. Then $b_\sp(f)$ and $b_\sp(f')$ are uniquely determined.
    The two possible results for $b_\sp(v')$ coincide if and only if the two dual edges $b(v) \vee b(v')$ and $b(f) \vee b(f')$ are orthogonal.
  }
  \label{fig:cross-spheres}
\end{figure}

\begin{remark}
  Note that for the existence of a Möbius lift, the binet regularity condition that $b(d) \neq b(d')$ for incident $d,d'\in D$ is essential. However, the net regularity condition that $b(v) \neq b(v')$ for adjacent $v,v'\in V$ is not, since $b(v) = b(v')$ just implies that $b_\mobq(v) = b_\mobq(v')$.
\end{remark}

From the proof of Theorem \ref{th:mobiusortho}
follows that each orthogonal binet has a 1-parameter family of Möbius lifts.
One point $b_\mobq(d)$ can be chosen arbitrarily on the line $b(d) \vee \PB$
(or correspondingly, the radius of one sphere $b_\sp(d)$ with center $b(d)$ can be chosen arbitrarily)
and then the remaining points of the lift are uniquely determined.

\begin{coordinates}
  The existence of the lift can easily be checked using the orthogonality condition \eqref{eq:orthogonal-spheres}.
  Given a binet $b$, a function $\rho : D \rightarrow \R$ with
  \[
    \sca{b(d), b(d')} = \rho(d) + \rho(d')
    \quad \text{for all incident}~ d, d' \in D,
  \]
  exists if and only if
  \[
    \sca{b(v), b(f)} - \sca{b(f), b(v')} + \sca{b(v'), b(f')} - \sca{b(f'), b(v)} = 0.
  \]
  But this equation is equivalent to $\sca{b(v) - b(v'), b(f) - b(f')} = 0$, i.e.,
  the condition that $b$ is an orthogonal binet.
  By \eqref{eq:Möbius-lift-spheres}, the Möbius lift of $b$ is then given by
  \[
    b_\mobq(d) = [b(d) + e_0 + 2\rho(d) e_\infty ].
  \]
  The function $\rho$ is related to the centers $c$ and radii $r$ of the spheres of the orthogonal sphere representation by
  \[
    2\rho = \abs{c}^2 - r^2.
  \]
  Note that -- using the function $\rho$ -- a different choice for the Möbius lift is represented by the simple transformation
  \[
    \begin{aligned}
      &\rho(v) \rightarrow \rho(v) + \varepsilon, &&v \in V,\\
      &\rho(f) \rightarrow \rho(f) - \varepsilon,  &&f \in F,
    \end{aligned}
  \]
  with $\varepsilon \in \R$.
\end{coordinates}

\begin{lemma}\label{lem:regularmobiuslift}
	If $b$ is a regular orthogonal binet, then every Möbius lift $b_\mobq$ of $b$ is a regular polar binet.
\end{lemma}
\proof{
	Note that for incident $d,d'\in D$ the lines $\ell(d), \ell(d')$ (as in the proof of Theorem \ref{th:mobiusortho}) intersect only in $\PB$, since $b(d) \neq b(d')$ in a regular binet. Therefore, we also obtain that $b_\mobq(d) \neq b_\mobq(d')$. An analogous argument shows that for adjacent $v,v' \in V$ (or $f,f' \in F$) the points $b_\mobq(v), b_\mobq(v')$ are different. Additionally, the dimension of subspaces under stereographic projection is non-increasing. Therefore, if three points span a plane in $\eucl$ their Möbius lifts also span a plane. Thus, all the regularity conditions of Definitions \ref{def:net} and \ref{def:binet} are satisfied.\qed
}

Note that the converse is not true, not every projection of a regular polar binet in $\RP^4$ is a regular orthogonal binet, since points of the polar binet may be in special position in relation to $\PB$.

\begin{remark}[Generalization to $\eucl^n$]
For simplicity, we restricted our discussion of the Möbius lift of orthogonal binets to $\eucl = \eucl^3$.
Yet with the projective model of $n$-dimensional Möbius geometry
and Definition~\ref{def:orthogonal-binets} of orthogonal binets for $\eucl^n$,
all claims generalize without obstruction. In particular, the following holds:

Let $b: D \rightarrow \eucl^n$ be a regular binet.
Then a Möbius lift $b_\mobq : D \rightarrow \RP^{n+1} \setminus \PB^\perp$ of $b$ exists if and only if $b$ is an orthogonal binet.
\end{remark}

\subsection*{Invariance}
Let $T$ be a Möbius transformation of $\eucl \cup \{\infty\}$ and $b$ an orthogonal binet.
In general, $T\circ b$ is not an orthogonal binet anymore.
However, if we choose a Möbius lift $b_\mobq$ we may apply the corresponding Möbius transformation $\tilde T$ of $\RP^4$ to $b_\mobq$.
Then $\tilde T \circ b_\mobq$ is a polar binet and therefore the projection $\pi_\eucl \circ \tilde T \circ b_\mobq$ is an orthogonal binet.
Since each orthogonal binet has a 1-parameter family of Möbius lifts, the action of a Möbius transformation on an orthogonal binet is not unique.
However, if we consider an orthogonal binet $b$ together with a fixed Möbius lift $b_\mobq$, or equivalently, together with an orthogonal sphere representation $b_{\sp}$, the action of a Möbius transformation is unique.

\section{*nets and bi*nets} \label{sec:bistarnets}

We define a *net as map $V \rightarrow \planesof{\RP^n}$ into the space of (2-dimensional) planes of $\RP^n$.
This notion is motivated by Q*-nets \cite{ddgbook}.
We view a *net as a discretization of the tangent planes of a smooth parametrized surface.

\begin{definition}[*nets] \label{def:starnet}
  A \emph{*net} is a map $p: V \rightarrow \planesof{\RP^n}$
  such that for every edge $(v,v') \in E$ the planes $p(v), p(v')$ intersect in a line.
  A \emph{regular *net} is a *net such that the planes of any three vertices of each face intersect in exactly one point.
\end{definition}

Again, the identification of the faces $F$ with the vertices of the dual graph (see Remark~\ref{rem:dual-graph})
yields an analogous definition for *nets as maps $p : F \rightarrow \planesof{\RP^n}$.
We define bi*nets as maps on $D$ that come from pairs of *nets on $V$ and $F$ (see Figure~\ref{fig:bi-star-net}).
\begin{definition}[Bi*nets]
  A \emph{bi*net} is a map $b: D \rightarrow \planesof{\RP^n}$. A \emph{regular bi*net} is a bi*net, such that
  \begin{enumerate}
	  \item
	    the restrictions to $V$ and $F$ are regular *nets,
	  \item
	    and for all incident $d,d'\in D$ the planes $b(d),b(d')$ are distinct.\qedhere
	  \end{enumerate}		    
\end{definition}
\begin{figure}[H]
  \centering
  \begin{overpic}[width=0.6\textwidth]{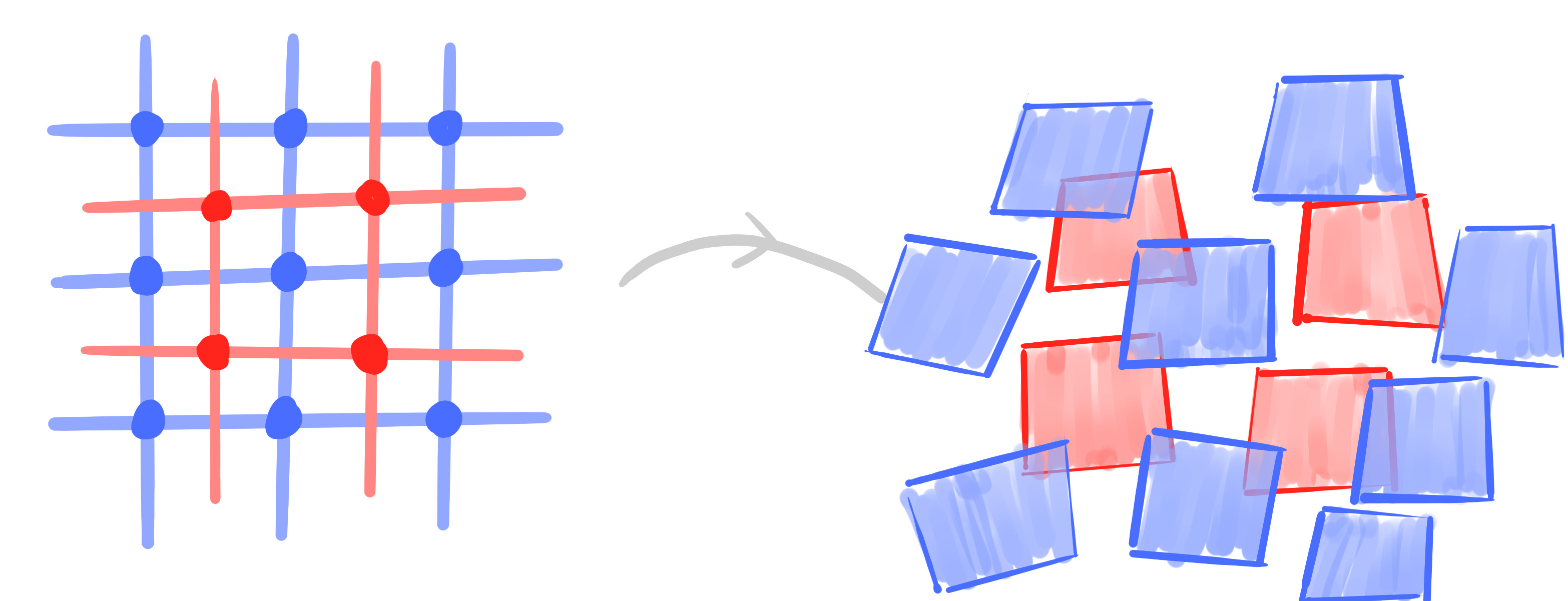}
    \put(3,5){$D$}
    \put(95,30){$\RP^n$}
  \end{overpic}
  \caption{A bi*net is a map $D = V \cup F \rightarrow \planesof{\RP^n}$.}
  \label{fig:bi-star-net}
\end{figure}

For a projective subspace $K \subset \RP^n$ we denote its dual subspace by
\[
\dual{K} \subset (\RP^n)^*.
\]
By projective duality the space $\planesof{\RP^3}$ can be identified with the dual projective space $(\RP^3)^*$.
Thus, for $n=3$ a *net $p: \Z^2 \rightarrow \planesof{\RP^3}$ defines a net on the dual space $\dual{p}: \Z^2 \rightarrow (\RP^3)^*$.
And similarly, a bi*net $b: D \rightarrow \planesof{\RP^3}$ defines a binet on the dual space $\dual{b}: D \rightarrow (\RP^3)^*$.

\section{Conjugate bi*nets}
\label{sec:conjugate-bi-star-nets}
A dual version of conjugate nets, where the role of points and planes is interchanged, are conjugate *nets.
They are also known as Q*-nets \cite{dsqstarnet, bprmultinets}, which were originally introduced for $\RP^3$ only.

\begin{definition}[Conjugate *nets]
  A \emph{conjugate *net} is a net $p: V \rightarrow \planesof{\RP^n}$, such that the four planes around each face intersect in a point.
\end{definition}

\begin{definition}[Conjugate bi*nets]
  A \emph{conjugate bi*net} is a bi*net $b$ such that the restrictions to $V$ and $F$ are conjugate *nets is .
\end{definition}

For each conjugate binet $b$ there is a conjugate bi*net $\square b$, such that for all $v\in V$ (and $f \in F$) the plane $\square b(v)$ contains the points of $b$ of the corresponding quad in $F$ (resp. $V$). If $b$ is a regular conjugate binet, then $\square b$ is unique. More explicitly,
\begin{align}
	\square b(d) = \spa\set{b(d')}{d' \inc d}. 
\end{align}
Note that even if $b$ is a regular conjugate binet, $\square b$ is not necessarily a regular conjugate bi*net.
Conversely, for each conjugate bi*net $b$ there is a conjugate binet $\square^* b$ defined by the intersection points of $b$. If $b$ is a regular conjugate bi*net then
\begin{align}
	\square^*b(d) = \cap\set{b(d')}{d' \inc d}. 
\end{align}
Moreover, assuming sufficient regularity, we have the obvious identities $\square^* \circ \square = \id$ and $\square \circ \square^* = \id$.

Thus, there is no actual difference between conjugate binets and conjugate bi*nets, since they are in bijection.

\section{Orthogonal bi*nets} \label{sec:orthobinstarnets}

For *nets and bi*nets in Euclidean space $\eucl^n$, we introduce some additional regularity conditions compared to the projective space case.
For the regularity of *nets we essentially replace ``distinct'' with ``non-parallel'' and restrict the incidence of planes to $\eucl^n$.
\begin{definition}[Regular~*nets in $\eucl^n$] \label{def:starneteuclidean}
  A \emph{regular *net} in Euclidean space is a *net $p: V \rightarrow \planesof{\eucl^n}$
  such that the planes of any three vertices of each face intersect in exactly one point in $\eucl^n$.
\end{definition}
\begin{remark}
  Note that the regularity constraint for *nets in $\eucl^n$ implies
  that for every edge $(v,v') \in E$ the planes $p(v), p(v')$ are not parallel.
\end{remark}
\begin{definition}[Regular bi*nets in $\eucl^n$]\label{def:bistarneteuclidean}
  A \emph{regular bi*net} in Euclidean space is a regular bi*net $p: D \rightarrow \planesof{\eucl^n}$, such that
    for all incident $d,d'\in D$ the two planes $b(d),b(d')$ are not parallel and not orthogonal.
\end{definition}
Analogous to orthogonal binets, we implement the following orthogonality condition on the crosses of bi*nets.
\begin{definition}[Orthogonal bi*nets]
  \label{def:orthogonal-bi-star-net}
  A bi*net $b: D \rightarrow \planesof{\eucl^n}$ is an \emph{orthogonal bi*net} if $b(v) \cap b(v') ~\orth~ b(f) \cap b(f')$ for all crosses $(v,f,v',f') \in C$, whenever the occurring lines are not at infinity.
\end{definition}

In particular, the definition implies that at each cross of a regular orthogonal bi*net the two lines are not at infinity, and are therefore orthogonal to each other.

For simplicity of the presentation and in regards of normal binets and the lift to Laguerre and Lie geometry,
we restrict the following discussion of orthogonal bi*nets to $n=3$.
Thus, again, we denote the 3-dimensional Euclidean space by
\[
  \eucl \coloneqq \eucl^3,
\]
and the space of planes in $\eucl$ by
\[
  \pl \coloneqq \planesof{\eucl}.
\]

\subsection*{Invariance}
The definition of orthogonal bi*nets is invariant under similarity transformations.

In the smooth theory, if a parametrized surface is described as the envelope of its (oriented) tangent planes,
the notion of a Gauß-orthogonal parametrization (third fundamental form diagonal)
is invariant under \emph{Laguerre transformations} of $\eucl$.
Thus, while Möbius invariance is a desirable property for a discretization of orthogonal parametrizations in terms of points on the surfaces,
Laguerre invariance is a desirable property for a discretization of Gauß-orthogonal parametrizations in terms of tangent planes of the surface.
Indeed, while circular nets are Möbius invariant, it is a feature of conical nets (see Definition~\ref{sec:laguerrelift}) that they are Laguerre invariant.
Orthogonal bi*nets are invariant under similarity transformations, yet a priori not under Laguerre transformations.
However, in Section~\ref{sec:laguerrelift} we introduce a Laguerre lift of orthogonal bi*nets,
which does give us the ability to apply Laguerre transformations to orthogonal bi*nets
such that we obtain an orthogonal bi*net again.

\section{Normal binets} \label{sec:normalbinets}
Normal binets play the role of the Gauß map of a surface in the smooth theory.
In the context of binets, the condition of its image points lying on the unit sphere is replaced by polarity with respect to the unit sphere (see Figure~\ref{fig:normal-binet}).
Thus, let
\[
  \unis \subset \eucl \subset \RP^3
\]
be the unit sphere centered at the origin $\ori \in \eucl$.

\begin{definition}[Normal binet]
  \label{def:normal-binet}
  Let $b: D \rightarrow \pl$ be a bi*net.
  A binet $n: D \rightarrow \eucl \setminus \{\ori\}$ is a \emph{normal binet} of $b$ if
  \begin{enumerate}
  \item
    \label{def:normal-binet1}
    $n$ is a polar binet with respect to the unit sphere $\unis$,
  \item
    \label{def:normal-binet2}
    $(\ori \vee n(d)) \orth b(d)$ for all $d \in D$.\qedhere
  \end{enumerate}
\end{definition}
\begin{figure}[H]
  \centering
  \begin{overpic}[width=0.7\textwidth]{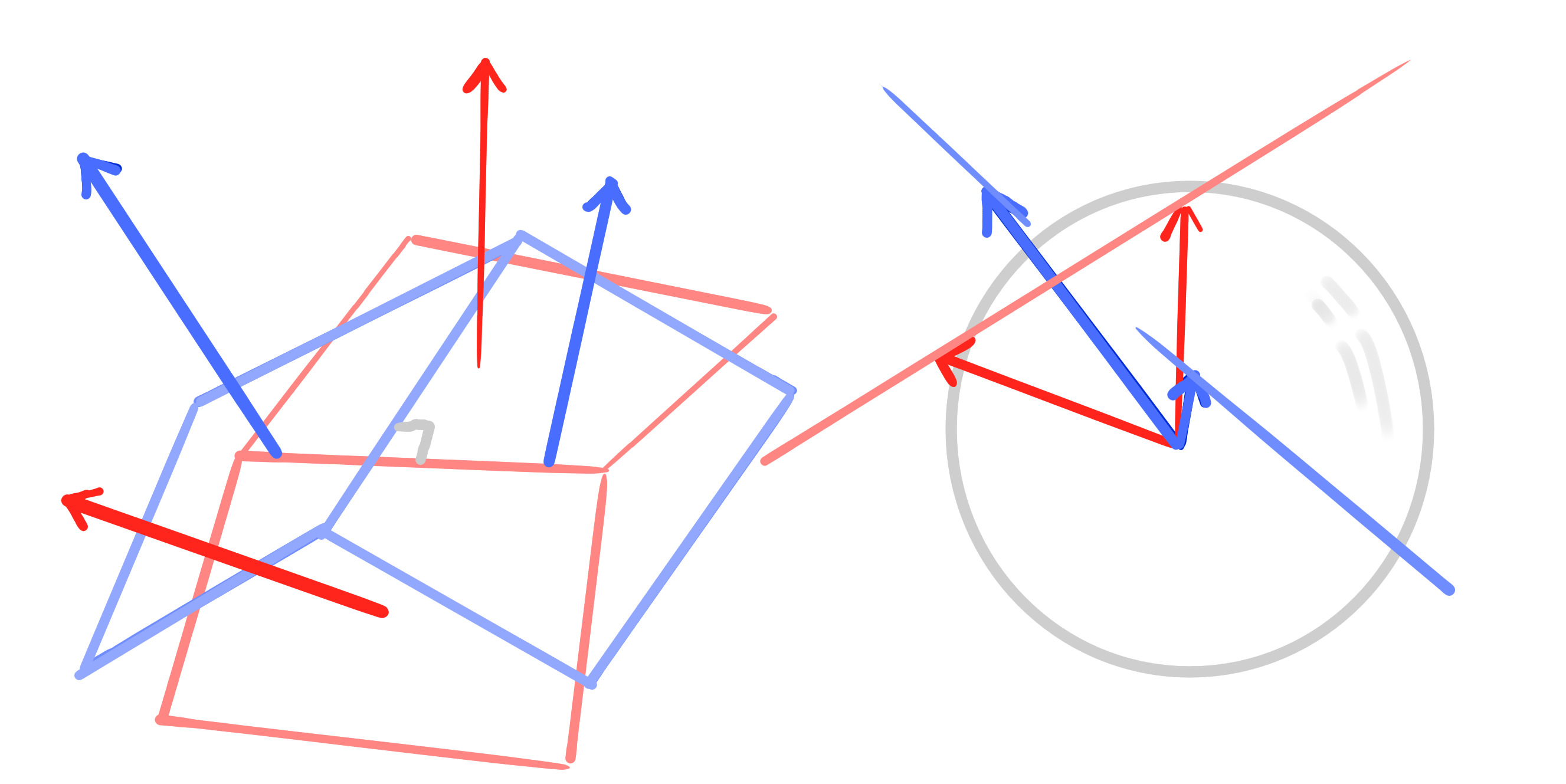}
    \put(91,31){$\unis \subset \eucl \subset \RP^3$}
    \put(75,18){$\ori$}
    \put(1,42){\color{blue}$u(v)$}
    \put(37,40){\color{blue}$u(v')$}
    \put(-3,19){\color{red}$u(f)$}
    \put(28,48){\color{red}$u(f')$}
    \put(63,39){\color{blue}$n(v)$}
    \put(76,27){\color{blue}$n(v')$}
    \put(52.9,28.4){\color{red}$n(f)$}
    \put(77,36){\color{red}$n(f')$}
  \end{overpic}
  \caption{
    Left: Unit normals $u$ of the planes of an orthogonal bi*net $b$ at a cross $(v,f,v',f') \in C$.
    Right: Normals $n$ centered at the origin $\ori$ and rescaled such that ${n(v) \vee n(v')}$ and ${n(f) \vee n(f')}$ are polar with respect to the unit sphere $\unis$.
    The normals $n$ constitute the normal binet of $b$.
  }
  \label{fig:normal-binet}
\end{figure}
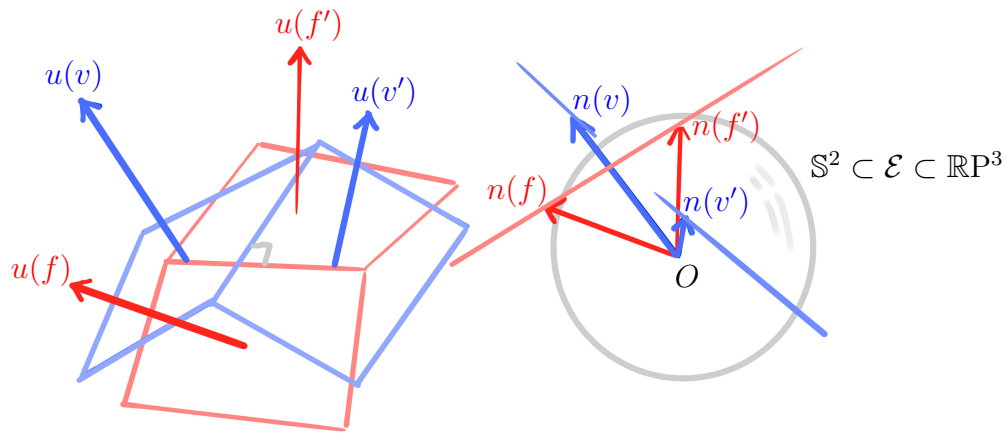
Two lines that are polar with respect to the unit sphere are also orthogonal.
Thus, normal binets are also orthogonal binets, which we put into a small lemma.
\begin{lemma}
  A normal binet is an orthogonal binet.
\end{lemma}
Interestingly, normal binets exist only for orthogonal bi*nets.
This establishes that orthogonal bi*nets may be viewed as a discretization of Gauß-orthogonal parametrizations,
that is parametrizations for which the third fundamental form is diagonal, or equivalently,
parametrizations for which the Gauß map is orthogonal.
\begin{theorem}
  \label{th:orthonormals}
  Let $b: D \rightarrow \pl$ be a regular bi*net.
  Then there exists a normal net $n: D \rightarrow \eucl \setminus \{O\}$ of $b$ if and only if $b$ is an orthogonal bi*net.
\end{theorem}
\begin{proof}\ \linebreak
  ($\Rightarrow$) 
  Assume that $n$ is a normal binet of $b$ and consider a cross $(v,f,v',f') \in C$.
  Let
  \begin{align}
    A &\coloneqq n(v) \vee n(v') \vee \ori, & A' &\coloneqq n(f) \vee n(f') \vee \ori,\\
    \ell &\coloneqq b(v) \cap b(v'), & \ell' &\coloneqq b(f) \cap b(f').
  \end{align}
  Due to the polarity condition on $n$,
  we observe that $A \orth A'$.
  The orthogonality condition of $n$ with respect to $b$
  implies that $\ell \orth A$ and $\ell' \orth A'$.
  Combined, we obtain that $\ell \orth \ell'$ and therefore $b$ is an orthogonal bi*net.
  \newline
  ($\Leftarrow$) 
  Assume that $b$ is an orthogonal bi*net.
  By the orthogonality condition, for $d\in D$ the lines
  \[
    U(d) \coloneqq \ori \vee n(d)
  \]
  are determined by $b(d)$.
  Moreover, for $d\inc d'$ the point $n(d)$ determines $n(d')$ by
  \[
    n(d') = U(d') \cap n(d)^\pol.
  \]
  This intersection always exists in $\eucl$ since, by Definition~\ref{def:bistarneteuclidean},
  the two planes $b(d)$ and $b(d')$ are not orthogonal.
  Therefore, for one $d\in D$ we may choose $n(d)$ on $U(d) \setminus \{\ori\}$, and then the whole image of $n$ is determined.

  It remains to check that in this manner $n$ is well-defined around a cross ${(v,f,v',f') \in C}$ (compare Figure~\ref{fig:normal-binet}).
  We use that $U(v)$ and $U(v')$ are distinct and also that $U(f)$ and $U(f')$ are distinct as a consequence of the regularity of *nets in $\eucl$ (Definition~\ref{def:starneteuclidean}).
  Assume $n(v),n(f),n(f')$ are already determined such that $n(f),n(f') \pol n(v)$.
  We show that $n(f)^\pol, n(f')^\pol$ and $U(v')$ intersect in exactly one point.
  Since $b$ is an orthogonal bi*net, the plane $A = U(v) \vee U(v')$ is orthogonal to the line $L' = n(f) \vee n(f')$.
  On the other hand, by construction, the line $L = n(f)^\pol \cap n(f')^\pol$ is orthogonal to the line $L'$ and contains $n(v)$.
  Therefore $L$ must lie in $A$ and thus intersects $U(v')$, which means $n(v') = L \cap U(v')$ is well-defined.
\end{proof}

From the proof of Theorem \ref{th:orthonormals} follows that each orthogonal bi*net has a 1-parameter family of normal binets.
The distance of one point $n(d)$ to the origin can be chosen arbitrarily
(the length of one normal vector of one plane can be chosen arbitrarily)
and then the remaining points of $n$ are uniquely determined.

\begin{coordinates}
  \label{coo:normallift}
  Let $b: D\rightarrow \pl$ be an orthogonal bi*net and let $u: D \rightarrow \unis$
  be a corresponding unit-normal binet,
  that is a binet such that $|u(d)| = 1$ and $(u(d) \vee \ori) \orth b(d)$ for all $d \in D$.

  The vectors describing the points of a normal binet $n : D \rightarrow \R^3 \setminus \{O\}$ must be proportional to $u$
  \[
    n(d) = \frac{u(d)}{\sigma(d)},
  \]
  with some $\sigma(d) \neq 0$ and satisfy the polarity condition
  \[
    \sca{n(d), n(d')} = 1.
  \]
  Thus, a normal binet of $b$ exists if and only there exists a function
  $\sigma: D \rightarrow \R \setminus \{0\}$, such that for all incident $d,d' \in D$ holds
  \begin{align}
    \sca{u(d), u(d')} = \sigma(d) \sigma(d').
  \end{align}
  Such a function $\sigma$ exist if only if around every cross $(v,f,v',f') \in C$ holds
  \begin{align}
    \sca{u(v), u(f)} \sca{u(v'), u(f')} = \sca{u(v'), u(f)} \sca{u(v), u(f')}.
    \label{eq:normallift}
  \end{align}

  On the other hand, the two lines $b(v) \cap b(v')$ and $b(f) \cap b(f')$ are orthogonal if and only if
  \[
    \begin{aligned}
      0 &= \sca{ u(v) \times u(v'), u(f) \times u(f')}\\
      &= \sca{u(v), u(f)} \sca{u(v'), u(f')} - \sca{u(v'), u(f)} \sca{u(v), u(f')}.
    \end{aligned}
  \]
  
\end{coordinates}

Let us also briefly discuss the following additional property of normal binets.

\begin{lemma} \label{lem:normalconjugate}
	A normal binet is a conjugate binet.
\end{lemma}
\proof{
	We have defined normal binets as polar binets with respect to the unit sphere $\unis \subset \RP^3$. Thus if $d_1,d_2,d_3,d_4 \in D$ are incident to some $d\in D$, then the points $n(d_i)$ are contained in the plane $n(d)^\pol$. Therefore $n$ is a polar binet.\qed
}

\begin{remark}
  \label{rem:koebe}
  With this observation it becomes quite clear that normal binets are a generalization of Koebe polyhedra as mentioned in Example~\ref{itm:examplekoebe} in the introduction (see Figure~\ref{fig:examples}). A Koebe polyhedron is a conjugate net whose edges are tangent to the unit sphere $\unis \subset \eucl$. It is always accompanied by a dual Koebe polyhedron, which has edges that intersect the primal edges orthogonally in the points of contact to the unit sphere. Thus each pair of primal and dual Koebe polyhedra constitutes a polar binet with respect to the unit sphere. In fact, in \cite{bhssminimal} Koebe polyhedra are interpreted as normal nets for certain discretizations of minimal surfaces. We give a few more details in Remark~\ref{rem:minimal}.
\end{remark}

\begin{remark}
	Let us discuss normal binets in the case of non-regular bi*nets.
	\begin{enumerate}
		\item Assume $b$ is a regular orthogonal bi*net except that $b(d) \orth b(d')$ for some incident $(d,d') \in D$, which violates the regularity condition of bi*nets (Definition~\ref{def:bistarneteuclidean}).
		In this case, a normal binet $n: D \rightarrow \eucl \setminus \{\ori\}$ does not exist,
		because $n(d)^\perp$ does not intersect the line $n(d') \vee \ori$ in $\eucl$.
		But an intersection point does exist in $\RP^3 \supset \eucl$.
		Thus, a normal binet $n: D \rightarrow \RP^3 \setminus \{\ori\}$ may indeed exist.
		However, if we choose $n(d)$ at infinity,
		for \emph{all} $d'' \in D$ incident to $d$ the lines $n(d'')\vee \ori$ need to be orthogonal to the line $n(d) \vee \ori$.
		This puts an \emph{additional} constraint on the orthogonal bi*net $b$.
		While we exclude this special case from the general definition of orthogonal bi*nets,
		it may still be practical to allow in some applications.

		\item Assume $b$ is a regular orthogonal bi*net except that $b(d) \parallel b(d')$ for some incident $d,d' \in D$, which violates the regularity condition of bi*nets (Definition~\ref{def:bistarneteuclidean}).
		This implies that $n(d) = n(d') \in \S^2$. If this happens for only one pair of incident $d,d'$, this is not yet a problem. 
		But if there is more than one pair, a normal binet only exists if additional, non-local constraints are satisfied by the orthogonal bi*net, and we did not investigate this case further.
		\item Assume $b$ is a regular orthogonal bi*net except that $b(v) \parallel b(v')$ for some adjacent $(v,v') \in E$, which violates the regularity condition of *nets (Definition~\ref{def:starneteuclidean}).
		Again, this implies that $n(d) = n(d')$, but not necessarily in $\S^2$. As a result, in this case the violation of regularity is not an obstruction to the existence of a normal binet.\qedhere
	\end{enumerate}
  
\end{remark}

\begin{lemma}
	\label{lem:regnormalnet}
	If $b$ is a regular orthogonal bi*net, then every normal binet $n$ of $b$ is a regular binet.
\end{lemma}
\proof{	
	Note that for incident $d,d'\in D$ the lines $U(d), U(d')$ (as in the proof of Theorem \ref{th:orthonormals}) intersect only in $\ori$, since $b(d) \neq b(d')$ in a regular bi*net. Also, for any $d \in D$ holds that $n(d) \neq \ori$. Therefore, we also obtain that $n(d) \neq n(d')$. An analogous argument shows that for adjacent $v,v' \in V$ (or $f,f' \in F$) the points $b_\mobq(v), b_\mobq(v')$ are different. Additionally, the condition that three planes of a face in $b$ intersect in exactly one point in $\eucl$, implies that any three points of a face in $n$ span a plane. Thus all the regularity conditions of Definitions \ref{def:starneteuclidean} and \ref{def:bistarneteuclidean} are satisfied.\qed
}

\section{Laguerre geometry}
\label{sec:laguerre}

First, recall the projective model of dual Euclidean geometry \cite[Appendix A]{blptlaguerre}.
For the 3-dimensional Euclidean space
\[
  \eucl = \RP^3 \setminus E^\infty,
\]
with plane at infinity $E^\infty \subset \RP^3$, the \emph{dual Euclidean space} is given by
\[
  \eucl^* = (\RP^3)^* \setminus \{ B \},
\]
where the point $B \coloneqq \star E^\infty$ is the dual of the plane at infinity.
By duality, every point $X \in \eucl^*$ corresponds to a Euclidean plane $\star X \subset \eucl$.
The absolute quadric $\star\euclq$ of dual Euclidean space has signature $\texttt{(+++0)}$.
It is an imaginary cone with real apex $B$, and corresponds dually to the absolute conic $\euclq$ in $E^\infty$ of Euclidean space.

\begin{coordinates}
  Using Coordinates~\ref{coords:euclidean} for the Euclidean space, we obtain by duality
  \[
    B = \star E^\infty = [0, 0, 0, 1].
  \]
  The dual absolute quadric is represented by
  \[
    \sca{x,x}_{\star \euclq} = x_1^2 + x_2^2 + x_3^3, \qquad x \in \R^4.
  \]
\end{coordinates}

Secondly, recall the projective model of Laguerre geometry \cite{blaschkevl,blptlaguerre}. 
Let $\sca{\cdot, \cdot}_{\blac}$ be a symmetric bilinear form of signature $\texttt{(+++-0)}$, and
\[
  \blac = \set{[x] \in \RP^4}{ \sca{x,x}_{\blac} = 0},
\]
the corresponding quadric in $\RP^4$, which we call the \emph{Blaschke cylinder} (see Figure~\ref{fig:laguerre-geometry}).
\begin{coordinates}
  In homogeneous coordinates of $\RP^4$ we choose
  \[
    \sca{x,x}_{\blac} = x_1^2 + x_2^2 + x_3^2 - x_4^2.
  \]
  And thus, in affine coordinates $x_4 = 1$, the Blaschke cylinder is the 3-cylinder:
  \[
    x_1^2 + x_2^2 + x_3^2 = 1.
  \]
\end{coordinates}

\begin{figure}[H]
  \centering
  \begin{overpic}[width=0.98\textwidth]{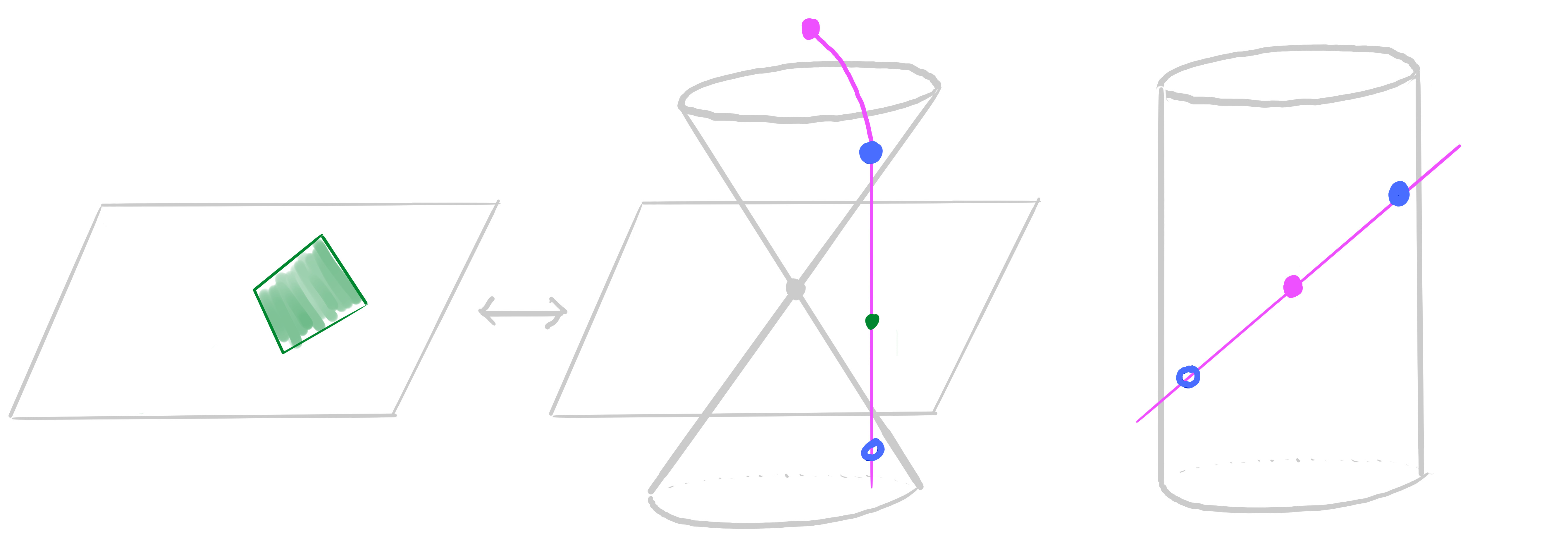}
    \put(2,8.4){$\eucl \subset \RP^3$}
    \put(20,11){$\color{darkgreen} \star X$}
    \put(29,11.2){\small duality}
    \put(36.5,8.4){$\eucl^* \subset S_{\eucl^*}$}
    \put(60.4,28){$\blac \subset \RP^4$}
    \put(47.4,15){$\PB$}
    \put(48.2,31.5){$\color{magenta}\PM$}
    \put(56.2,12.7){$\color{darkgreen}X$}
    \put(56.2,23.4){$\color{blue}X_1$}
    \put(56.2,4.5){$\color{blue}X_2$}
    \put(91,28){$\blac \subset \RP^4$}
    \put(81,17){$\color{magenta}\PM$}
    \put(86.8,23.2){$\color{blue}X_1$}
    \put(76,7){$\color{blue}X_2$}
  \end{overpic}
  \caption{
    The representation of oriented planes in the Blaschke cylinder model of Laguerre geometry.
    Left: A plane $\star X$ in Euclidean space $\eucl$.
    Middle: The dual point $X$ of the plane $\star X$ in dual Euclidean space $\eucl^*$
    and its two lifts $X_1, X_2$ to the Blaschke cylinder $\blac$.
    Right: Affine coordinates of $\RP^4$ for which the hyperplane $S_{\eucl^*}$ is the hyperplane at infinity and $\blac$ appears as a cylinder.
  }
  \label{fig:laguerre-geometry}
\end{figure}
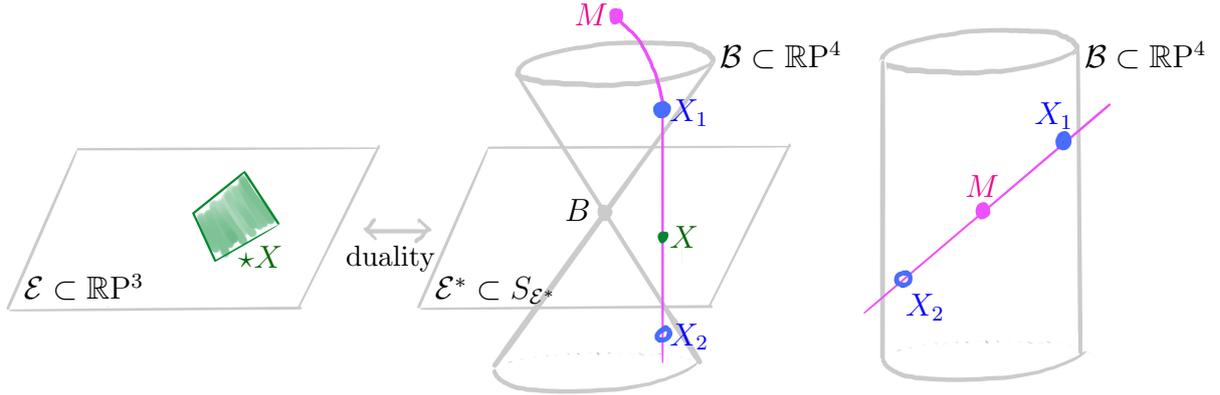

We now embed the dual Euclidean space $\eucl^*$ into $\RP^4$ in the following way.
Let us denote the apex of $\blac$ by $\PB$,
and let $\PM \in \blac^-$ be a point inside the Blaschke cylinder, that is a point with signature $\texttt{(-)}$.
Define $S_{\eucl^*} \coloneqq M^\pol$.
We identify
\[
  \eucl^* \cong S_{\eucl^*} \setminus \{ B \}.
\]
with the 3-dimensional dual Euclidean space.
The restriction of $\sca{\cdot,\cdot}_\blac$ to $S_{\eucl^*}$ has signature $\texttt{(+++0)}$
and defines the absolute quadric $\star \euclq$ of dual Euclidean geometry.

Consider the central projection to $S_{\eucl^*}$ with center $M$
\begin{align}
	\pi_{\eucl^*}&: \RP^4 \setminus \{ \PM \} \rightarrow S_{\eucl^*}, \quad X \mapsto  (X \vee \PM) \cap S_{\eucl^*}.
\end{align}
Its restriction to the Blaschke cylinder is a map
\[
  \blac \setminus \{\PB\} \rightarrow \eucl^*
\]
which is two-to-one,
that is a double cover of the dual Euclidean space $\eucl^*$.
Indeed, in Laguerre geometry, two different points $X_1, X_2 \in \blac$ on the Blaschke cylinder
that are mapped to the same point $X = \pi_{\eucl^*}(X_1) = \pi_{\eucl^*}(X_2)$
correspond to the two possible orientations of the plane $\star X \subset \eucl$.
This choice of orientation can be done in a consistent way.
If we denote the space of \emph{oriented Euclidean planes} by
\[
  \opl \coloneqq \text{Oriented planes}(\eucl),
\]
this gives rise to the map
\[
  \xi_\opl : \blac \setminus \{\PB\} \rightarrow \opl.
\]
The corresponding non-oriented planes are given by the map
\[
  \xi_\pl : \RP^4 \setminus \ell_\blac \rightarrow \pl,\quad
  X \mapsto \star\pi_{\eucl^*}(X) \subset \eucl,
\]
which, for later purposes, we define on the entire space except the \emph{axis} of the Blaschke cylinder
\[
  \ell_\blac \coloneqq \PB \vee \PM.
\]
In this way, every point in $\RP^4$ except for the line $\ell_\blac$ has an associated (non-oriented) plane in $\eucl$.
\begin{coordinates}
  By the previous choice of $\sca{\cdot,\cdot}_{\blac}$ we have
  \begin{align}
    \PB = [0,0,0,0,1].
  \end{align}
  We choose
  \[
    \PM = [0,0,0,1,0], \qquad
    S_{\eucl^*} = \set{[x] \in \RP^4}{ x_4 = 0}.
  \]
  The central projection to $S_{\eucl^*}$ is given by
  \begin{align}
    \pi_{\eucl^*}([x]) &= [x_1, x_2, x_3, 0, x_5].
  \end{align}
  A point $[u, 1, h] \in \blac$, with $u \in \S^2$, $h \in \R$ represents
  an oriented plane $\xi_\opl([u,1,h])$ in $\eucl$ given by the (non-oriented) plane
  \[
    \xi_\pl([u,1,h]) = \set{y \in \R^3}{\sca{u,y} + h = 0} \subset \R^3 \simeq \eucl,
  \]
  and orientation determined by the choice of the normal direction $u$.
  The point $[u, -1, h]$ represents the same plane with opposite orientation.
\end{coordinates}

\subsection*{Spheres in Laguerre geometry}
We denote the space of \emph{oriented Euclidean spheres} by
\[
  \osp \coloneqq \text{Oriented spheres}(\eucl).
\]
Let $H \subset \RP^4$ be a hyperplane that does not contain the apex $\PB$ of the Blaschke cylinder $\blac$,
that is a hyperplane of signature $\texttt{(+++-)}$.
Then the projection
\[
  \xi_\opl(H\cap \blac),
\]
is the 2-parameter family of oriented planes touching a fixed oriented sphere in $\eucl$, which we denote by
\[
  \xi_\osp(H).
\]
Thus, hyperplanes that do not contain $\PB$ correspond to oriented spheres in $\eucl$.
If the hyperplane $H$ contains the point $M$, then $\xi_\opl(H\cap \blac)$ consists of all oriented planes through a point in $\eucl$,
and thus $\xi_\osp(H)$ describes a \emph{null-sphere}, which has no orientation.

We now use this correspondence to embed the unit sphere $\unis$ into $\blac$.
Let $S_\unis \subset \RP^4$ be the hyperplane of signature $\texttt{(+++-)}$
through the point $M$ such that $\xi_\osp(S_\unis)$ is the null-sphere at the origin $\ori \in \eucl$.
Thus, every point in $X \in S_\unis \cap \blac$ corresponds to an oriented plane $\xi_\opl(X)$ containing $\ori$ and with normal vector in $\unis$.
In this way, we identify
\[
  \unis \cong S_\unis \cap \blac.
\]
Thus, if we consider the central projection to $S_\unis$ with center $B$
\[
  \pi_{\unis} : \RP^4 \setminus \{\PB\} \rightarrow S_\unis, \quad X \mapsto (X \vee \PB) \cap S_\unis,
\]
its restriction to the Blaschke cylinder is a map
\[
  \blac \setminus \{\PB\} \rightarrow \unis.
\]
And by the above identification, we interpret the image point $\pi_\unis(X)$ as the unit normal direction of the oriented plane $\xi_\opl(X)$,
which in particular determines the choice of orientation.
Consequently, we obtain a decomposition of a point $X \in \blac \setminus \{\PB\}$ into its two projections
\[
  \pi_{\eucl^*}(X), \qquad
  \pi_{\unis}(X),
\]
which correspond to the (non-oriented) plane $\xi_\pl(X) \subset \eucl$ and its unit normal direction,
and together can be identified with the oriented plane
\[
  \xi_\opl(X) = \left(\xi_\pl(X), \pi_{\unis}(X)\right).
\]
Note, that all points on a generator of $\blac$ are mapped to the same point under the map~$\pi_\unis$.
Thus, a generator corresponds to a family of parallel oriented planes with matching orientation.

\begin{coordinates}
  A hyperplane
  \[
    \star [c,r,1] = \set{[x] \in \RP^4}{c_1x_1 + c_2x_2 + c_3x_3 + rx_4 + x_5 = 0} \subset \RP^4,
  \]
  which does not contain the apex $B = [0,0,0,0,1]$,
  corresponds to an oriented sphere $\xi_\osp(\star [c,r,1])$ with center $c \in \R^3$ and signed radius $r \in \R$,
  where positive sign encodes outside orientation.

  The intersection with the Blaschke cylinder
  \[
    \star [c,r,1] \cap \blac = \set{[u,1,h] \in \blac}{\sca{c,u} + r + h = 0},
  \]
  corresponds to all oriented hyperplanes $\xi_\opl([u,1,h])$ in oriented contact with the oriented sphere $\xi_\osp(\star [c,r,1])$.
  In particular, the hyperplane
  \[
    S_{\unis} = \set{[x] \in \RP^4}{x_5 = 0} \ni M,
  \]
  corresponds to the null-sphere at the origin $\ori = (0,0,0) \in \eucl \cong \R^3$.
  
  Every point on 
  \[
    \unis \cong S_{\unis} \cap \blac = \set{[x] \in S_\opl}{x_1^2 + x_2 ^2 + x_3 ^3 = x_4^2},
  \]
  is identified with the unit normal direction of an oriented plane through the origin.

  The two projections of a point $[u,1,h] \in \blac$
  \[
    \pi_{\eucl^*}([u,1,h]) = [u, 0, h],\qquad
    \pi_{\unis}([u,1,h]) = [u, 1, 0],
  \]
  yield a point in dual Euclidean space $[u,h] \in \eucl^*$ and a unit normal direction $[u,1] \in \unis$.
\end{coordinates}

\begin{remark}
  In the same way that hyperplanes in $\RP^4 \setminus \{\PB\}$ can be identified with oriented spheres,
  2-dimensional planes in $\RP^4  \setminus \{\PB\}$ can be identified
  with \emph{oriented cones} (of revolution) in $\eucl$.
\end{remark}

\subsection*{Angled planes}

The \emph{outside} of the Blaschke cylinder is given by
\[
  \blac^+ = \set{[x] \in \RP^4}{ \sca{x,x}_{\blac} > 0}.
\]
For a point $X \in \blac^+$, the polar hyperplane $X^\pol$ is a 3-dimensional space of signature $\texttt{(++-0)}$,
which therefore must contain the apex $\PB \in \blac$.
Thus, unlike in Möbius geometry, there is no identification of hyperplanar sections of $\blac$ with points in $\blac^+$ via polarity.

Still, for our purposes it is practical to have a geometric description of the points in $\blac^+$,
which we present in the following.
The two projections of a point $X \in \blac^+$
\[
  \pi_{\eucl^*}(X), \qquad
  \pi_{\unis}(X),
\]
yield a point in dual Euclidean space,
that corresponds to a (non-oriented) Euclidean plane $\xi_\pl(X)$,
and a point outside the unit sphere, to which we associate the corresponding polar circle on $\unis$,
which we call the \emph{normal circle} $\xi_\nci(X)$.
It is given by the map
\[
  \xi_\nci : \blac^+ \rightarrow \circlesof{\unis},\qquad
  X \mapsto \pi_{\unis}(X)^\perp \cap \unis.
\]

Note that we can define $\xi_\nci$ analogously for $X\in \blac$.
In this case the normal circle $\xi_\nci(X)$ is a circle of radius 0 coinciding with the point $\pi_{\unis}(X)$.
Moreover, we identified a point $X \in \blac$ with an oriented plane given by
$\left(\xi_\pl(X), \pi_{\unis}(X)\right)$.
Analogously, we identify a point $X \in \blac^+$ with the pair $\left(\xi_\pl(X), \xi_{\nci}(X)\right)$,
which we call an \emph{angled Euclidean plane} (see Figure~\ref{fig:angled-planes}).
Thus, we introduce the space of angled Euclidean planes, as the space of (non-oriented) Euclidean planes together with a normal circle, with axis orthogonal to the plane, and denote it by
\[
  \apl \coloneqq \text{Angled planes}(\eucl).
\]
This gives rise to the map
\[
  \xi_\apl : \blac^+ \rightarrow \apl,\qquad
  X \mapsto \left(\xi_\pl(X), \xi_{\nci}(X)\right).
\]

In analogy with Möbius geometry, we consider the plane $\xi_\pl(X)$ to be the \emph{center} of the angled plane,
while the normal circle $\xi_\nci(X)$ -- or the corresponding intersection angle -- corresponds to the radius of a sphere.

We may also choose a normal vector $u$ (an orientation) of the center plane $\xi_\pl(X)$. Then the polar complement of $X$ intersected with $\blac\setminus\{\PB\}$ corresponds to all oriented planes that have a normal vector with a fixed angle $\varphi$ to $u$. For the other orientation of $\xi_\pl(X)$ the fixed angle is $\pi - \varphi$.

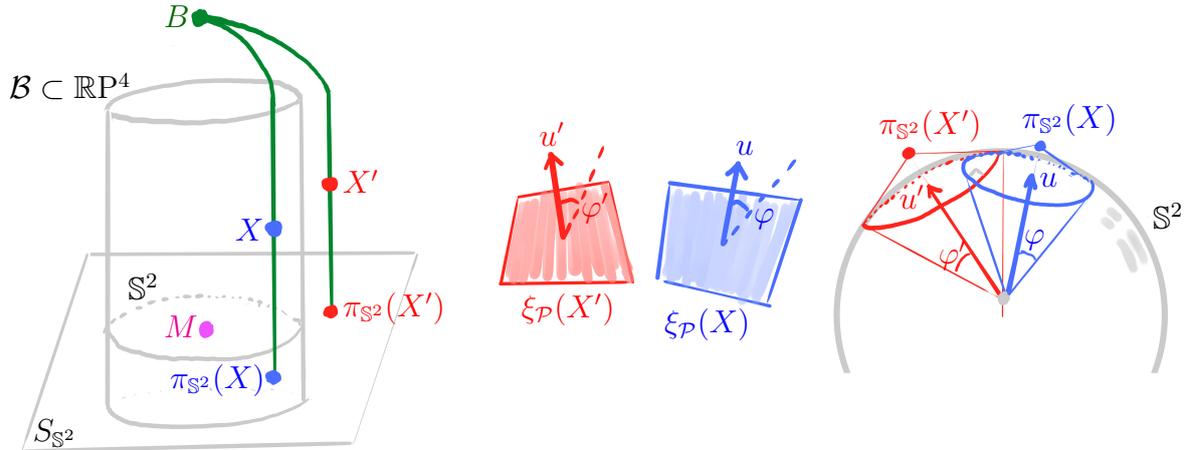
\begin{figure}[H]
  \centering
  \begin{overpic}[width=0.98\textwidth]{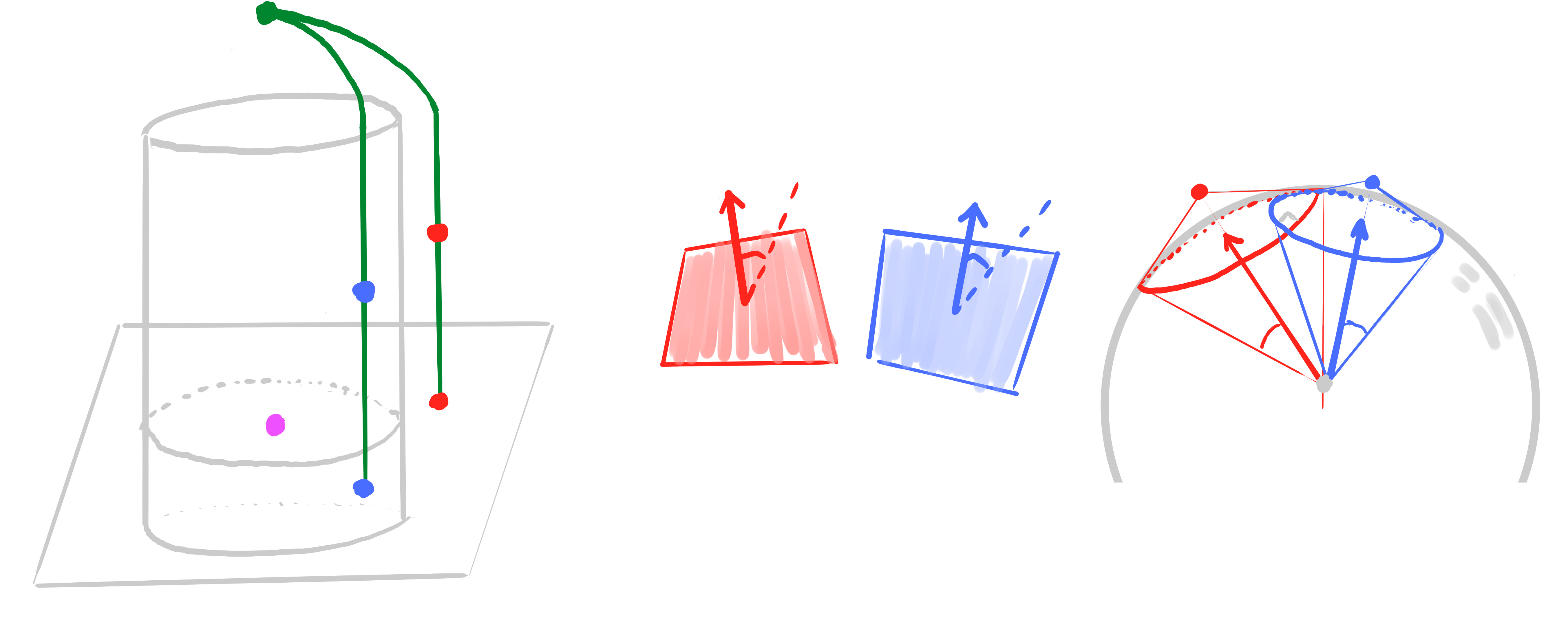}
    \put(1,32){$\blac\subset\RP^4$}
    \put(3,3){$S_\unis$}
    \put(11,15){$\unis$}
    \put(14,38){$\color{darkgreen}\PB$}
    \put(14,11.8){$\color{magenta}\PM$}
    \put(20,20){$\color{blue}X$}
    \put(29,24){$\color{red}X'$}
    \put(14.5,7.5){$\color{blue}\pi_\unis(X)$}
    \put(29,13.6){$\color{red}\pi_\unis(X')$}
    \put(56,12.2){$\color{blue}\xi_\pl(X)$}
    \put(44,13.5){$\color{red}\xi_\pl(X')$}
    \put(62,27.2){$\color{blue}u$}
    \put(45.5,28){$\color{red}u'$}
    \put(63.5,21.5){$\color{blue}\varphi$}
    \put(49,22){$\color{red}\varphi'$}
    \put(97,21){$\unis$}
    \put(86,29.5){$\color{blue}\pi_\unis(X)$}
    \put(74,29){$\color{red}\pi_\unis(X')$}
    \put(87.5,24.5){$\color{blue}u$}
    \put(75.7,22.7){$\color{red}u'$}
    \put(86.2,19.7){$\color{blue}\varphi$}
    \put(79,18){$\color{red}\varphi'$}
  \end{overpic}
  \caption{
    Angled planes in Laguerre geometry.
    A point $X \in \blac^+$ outside the Blaschke cylinder corresponds to an angled plane $\xi_\apl(X)$, which can be represented by the center plane $\xi_\pl(X)$ and the normal circle $\xi_\nci(X)$, or equivalently, by the choice of a unit normal vector $u$ and an angle $\varphi$.
    Polarity of points $X, X' \in \blac^+$ corresponds to the orthogonality of the normal circles $\xi_\nci(X), \xi_\nci(X')$.
  }
  \label{fig:angled-planes}
\end{figure}

Note that every point in $\RP^4$ is polar to $\PB$ and $\pi_{\S^2}$ is a central projection with center $\PB$.
Thus, for $X, X' \in \RP^4 \setminus \{\PB\}$,
\[
  X \perp X' \quad \Leftrightarrow \quad \pi_{\S^2}(X) \perp \pi_{\S^2}(X').
\]
In particular, polarity of two points $X, X' \in \blac^+$ outside the Blaschke cylinder
corresponds to the orthogonality of the normal circles of the two corresponding angled planes (see Figure~\ref{fig:angled-planes}).

\begin{proposition}
  \label{prop:angled-planes-polarity}
  Two points $X, X' \in \blac^+$ are polar with respect to $\blac$ if and only if
  the two normal circles $\xi_\nci(X)$ and $\xi_\nci(X')$ are orthogonal.
\end{proposition}

Another way to obtain a geometric description of an angled plane $\xi_\apl(X)$ with $X \in \blac^+$
is to consider the three-parameter family of hyperplanes that contain $X$.
The corresponding oriented spheres intersect the (non-oriented) Euclidean plane $\xi_\pl(X) \subset \eucl$
in a constant angle determined by the normal circle $\xi_\nci(X)$ in the following way.

\begin{proposition}
  Let $X \in \blac^+$ be a point outside the Blaschke cylinder and $H$ a hyperplane of $\RP^4$ that contains $X$ but not the apex $\PB$.
  Then the normal vector of the oriented sphere $\xi_\osp(H)$ at a point in the intersection with the (non-oriented) Euclidean plane $\xi_\pl(X)$
  lies on the normal circle $\xi_\nci(X)$.
\end{proposition}

\begin{proof}
  Consider a point $P \in \xi_\osp(H) \cap \xi_\pl(X)$,
  and the oriented plane $\vec E$ through $P$ in oriented contact with $\xi_\osp(H)$.

  Let $\hat{E} \in \blac$ be the point such that $\xi_\opl(\hat{E}) = \vec E$,
  and let $\hat{P}$ be the hyperplane in $\RP^4$ such that $\xi_\osp(\hat{P}) = P$.

  That $\vec E$ is the oriented tangent plane of $\xi_\osp(H)$ in $P$ means that $\hat{P} \cap \blac$ and $H \cap \blac$ are tangent in $\hat{E}$,
  and in particular,
  \[
    \hat{P} \cap H \subset \hat{E}^\perp.
  \]
  Because $\hat{P}$ corresponds to a point, $\hat{P}$ must contain $\PM$.
  And since $P$ is on $\xi_\pl(X)$, $\hat{P}$ must contain $X$.
  Thus, $X \in \hat{P} \cap H$, and therefore
  \[
    X \perp \hat{E},
  \]
  which is equivalent to
  \[
    \pi_\unis(X) \perp \pi_\unis(\hat{E}).
  \]
  The normal vector $\pi_{\unis}(\hat{E})$ of $\vec E$ is the normal vector of $\xi_\osp(H)$ at $P$.
  Thus, $\pi_{\unis}(\hat{E})$ lies on the normal circle $\xi_\ci(X)$.
\end{proof}

As a result of the proposition, we obtain a correspondence between an angled plane $\xi_\apl(X)$ for $X \in \blac^+$
and the set of all oriented spheres that intersect the plane $\xi_\pl(X)$ in a constant angle,
which is determined by the normal circle $\xi_\nci(X)$.

\begin{coordinates}
  A point $[u, \cos\varphi, h] \in \blac^+$ with $u \in \S^2$, $\varphi \in [0,\pi)$, $h \in \R$
  represents an angled plane $\xi_\apl([u, \cos\varphi, h])$ with center plane
  \[
    \xi_{\pl}([u, \cos\varphi, h]) = \xi_{\pl}([u, \pm 1, h]),
  \]
  and normal circle
  \[
    \xi_{\nci}([u, \cos\varphi, h]) = [u, \cos\varphi, 0]^\perp \cap \S^2,
  \]
  which is a circle on $\S^2$ with spherical radius $\varphi$. Moreover, by choosing one of the oriented planes $\xi_{\opl}([u, \pm 1, h])$, we see that the polar complement of $[u, \cos\varphi, h]$ intersected with $\blac \setminus \{\PB\}$ corresponds to all oriented planes
	\begin{align}
		\set{[\tilde u, 1, \tilde h] \in \blac}{ \sca{u, \tilde u} = \cos\varphi}.
	\end{align}  
	The normal vectors of these oriented planes have a fixed angle $\varphi$ to the normal vector of the oriented plane $\xi_{\opl}([u, + 1, h])$, or equivalently a fixed angle $\pi - \varphi$ to the normal vector of the oriented plane $\xi_{\opl}([u, - 1, h])$.

  Note that in affine coordinates $x_5 = 1$, the points corresponding to the angled planes with a fixed angle $\varphi$
  lie on a cylinder concentric to $\blac$.

  A hyperplane $\star [c,r,1]$ containing the point $[u, \cos\varphi, h]$ must satisfy
  \[
    \sca{c, u} + r\cos\varphi + h = 0,
  \]
  and thus corresponds to a sphere with normal vectors that have angle $\varphi$ with $\xi_{\opl}([u, + 1, h])$.
  Two points $[u, \cos\varphi, h], [u', \cos\varphi', h'] \in \blac^+$ are polar if and only if
  \[
    \sca{u, u'} = \cos\varphi \cos\varphi',
  \]
  which by the spherical Pythagorean theorem is equivalent to the orthogonality of the two spherical circles
  with center $u, u'$ and radii $\varphi, \varphi'$.
\end{coordinates}

For a point $X \in \blac^-$ inside the Blaschke cylinder,
the polar hyperplane does not intersect $\blac$ in any real points apart from $\PB$.
Yet the projection $\pi_{\eucl^*}(X) \in \eucl^*$ still yields a real point in dual Euclidean space.
The point $X$, can be interpreted as an \emph{imaginary angled plane} with real center and imaginary normal circle. The imaginary normal circle is an imaginary circle in $\S^2$ analogous to imaginary spheres as discussed in Section~\ref{sec:moebius}.

\begin{remark}
  Let $X \in \blac^- \setminus \ell_\blac$ be a point inside the Blaschke cylinder
  and $H$ a hyperplane of $\RP^4$ that contains $X$ but not the apex $\PB$.
  Let $E = \xi_{\pl}(X)$ and $\vec S = \xi_\osp(H)$,
  and consider the line $L$ through the center of $\vec S$ perpendicular to $E$.
  Let $\cone$ be the oriented cone with apex $E \cap L$ in oriented contact with $\vec S$.
  We call the (signed) opening angle of $\cone$ the \emph{widest angle} of $\vec S$ on $E$.

  Every hyperplane $H' \subset \RP^4 \setminus \{\PB\}$ that contains $X$
  corresponds to an oriented sphere $\xi_\osp(H')$ that has the same widest angle on $E$ as $\vec S = \xi_\osp(H)$.
  Thus, there is a correspondence between an imaginary angled plane $\xi_\apl(X)$ and the set of all oriented spheres
  with constant widest angle on $\xi_\pl(X)$.
\end{remark}

If we extend the set of angled planes $\apl$ by the set of oriented planes and imaginary angled planes,
the map $\xi_\apl$ can be extended to the entire space excluding the axis of the Blaschke cylinder $\ell_\blac$
\[
  \xi_\apl : \RP^4\setminus \ell_\blac \rightarrow \apl.
\]
The orthogonality of angled planes is extended to imaginary angled planes, by the orthogonality of their (possibly imaginary) normal circles, and in this way Proposition~\ref{prop:angled-planes-polarity} generalizes in the following way.
\begin{proposition}
  \label{prop:orthogonal-circles-ext}
  Two points $X, X' \in \RP^4\setminus \ell_\blac$ are polar with respect to $\blac$ if and only if
  the two corresponding (possibly imaginary) normal circles $\xi_\nci(X)$ and $\xi_\nci(X')$ are orthogonal.
\end{proposition}

\begin{remark}
  Let $X \in \blac^+, X'\in \blac^-$.
  Then the circle $\xi_\nci(X)$ is orthogonal to the imaginary normal circle $\xi_\nci(X')$
  if and only if $\xi_\nci(X)$ intersects the real representative of $\xi_\nci(X')$ in opposite points on $\xi_\nci(X')$
  (see also Remark~\ref{rem:orthogonal-imaginary-sphere}).
\end{remark}

\begin{coordinates}
  A point $[u, (\sin\alpha)^{-1}, h] \in \blac^-$ with $u \in \S^2$, $\alpha \in (-\tfrac{\pi}{2},0) \cup (0,\tfrac{\pi}{2})$, $h \in \R$
  represents an imaginary angled plane $\xi_\apl([u, (\sin\alpha)^{-1}, h])$
  with real center plane and imaginary normal circle.
  
  A hyperplane $\star [c,r,1]$ containing the point $[u, (\sin\alpha)^{-1}, h]$ must satisfy
  \[
    \sca{c, u} + \frac{r}{\sin\alpha} + h = 0,
  \]
  and thus corresponds to a sphere of \emph{widest angle} $\alpha$ from the center plane.
\end{coordinates}

\section{Laguerre lift of orthogonal bi*nets}
\label{sec:laguerrelift}

We use the projective model of Laguerre geometry $\blac \subset \RP^4$,
and embed the dual Euclidean space $\eucl^* \subset S_{\eucl^*} \subset \RP^4$ as described in Section~\ref{sec:laguerre}.

We identified points in $\RP^4 \setminus \ell_\blac$ with angled planes in $\eucl$
by means of the map $\xi_\apl$.
Thus, the points of a binet in $\RP^4$ can be represented in terms of angled planes.
The polarity of two points representing two angled planes can be
described in terms of the orthogonality of their normal circles,
which are given by the map $\xi_\nci$.
This leads to an \emph{orthogonal circle representation} of polar binets in Laguerre geometry (see Figure~\ref{fig:principal-cones}, left).
\begin{lemma}
  Let $b_\blac : D \rightarrow \RP^4 \setminus \ell_\blac$ be a polar binet
  with respect to the Blaschke cylinder $\blac \subset \RP^4$.
  Let
  \[
    b_\apl \coloneqq \xi_\apl \circ b_\blac,
    \qquad
    b_\nci \coloneqq \xi_\nci \circ b_\blac,
  \]
  be the binet of its corresponding angled planes
  and its orthogonal circle representation respectively.
  Let
  \[
    b_{\eucl^*} \coloneqq \pi_{\eucl^*} \circ b_\blac,
  \]
  be the projection of $b_\blac$ to the dual Euclidean space $\eucl^*$.
  Then
  \begin{enumerate}
  \item
    the two circles $b_\nci(d)$ and $b_\nci(d')$ intersect orthogonally for all incident $d, d' \in D$,
  \item
    $\star b_{\eucl^*}(d)$ is the center plane of $b_\apl(d)$ for all $d \in D$.\qedhere
  \end{enumerate}
\end{lemma}
\begin{proof}\
  \begin{enumerate}
  \item
    Follows from Proposition~\ref{prop:orthogonal-circles-ext}.
  \item
    As discussed in Section~\ref{sec:laguerre}, for $X \in \RP^4 \setminus \ell_\blac$ the point $\pi_{\eucl^*}(X)$ is the dual of the center plane of the angled plane $\xi_\apl(X)$.\qedhere
  \end{enumerate}
\end{proof}

Thus, the orthogonal circle representation of a polar binet in Laguerre geometry
can be viewed as an $\unis$ analogue of the orthogonal sphere representation in Möbius geometry.
Note however, that the orthogonal circle representation does not
contain the full information of the polar binet.

Instead of the normal circles of the angled planes, we can also look at the projection $\pi_\unis$,
which yields the poles of the planes of the normal circles with respect to $\unis$.
\begin{lemma}
  \label{lem:laguerre-lift-projection}
  Let $b_\blac : D \rightarrow \RP^4 \setminus \ell_\blac$ be a polar binet
  with respect to the Blaschke cylinder $\blac \subset \RP^4$.
  Then the dual of its projection
  \[
    b:  D \rightarrow \pl, \quad b = \star (\pi_{\eucl^*} \circ b_\blac),
  \]
  is an orthogonal bi*net,
  and the projection
  \[
    n \coloneqq \pi_\unis \circ b_\blac,
  \]
  is a normal net of $b$.
\end{lemma}

\begin{proof}
	We prove that $n$ is a normal binet of $b$ first. Note that the projection $\pi_\unis$ preserves polarity, therefore $n$ is a polar binet. Moreover, by definition  for $d\in D$ the plane $b(d)$ is orthogonal to the normal $n(d)$. Therefore the conditions of Definition~\ref{def:normal-binet} are satisfied. Additionally, using the arguments of the ($\Rightarrow$) direction in the proof of Theorem~\ref{th:orthonormals}, $b$ is indeed an orthogonal bi*net.
\end{proof}

Analogous to the corresponding relation in Möbius geometry,
we obtain a close relation between polar binets in Laguerre geometry and orthogonal bi*nets in $\eucl$ (see Figure~\ref{fig:laguerre-lift}).

\begin{figure}[H]
  \centering
  \begin{overpic}[width=0.4\textwidth,trim={0 220 0 0},clip]{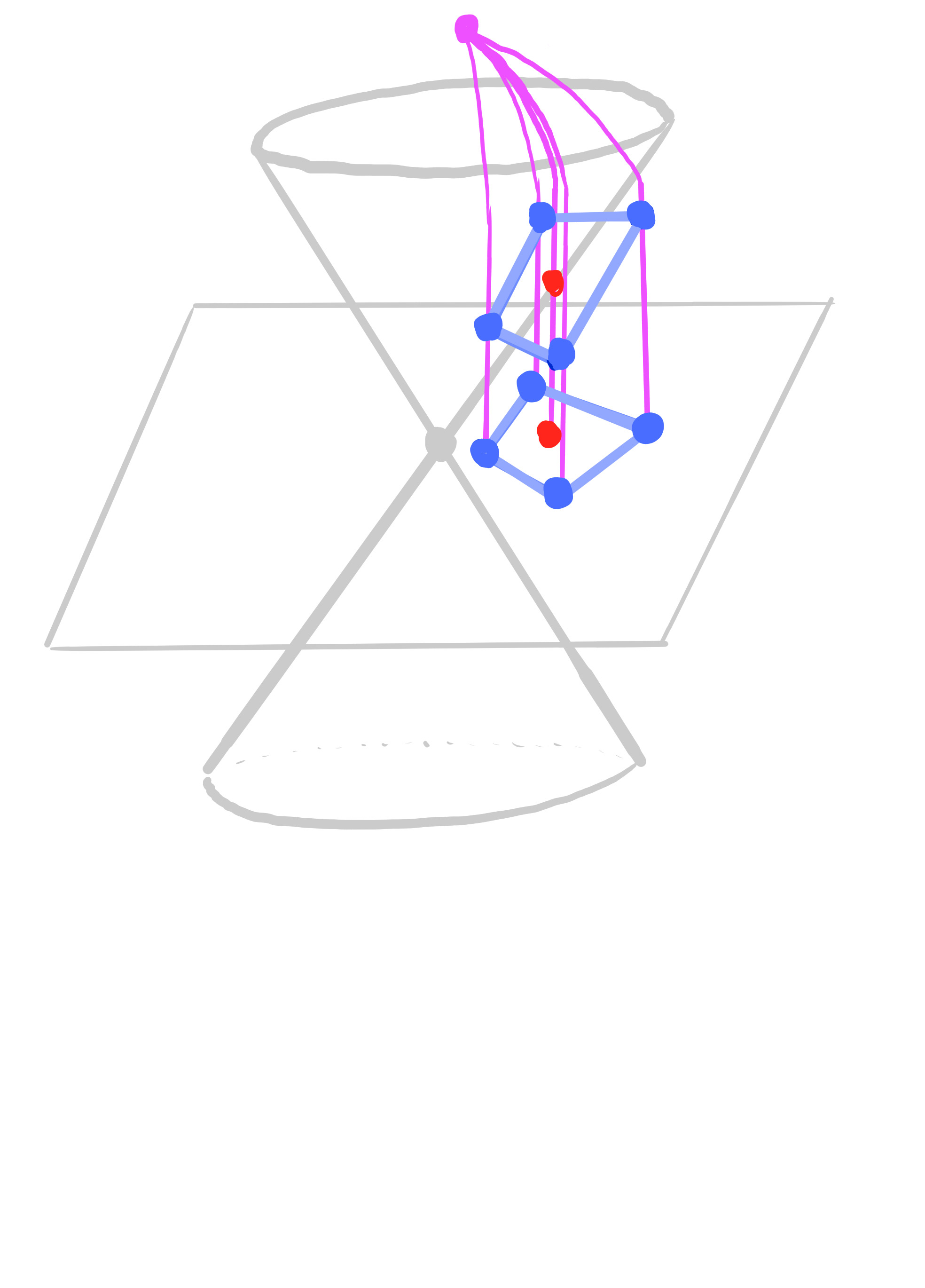}
    \put(71,9){$\blac \subset \RP^4$}
    \put(8,25){$\eucl^* \subset S_{\eucl^*}$}
    \put(42,88){$\color{magenta}\PM$}
    \put(72,44){$\color{blue}\star b(d)$}
    \put(72,67){$\color{blue}b_\blac(d)$}
  \end{overpic}
  \caption{
    Laguerre lift $b_\blac$ of a bi*net $b$.
    A point $\star b(d) \in \eucl^*$ is lifted to a point $b_\blac(d) \in \RP^4$
    on the line $\star b(d) \vee \PM$.
  }
  \label{fig:laguerre-lift}
\end{figure}
\begin{definition}[Laguerre lift]
  Let $b: D \rightarrow \pl$ be a bi*net.
  Then a binet $b_\blac: D \rightarrow \RP^4 \setminus \ell_\blac$
  is called a \emph{Laguerre lift} of $b$ if
  \begin{enumerate}
  \item
    $b_\blac$ is a polar binet with respect to $\blac$,
  \item
    and $b$ is the dual of the projection of $b_\blac$, that is, $\star (\pi_{\eucl^*} \circ b_\blac) = b$.\qedhere
  \end{enumerate}
\end{definition}

Lemma~\ref{lem:laguerre-lift-projection} implies that the condition that $b$
is an orthogonal bi*net is necessary for the existence of a Laguerre lift $b_\blac$.
The following theorem shows that it is also sufficient,
and thus guarantees the existence of a Laguerre lift for an orthogonal bi*net.

\begin{theorem}
  \label{thm:laguerre-lift}
  Let $b : D \rightarrow \pl$ be a regular bi*net.
  Then a Laguerre lift $b_\blac : D \rightarrow \RP^4 \setminus \ell_\blac$ exists
  if and only if $b$ is a regular orthogonal bi*net.
\end{theorem}
\begin{proof}\ \linebreak
  ($\Rightarrow$) 
  Lemma~\ref{lem:laguerre-lift-projection} implies that $b$ is an orthogonal bi*net.
  \newline
  ($\Leftarrow$) 
  By Theorem~\ref{th:orthonormals}, the regular orthogonal bi*net $b$ admits a normal binet $n$.
  The Laguerre lift $b_\blac$ is then uniquely determined by the two conditions
  $
  \pi_\unis \circ b_\blac = n
  $
  and
  $
  \pi_{\eucl^*} \circ b_\blac  = \star b
  $.
  More specifically,
  \[
    b_\blac(d) = (\PB \vee n(d)) \cap (\PM \vee \star b(d)).\qedhere
  \]
\end{proof}

\begin{remark}
  It is also possible to prove Theorem~\ref{thm:laguerre-lift}
  in a way analogous to Theorem \ref{th:mobiusortho}.
  Yet, once Theorem~\ref{th:orthonormals} on normal binets is established,
  the given proof is even simpler.
\end{remark}

The proof of Theorem \ref{thm:laguerre-lift} uses Theorem \ref{th:orthonormals}, which shows that for each orthogonal bi*net there is a 1-parameter family of normal binets. Hence, there is a 1-parameter family of Laguerre lifts for each orthogonal bi*net.

\begin{coordinates}
  Let $b: D \rightarrow \pl$ be an orthogonal bi*net
  let $n: D \rightarrow \R^3 \setminus \{\ori\}$ be a normal binet,
  and $u: D \rightarrow \unis$ be a unit-normal binet (see Coordinates~\ref{coo:normallift}).
  Thus,
  \[
    \abs{u(d)} = 1,
    \quad\text{and}\quad
    n(d) = \frac{u(d)}{\sigma(d)}
  \]
  with some function $\sigma(d) \neq 1$.
  Let the function $h: D \rightarrow \R$ be determined by
  \[
    b(d) = \set{x \in \eucl}{\sca{u(d), x} + h(d) = 0}.
  \]
  Then the Laguerre lift $b_\blac$ of $b$ is given by
  \[
    b_\blac(d) = [u(d), \sigma(d), h(d)] = [n(d), 1, \tfrac{h(d)}{\sigma(d)}].
  \]  
  The planes are recovered by the projection
  \[
    \star b(d) \cong (\pi_{\eucl^*} \circ b_\blac)(d) = [u(d), 0 ,h(d)],
  \]
  and the normal binet by
  \[
    n(d) \cong (\pi_{\unis} \circ b_\blac)(d) = [u(d), \sigma(d), 0] = [n(d), 1, 0].
  \]
  Moreover, in the case $\abs{\sigma(d)} < 1$ of real angled planes, their angle $\varphi(d)$ --
  which is the spherical radius of the normal circles -- is given by $\sigma(d) = \cos\varphi(d)$.
  Note that -- using the function $\sigma$ -- a different choice for the Laguerre lift is represented by the simple transformation
  \[
    \begin{aligned}
      &\sigma(v) \rightarrow \alpha\sigma(v), &&v \in V,\\
      &\sigma(f) \rightarrow \alpha^{-1}\sigma(f),  &&f \in F,
    \end{aligned}
  \]
  with $\alpha \in \R\setminus\{0\}$ (see Coordinates~\ref{coo:normallift}).
\end{coordinates}

\begin{lemma}
	If $b$ is a regular orthogonal bi*net, then every Laguerre lift $b_\blac$ of $b$ is a regular binet.
\end{lemma}
\proof{	
	Due to Lemma~\ref{lem:regnormalnet}, any normal net $n$ of $b$ is a regular binet. For every $d\in D$, $b_\blac(d)$ is on the line $B\vee n(d)$, but not equal to $B$. Therefore if $n(d)$, $n(d')$ for $d,d'\in D$ do not coincide, neither do $b_\blac(d)$, $b_\blac(d')$.
\qed
}

\subsection*{Invariance}
Let $T$ be a Laguerre transformation and $b$ an orthogonal bi*net. A priori, it is not clear how $T$ should act on $b$, since $b$ maps to $\pl$ -- not to $\opl$. Moreover, no matter how we choose orientations for the planes of $b$, the Laguerre transformation of the so oriented planes will in general not be an orthogonal bi*net. However, if we choose a Laguerre lift $b_\blac$ we may apply the corresponding projective transformation $\tilde T$ of $\RP^4$ to $b_\blac$. In this case, $\tilde T \circ b_\blac$ is again a polar binet and therefore the projection $\star (\pi_{\eucl^*} \circ \tilde T \circ b_\blac)$ is again an orthogonal bi*net. Since each orthogonal bi*net has a 1-parameter family of Laguerre lifts, the action of a Laguerre transformation on an orthogonal bi*net is not unique. However, if we consider an orthogonal bi*net $b$ together with a normal binet $n$, which uniquely determines the Laguerre lift $b_\blac$, the action of Laguerre transformations is unique.

\section{Line bicongruences} \label{sec:linebi}

In preparation of the following sections,
we recall the notion of (discrete) line congruences \cite{dsmlinecongruence,ddgbook},
and then introduce a corresponding notion of line bicongruences.
The definition of line congruences that we use for our purposes here,
is a discretization of torsal parametrizations of smooth line congruences (two-parameter families of lines) \cite{pwshading}.
\begin{definition}[Line congruence]
  A \emph{line congruence} is a map $\ell: V \rightarrow \linesof{\RP^n}$ such that the lines of adjacent vertices intersect in a point.
\end{definition}
\begin{remark}\label{rem:focalcongruence}
  For each of the two directions of $V = \Z^2$, the points of intersection of the lines of $\ell$ form a conjugate net,
  called a \emph{focal net} of $\ell$.
\end{remark}
The identification of the faces $F$ with the vertices of the dual graph (see Remark~\ref{rem:dual-graph})
yields an analogous definition for line congruences on $F$.
Thus, we define line bicongruences on $D$ as pairs of line congruences on $V$ and $F$.
\begin{definition}[Line bicongruence]
  A \emph{line bicongruence} is a map $\ell: D \rightarrow \linesof{\RP^n}$ such that the restrictions to $V$ and $F$ are line congruences.
\end{definition}
\begin{remark}
  \label{rem:focal-binets}
  For each direction, the two focal nets (see Remark~\ref{rem:focalcongruence}) of a line bicongruence restricted to $V$ and restricted to $F$ respectively,
  together form a conjugate binet, which we call a \emph{focal binet} of $\ell$.
\end{remark}
In analogy to the definition of polar binets, we introduce polar line bicongruences
as line bicongruences in which incident lines are polar with respect to a given quadric.
\begin{definition}[Polar line bicongruence]
  Let $Q \in \RP^n$ be a quadric.
  A \emph{polar line bicongruence}  is a line bicongruence $\ell: D \rightarrow \linesof{\RP^n}$, such that
  \[
    \ell(d) \pol \ell(d') \qquad   \text{for all incident}~ d,d'\in D.\qedhere
  \]
\end{definition}

\section{Principal binets}  \label{sec:principal}
Away from umbilic points, a parametrization of a smooth surface is a (principal) curvature line paramatrization
if it is a conjugate line paramatrization and an orthogonal parametrization.
We combine the conditions of conjugate binets and orthogonal binets to define principal binets.
Recall that every regular conjugate binet $b$ defines an associated conjugate bi*net $\square b$ (see Section~\ref{sec:conjugate-bi-star-nets}).
\begin{definition}[Principal binets]\
  \nobreakpar
  \begin{enumerate}
  \item
    A \emph{principal binet} is a binet $b: D \rightarrow \eucl$ that is both conjugate and orthogonal.
  \item
    A principal binet $b$ is \emph{regular} if $b$ is a regular binet and $\square b$ is a regular bi*net.	\qedhere
  \end{enumerate}
\end{definition}

Similarly, we combine the conditions of conjugate bi*nets and orthogonal bi*nets to define principal bi*nets.
Recall that every regular conjugate bi*net $b$ defines an associated binet $\square^* b$.
\begin{definition}[Principal bi*nets]\
  \nobreakpar
  \begin{enumerate}
  \item
    A \emph{principal bi*net} is a binet $b: D \rightarrow \pl$ that is both conjugate and orthogonal.
  \item
    A principal bi*net $b$ is \emph{regular} if $b$ is a regular bi*net and $\square^* b$ is a regular binet.\qedhere
  \end{enumerate}
\end{definition}

By imposing regularity on both $b$ and $\square b$,
the $\square$-operator yields a close relation between principal binets and principal bi*nets.
This relation corresponds to the fact that in the smooth case
a conjugate line parametrization (second fundamental from diagonal)
is orthogonal (first fundamental form diagonal)
if and only if it is Gauß-orthogonal (third fundamental form diagonal).
\begin{lemma}\
  \nobreakpar
  \begin{enumerate}
  \item
    If $b$ is a regular principal binet then $\square b$ is a regular principal bi*net.
  \item
    If $b$ is a regular principal bi*net then $\square^*b$ is a regular principal binet.\qedhere
  \end{enumerate}
\end{lemma}
\begin{proof}\
  \begin{enumerate}
  \item
    Assume $b$ is a principal binet and let $(v,f,v',f') \in C$ be a cross. By definition of $\square$ we have that 
    \begin{align}
      \square b(v) \cap \square b(v') &= b(f) \vee b(f'),\\
      \square b(f) \cap \square b(f') &= b(v) \vee b(v').
    \end{align}
    Therefore the orthogonality conditions of $b$ and $\square b$ coincide.
    Note that the regularity of $b$ and $\square b$ ensures that both sides of these two equations are well-defined. 
    Moreover, in the regular case $\square^* \circ \square = \mbox{id}$, therefore the regularity of $b^*$ is also ensured.
  \item is proven analogously.\qedhere
  \end{enumerate}
\end{proof}

As a result of this lemma, principal binets and principal bi*nets always come in pairs~$(b, \square b)$. 

\begin{remark}
  \label{rem:parallel-normal-binet}
  Consider the normal binet $n$ (see Section~\ref{sec:normalbinets}) of a principal binet $b$ and recall that,
  due to Lemma~\ref{lem:normalconjugate}, normal binets are conjugate binets.
  Since  $n(d)$ is orthogonal to both $\square b(d)$ and $\square n(d)$, the two planes $\square b(d)$ and $\square n(d)$ are parallel.
  For a pair of conjugate binets, the condition of parallel corresponding faces is equivalent the condition of parallel corresponding edges.
  In fact, it is not hard to see that one could replace Condition~\ref{def:normal-binet1} in Definition~\ref{def:normal-binet}
  by requiring parallel edges of $b$ and $n$, and polarity for $n$ for some initial incident $d_0, d'_0 \in D$.
  With this parallel normal binet, the curvature theory for discrete surfaces based on mesh parallelity is applicable to principal binets \cite{bpwcurvature}.
\end{remark}

\begin{remark}\label{rem:minimal}
  Now that we have introduced principal binets, let us continue to discuss (see Remark~\ref{rem:koebe}) how normal binets generalize pairs of primal and dual Koebe polyhedra (see Figure~\ref{fig:examples}, left/middle).
  Recall that Koebe polyhedra are conjugate nets with edges tangent to the unit sphere. The authors in \cite{bhssminimal} consider the centers of certain touching spheres to constitute an \emph{S-minimal net}, which is obtained as Christoffel dual of the primal Koebe polyhedron. Moreover, the S-minimal net is a conjugate net with face-normals given by the dual Koebe polyhedron. Additionally, each face of an S-minimal net comes with a circle that is orthogonal to the incident spheres. It was also shown that the centers of these circles constitute a conjugate net with dual combinatorics. The face-normals of these circle center nets are given by the primal Koebe polyhedron. Moreover, the edges of the S-minimal net and the circle center net are orthogonal. Thus the combination of the S-minimal net and the circle center net is a principal binet with normal binet given by the pair of the primal and dual Koebe polyhedra.
	
	There is also a second way to construct a principal binet from a pair of primal and dual Koebe polyhedra, by applying the Christoffel dualization construction of \cite{bhssminimal} twice: once for the primal and once for the dual Koebe polyhedron. This construction has the advantage of being more symmetric. However, this ``twin'' construction has not been considered in the literature to date (see Figure~\ref{fig:examples}, right).
\end{remark}

\subsection*{Normal bicongruences}

At every vertex or face $d \in D$ of a conjugate binet $b : D \rightarrow \eucl$
a unique normal line $N(d)$ is defined by the two conditions
\[
  b(d) \in N(d)
  \quad\text{and}\quad
  \square b(d) \orth N(d).
\]
Now assume that $b$ is a principal binet and consider a cross $(v,f,v',f') \in C$.
Since $b(v) \vee b(v')$ is orthogonal to $b(f) \vee b(f')$,
the two lines $N(v)$ and $N(v')$ are in the plane that is orthogonal to $b(f) \vee b(f')$
and that contains $b(v) \vee b(v')$.
Thus, $N(v)$ and $N(v')$ intersect, and $N$ defines a line bicongruence.
Consequently, along a parameter line of $\Z^2$, consecutive lines of a line bicongruence intersect, so that the lines belonging to such a parameter line may be interpreted as a discretization of a developable surface. The intersection points along this parameter line correspond to the line of striction.
Thus, the fact that $N$ is a line bicongruence discretizes the fact that the normal lines along a curvature line of a smooth surface
trace out a developable surface.
\begin{definition}[Normal bicongruence]\label{def:normalbicongruence}
  Let $b: D \rightarrow \eucl$ be a regular conjugate binet.
  A \emph{normal bicongruence} $N: D\rightarrow \linesof{\eucl}$ of $b$ is a line bicongruence,
  such that
  \[
    b(d) \in N(d)
    \quad\text{and}\quad
    \square b(d) \orth N(d)
    \quad
    \text{for all}~
    d \in D.\qedhere
  \]
\end{definition}
We have seen that principal binets define a normal bicongruence.
In fact, the existence of a normal bicongruence characterizes principal binets.
\begin{theorem}
  Let $b: D \rightarrow \eucl$ be a regular conjugate binet.
  Then $b$ has a normal bicongruence $N: D\rightarrow \linesof{\eucl}$
  if and only if $b$ is a principal binet.
\end{theorem}
\begin{proof}
  Let $b: D \rightarrow \eucl$ be a conjugate binet.
  For $d \in D$ the line $N(d)$ is uniquely determined by the incidence and the orthogonality condition of Definition \ref{def:normalbicongruence}.
  Assume that adjacent lines of $N$ intersect.
  At every cross $(v,f,v',f') \in C$,
  the line $b(f) \vee b(f')$ is orthogonal to the plane $N(v) \vee N(v')$
  and thus in particular orthogonal to $b(v) \vee b(v')$.
\end{proof}

\begin{remark}
  The normal bicongruence gives rise to a canonical definition of a 2-parameter family of parallel binets (see Figure~\ref{fig:focal-binet}, left),
  that is binets with the same normal bicongruence (compare Remark~\ref{rem:parallel-normal-binet}).
  Furthermore, for each parameter direction,
  the points of intersection of the normal bicongruence yield a \emph{focal binet} (see Figure~\ref{fig:focal-binet}, right), which is a conjugate binet
  (see Remark~\ref{rem:focal-binets}).
  In the smooth case, the focal surfaces obtained from the normal lines of a principal parametrization
  come as conjugate line parametrizations with one family of geodesic parameter lines.
\end{remark}
\begin{figure}[H]
  \centering
  \includegraphics[width=0.45\textwidth]{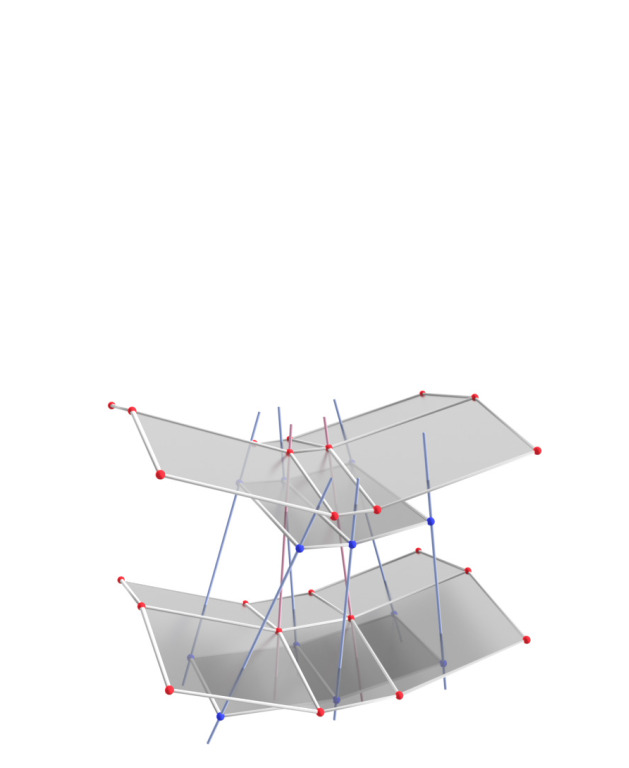}
  \includegraphics[width=0.45\textwidth]{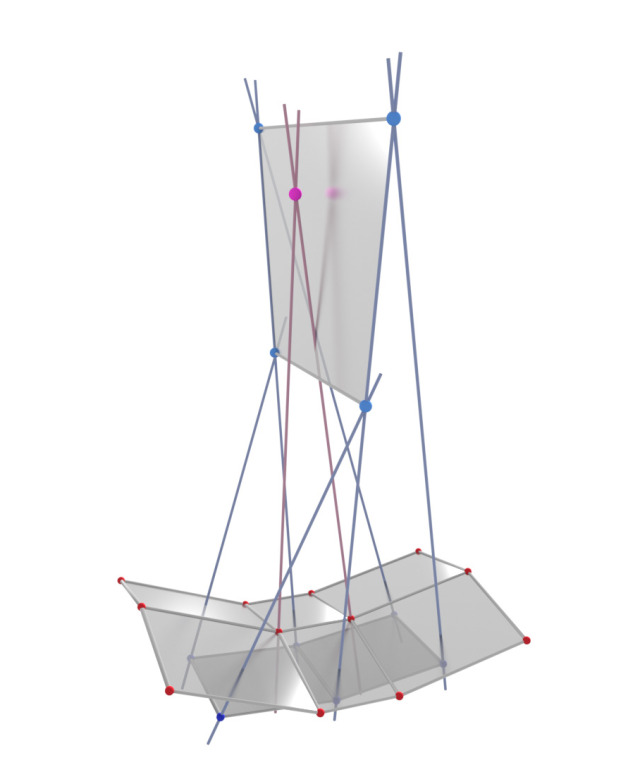}
  \caption{
    Left: Parallel binet of a principal binet.
    Right: Focal binet of principal binet.
  }
  \label{fig:focal-binet}
\end{figure}
\begin{question}
 	In what sense do the focal binets obtained from principal binets constitute semi-geodesic binets?  For circular and conical nets, this has been studied in \cite{hsssemigeodesic}.  
\end{question}

\subsection*{Möbius lift}
Since a principal binet is a special case of an orthogonal binet, let us discuss properties of the Möbius lift of a principal binet. We begin with the specialization of Lemma~\ref{lem:moebius-lift-projection}.

\begin{lemma}\label{lem:principalmobiuslift}
	Let $b_\mobq : D \rightarrow \RP^4 \setminus \PB^\pol$ be a conjugate and polar binet with respect to the Möbius quadric $\mobq \subset \RP^4$.
	Then its projection
	\[
	b \coloneqq \pi_\eucl \circ b_\mobq,
	\]
	to Euclidean space $\eucl$ is a principal binet.
\end{lemma}
\proof{
	Due to Lemma~\ref{lem:moebius-lift-projection} the projection $b$ is an orthogonal binet. Moreover, if four points are coplanar before projection they are also coplanar after projection. Thus, $b$ is a conjugate binet and therefore a principal binet as well.\qed
}

Indeed, conjugate polar binets are (generically) in bijection with principal binets, as captured by the following specialization of Theorem~\ref{th:mobiusortho}.

\begin{theorem}
  \label{thm:conjugatemobiuslift}
  Let $b : D \rightarrow \eucl$ be a regular orthogonal binet.
  Then its Möbius lift $b_\mobq: D \rightarrow \RP^4 \setminus \PB^\pol$ is a conjugate binet if and only if $b$ is a principal binet
\end{theorem}

\begin{proof}\ \linebreak
 ($\Rightarrow$) 
  Follows from Lemma~\ref{lem:principalmobiuslift}.
  \newline
  ($\Leftarrow$) 
  Consider a quad $f\in F$ consisting of the four vertices $v_1,v_2,v_3,v_4\in V$.
  Since $b$ is a conjugate binet, the four points $b(v_1), b(v_2), b(v_3), b(v_4)$
  lie in the plane $\square b (f)$.
  The four points $b_\mobq(v_1), b_\mobq(v_2), b_\mobq(v_3), b_\mobq(v_4)$ lie in
  the span of $\square b(f)$ and $\PB$, wich is 3-dimensional,
  and in the polar space of $b_\mobq(f)$, which is also 3-dimensional.
  Thus, their intersection 
  $
    (\square b(f) \vee \PB) \cap b_\mobq(f)^\pol
  $
  is a plane.
  The same argument holds for four faces adjacent to a common vertex. Hence, the Möbius lift is a conjugate binet.
\end{proof}

Note that Lemma~\ref{lem:regularmobiuslift} also covers the case of principal binets, that is the Möbius lift of a regular principal binet is a regular conjugate polar binet, as there are no additional regularity constraints for conjugate binets.

The fact that $b_\mobq$ is a conjugate binet implies an additional Möbius geometric structure of principal binets.
For any vertex or face $d \in D$ the projection of the planar section
\begin{align}
	 \circl(d) \coloneqq \pi_\eucl(\square b_\mobq(d) \cap \mobq), \label{eq:orthocircle}
\end{align}
is a (possibly imaginary) circle (see Figure~\ref{fig:principal-binet-circles}).
\begin{figure}[H]
  \centering
  \includegraphics[width=0.3\textwidth, trim={800 0 300 0},clip]{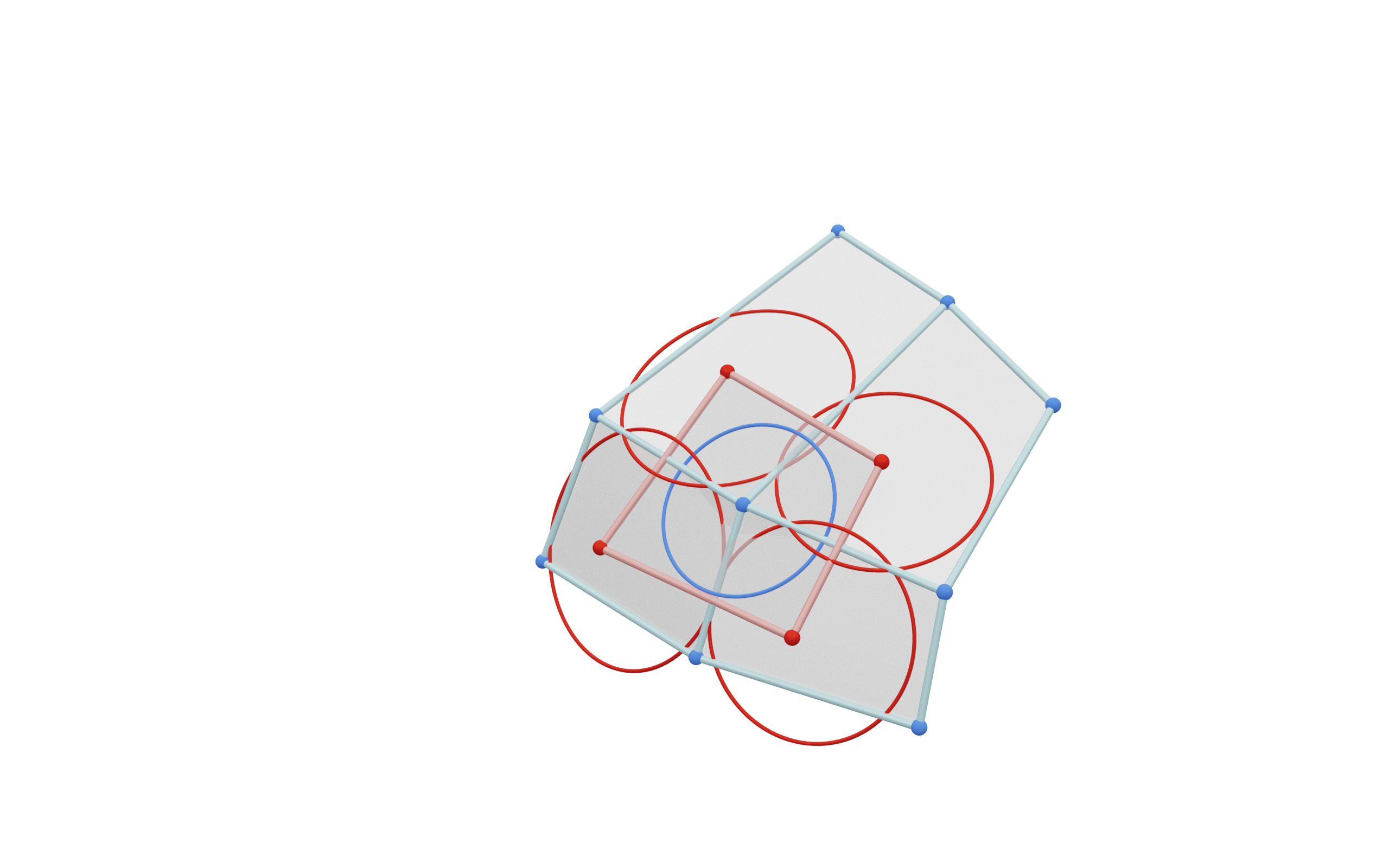}
  \includegraphics[width=0.3\textwidth, trim={800 0 300 0},clip]{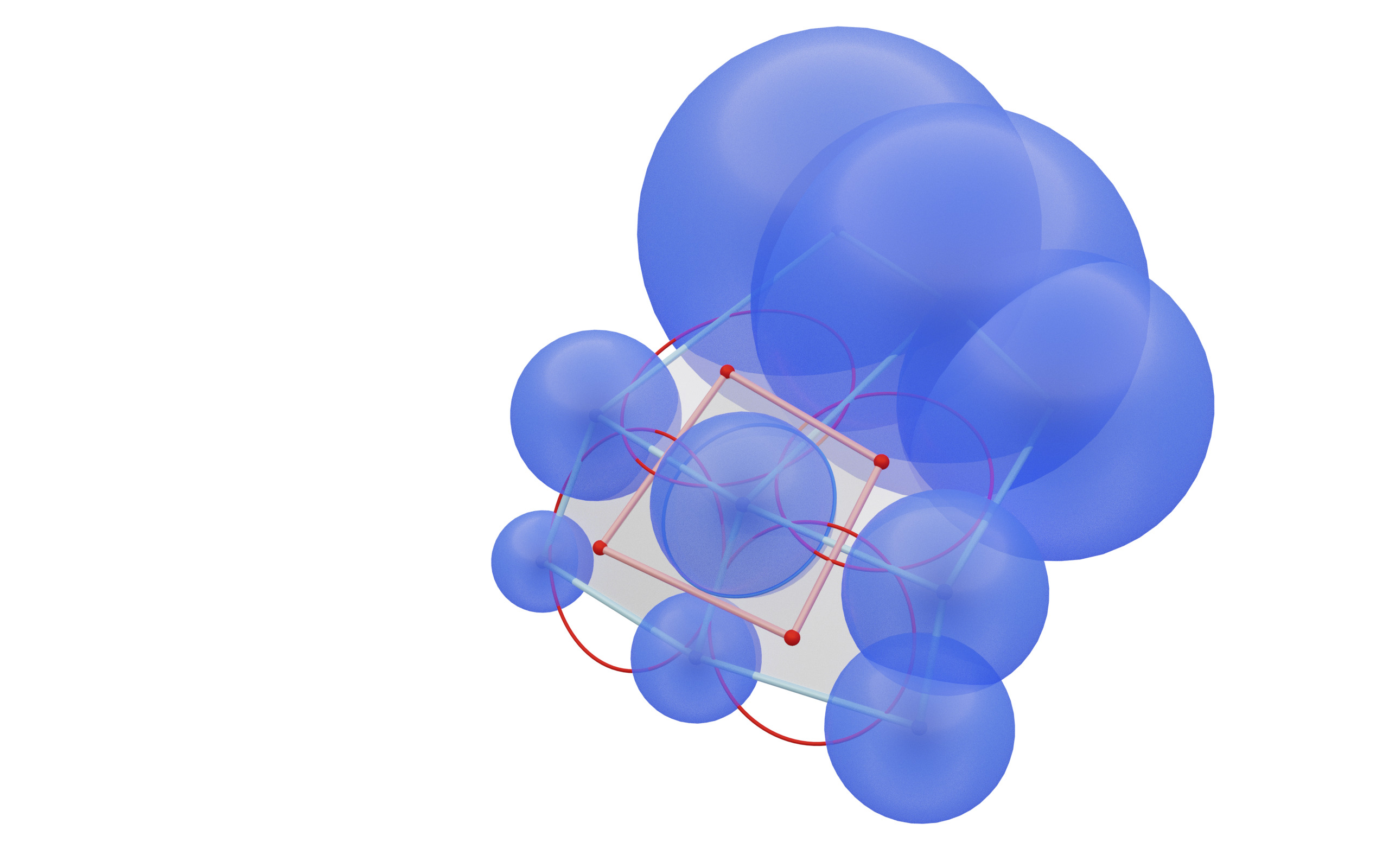}
  \includegraphics[width=0.3\textwidth, trim={800 0 300 0},clip]{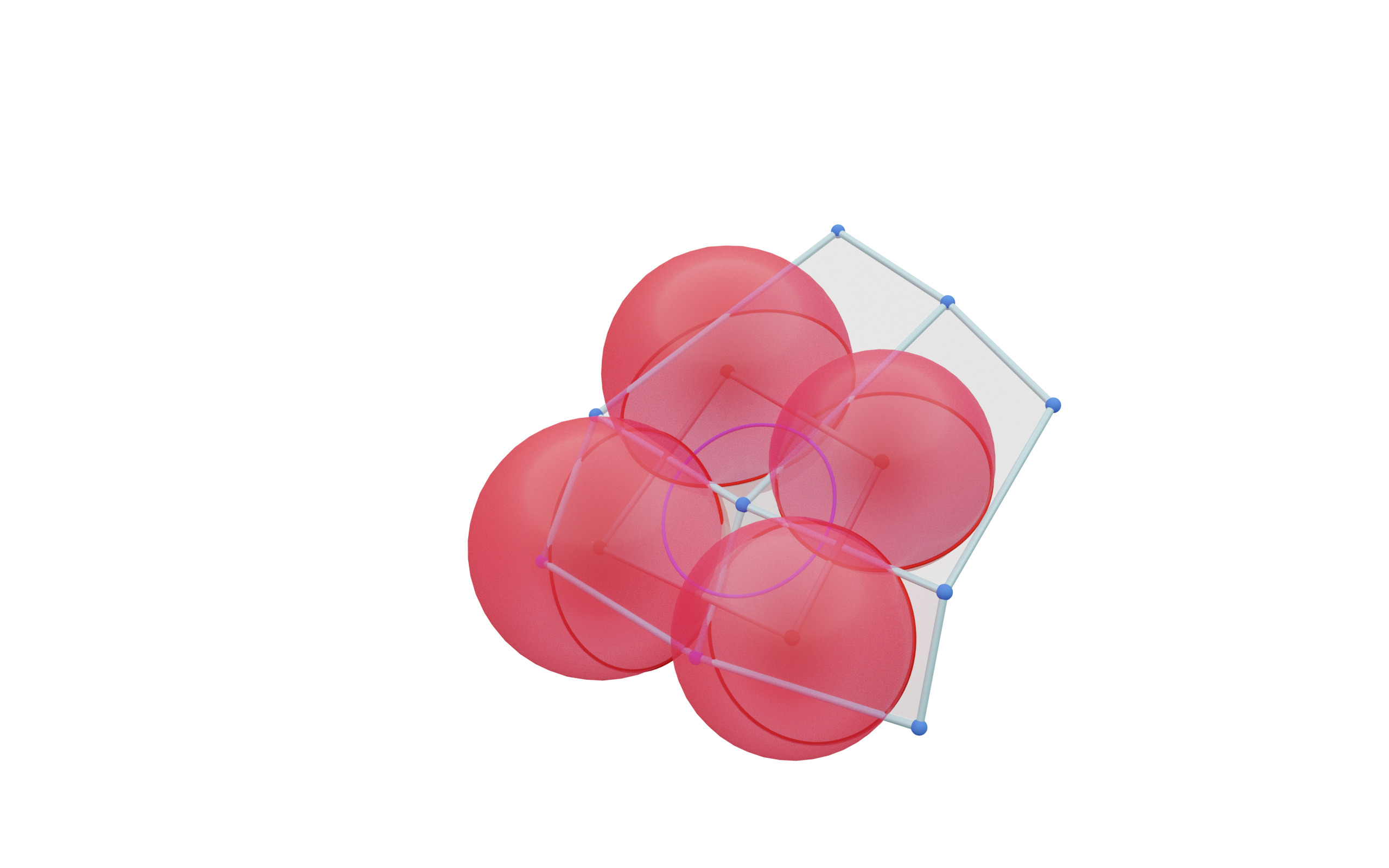}
  \caption{
    Circles of a principal binet.
    The circle $\circl(d)$ is contained in the sphere $b_\sp(d)$ of the orthogonal sphere representation for $d \in D$.
    Furthermore $\circl(d)$ is orthogonal to the sphere $b_\sp(d')$ for incident $d, d'$.
  }
  \label{fig:principal-binet-circles}
\end{figure}
Let $d_1, d_2, d_3, d_4 \in D$ be the four vertices (or faces) incident to $d$.
Then, the circle $\circl(d)$ is contained in every sphere that is orthogonal
to all four spheres  $b_\sp(d_1), b_\sp(d_2), b_\sp(d_3), b_\sp(d_4)$.
Equivalently, $\circl(d)$ is orthogonal to all four spheres,
and, in particular, $\circl(d)$ is contained in $b_\sp(d)$.
\begin{remark}
  The circles $\circl(d)$ may be thought of as an analogue of the circles appearing in S-isothermic nets \cite{bhssminimal}.
\end{remark}
\begin{proposition}
  \label{prop:circle-axis}
  Let $b : D \rightarrow \eucl$ be a regular principal binet,
  and $N : D \rightarrow \linesof{\eucl}$ be its normal bicongruence.
  For every vertex or face $d \in D$
  \begin{enumerate}
  \item
    the circle $\circl(d)$ lies in the plane $\square b(d)$,
  \item
    and the axis of $\circl(d)$ is the normal line $N(d)$.\qedhere
  \end{enumerate}
\end{proposition}
\begin{proof}
  With the orthogonal sphere representation $b_\sp = \xi_\sp \circ b_\mobq$ this circle can also be described as
  \[
    \circl(d) = \square b(d) \cap b_\sp(d).
  \]
  Thus, the axis of $\circl(d)$ is the unique line that contains the center $b(d)$ and is orthogonal to $\square b(d)$.
\end{proof}

\subsection*{Laguerre lift}

Since a principal binet is also a special case of an orthogonal bi*net, let us discuss properties of the Laguerre lift of a principal binet. We begin with the specialization of Lemma~\ref{lem:laguerre-lift-projection}.

\begin{lemma}
  \label{lem:principallaguerrelift}
  Let $b_\blac : D \rightarrow \RP^4 \setminus \ell_\blac$ be a polar binet
  with respect to the Blaschke cylinder $\blac \subset \RP^4$.
  Then the dual of its projection
  \[
    b:  D \rightarrow \pl, \quad b = \star (\pi_{\eucl^*} \circ b_\blac),
  \]
  is a principal bi*net.
\end{lemma}
\proof{
	As in the proof of Lemma~\ref{lem:principalmobiuslift}, the projection $\pi_{\eucl^*}$ preserves conjugacy, which in turn translates to the conjugacy of the bi*net $b$. Combined with Lemma~\ref{lem:laguerre-lift-projection} the claim follows. \qed
}

Note that the part about normal binets of Lemma~\ref{lem:laguerre-lift-projection} does not specialize in any sense to principal bi*nets: there is no way to distinguish the normal binet of a principal bi*net from the normal binet of an orthogonal bi*net.

We define the Laguerre lift $b_\blac$ of a principal binet $b$ to be the Laguerre lift of principal bi*net $\square b$.
As in the case of the Möbius lift, the conjugacy of the Laguerre lift is characterizing for principal binets.

\begin{theorem}
  Let $b : D \rightarrow \eucl$ be a regular orthogonal bi*net.
  Then its Laguerre lift $b_\blac : D \rightarrow \RP^4 \setminus \ell_\blac$ is a conjugate binet if and only if $b$ is a conjugate bi*net.
\end{theorem}

\begin{proof}\ \linebreak
 ($\Rightarrow$) 
  Follows from Lemma~\ref{lem:principallaguerrelift}.
  \newline
  ($\Leftarrow$) 
  Consider a quad $f\in F$ consisting of the four vertices $v_1,v_2,v_3,v_4\in V$.
  Since $\square b$ is a conjugate bi*net,
  the four planes $\square b(v_1), \square b(v_2), \square b(v_3), \square b(v_4)$ intersect in a point. Correspondingly the four points $\star\square b(v_1), \star\square b(v_2), \star\square b(v_3), \star\square b(v_4) \in \eucl^*$ lie in the plane $\star b(f)$.
  Thus, the four points of the lift $b_\blac(v_1), b_\blac(v_2), b_\blac(v_3), b_\blac(v_4)$ lie in
  the span of $\star b(f)$ and $\PM$, which is 3-dimensional,
  and in the polar space of $b_\blac(f)$, which is also 3-dimensional.
  Thus, their intersection 
  $
    (\star b(f) \vee \PM) \cap b_\blac(f)^\pol
  $
  is a plane.
  The same argument holds for four faces adjacent to a common vertex. Hence, the Laguerre lift is a conjugate binet.
\end{proof}

Analogously to the Möbius lift, the fact that $b_\blac$ is a conjugate binet provides additional Laguerre geometric structure for principal binets.
For any vertex or face $d \in D$, the oriented planes corresponding to the planar section
\begin{align}
  \cone(d) \coloneqq
  \xi_\opl(\square b_\blac(d) \cap \blac),
  \label{eq:principalcone}
\end{align}
are a 1-parameter family of oriented planes touching an oriented cone in $\eucl$.
We identify $\cone(d)$ with this oriented cone (see Figure~\ref{fig:principal-cones}).
Note that the normals of the oriented tangent planes of $\cone(d)$ are contained in the normal circle
\[
  b_\nci(d) = \xi_\nci \circ b_\blac(d)
  = \pi_\unis(b_\blac(d))^\pol \cap \unis
  = \pi_{\unis}(\square b_\blac(d) \cap \blac),
\]
from the orthogonal circle representation $b_\nci$ of $b$.
In particular, the axis of $\cone(d)$ and the axis of $b_\nci(d)$ are parallel.
\begin{figure}[H]
  \centering
  \includegraphics[width=0.38\textwidth]{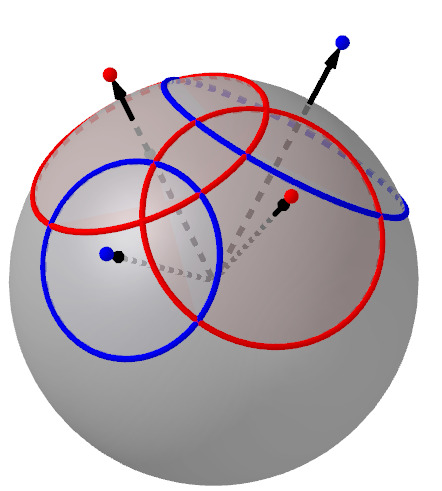}
  \hspace{0.05\textwidth}
  \includegraphics[width=0.5\textwidth]{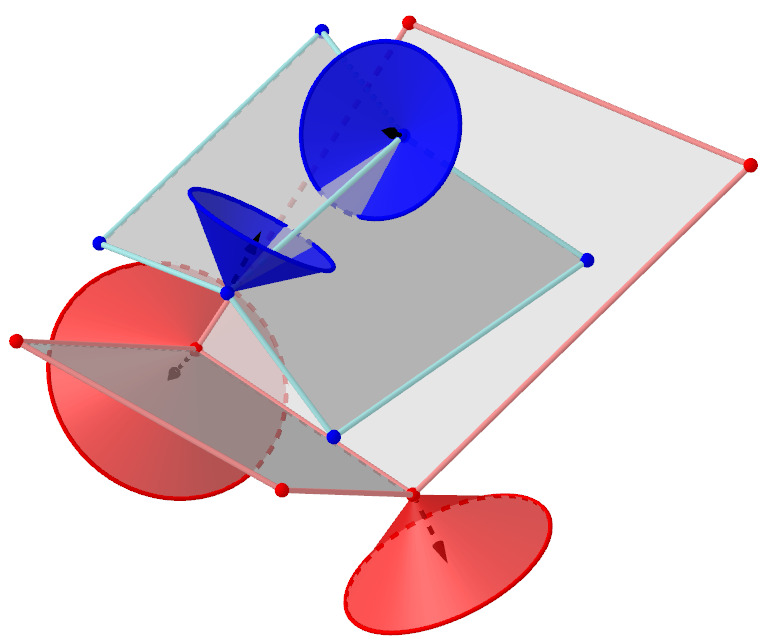}
  \caption{
    Left: Normal binet and orthogonal circle representation of a principal binet.
    Right: The corresponding cones $\cone$.
    The normals of the oriented tangent planes of $\cone(d)$ are contained in the corresponding normal circle $\xi_\nci(d)$ from the orthogonal circle representation.
  }
  \label{fig:principal-cones}
\end{figure}

\begin{proposition}
  \label{prop:cone-axis}
  Let $b : D \rightarrow \eucl$ be a regular principal binet,
  and let $N : D \rightarrow \linesof{\eucl}$ be its normal bicongruence.
  For every vertex or face $d \in D$
  \begin{enumerate}
  \item
    the apex of the oriented cone $\cone(d)$ is $b(d)$,
  \item
    and the axis of $\cone(d)$ is the normal line $N(d)$.\qedhere
  \end{enumerate}
\end{proposition}

\begin{proof}\
  \nobreakpar
  \begin{enumerate}
  \item
    The (non-oriented) planes
    $
    \xi_\pl(\star b(d))
    $
    are all the planes in $\eucl$ through the point $b(d)$.
    Since the projection of the planar section
    \[
      \pi_{\eucl^*}(\square b_\blac(d) \cap \blac) \subset \star b(d),
    \]
    lies in the plane $\star b(d)$,
    the corresponding (non-oriented) planes
    \[
      \xi_\pl \circ \pi_{\eucl^*}(\square b_\blac(d) \cap \blac),
    \]
    all go through the point $b(d)$.
    And thus the apex of $\cone(d)$ is the point $b(d)$.
  \item
    The normal circle $b_\nci(d)$ from the orthogonal circle representation is a circle in $\unis$
    whose axis coincides with the normal of $\square b(d)$.
    Therefore, the axis of $\cone(d)$ is the line through $b(d)$ that is normal to $\square b(d)$,
    which is the normal line $N(d)$.\qedhere
  \end{enumerate}
\end{proof}

\section{Lie geometry}
\label{sec:lie}

Recall the projective model of Lie geometry \cite{blaschkevl, cecillie}.
Let $\sca{\cdot, \cdot}_{\lieq}$ be a symmetric bilinear form of signature $\texttt{(++++--)}$, and
\[
\lieq = \set{[x] \in \RP^5}{ \sca{x,x}_{\lieq} = 0}
\]
the corresponding quadric in $\RP^5$, which we call the \emph{Lie quadric}.
\begin{coordinates}
	In homogeneous coordinates of $\RP^5$ we choose
	\[
	\sca{x,x}_{\blac} = x_1^2 + x_2^2 + x_3^2 + x_4^2 - x_5^2 - x_6^2.
	\]
	And thus, in affine coordinates $x_5 = 1$, the Lie quadric is:
	\[
	x_1^2 + x_2^2 + x_3^2 + x_4^2 - x_6^2 = 1.
	\]
\end{coordinates}

We now embed Möbius geometry and Laguerre geometry into Lie geometry in the following way (see Figure~\ref{fig:lie-geometry}).
Let $\PM \in \lieq^-$ be a point inside the Lie quadric, that is a point with signature $\texttt{(-)}$.
Then
\[
  \mobq \coloneqq \lieq \cap \PM^\pol,
\]
is a quadric of signature $\texttt{(++++-)}$, which we identify with the \emph{Möbius quadric}.
Let $\PB \in \mobq$ be a point on the Möbius quadric, that is a point with signature $\texttt{(0)}$ and $\PB \pol \PM$.
Then
\[
  \blac \coloneqq \lieq \cap \PB^\pol,
\]
is a quadric of signature $\texttt{(+++-0)}$, which we identify with the \emph{Blaschke cylinder}.

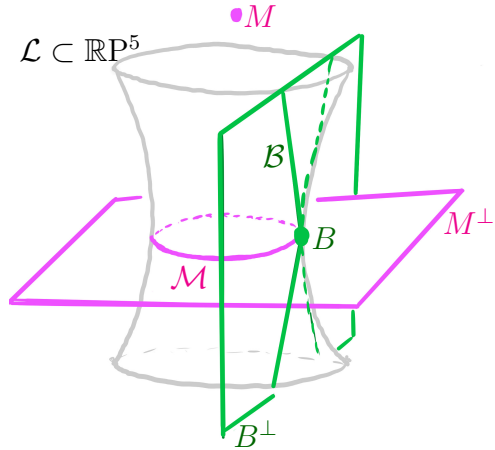
\begin{figure}[H]
  \centering
  \begin{overpic}[width=0.4\textwidth]{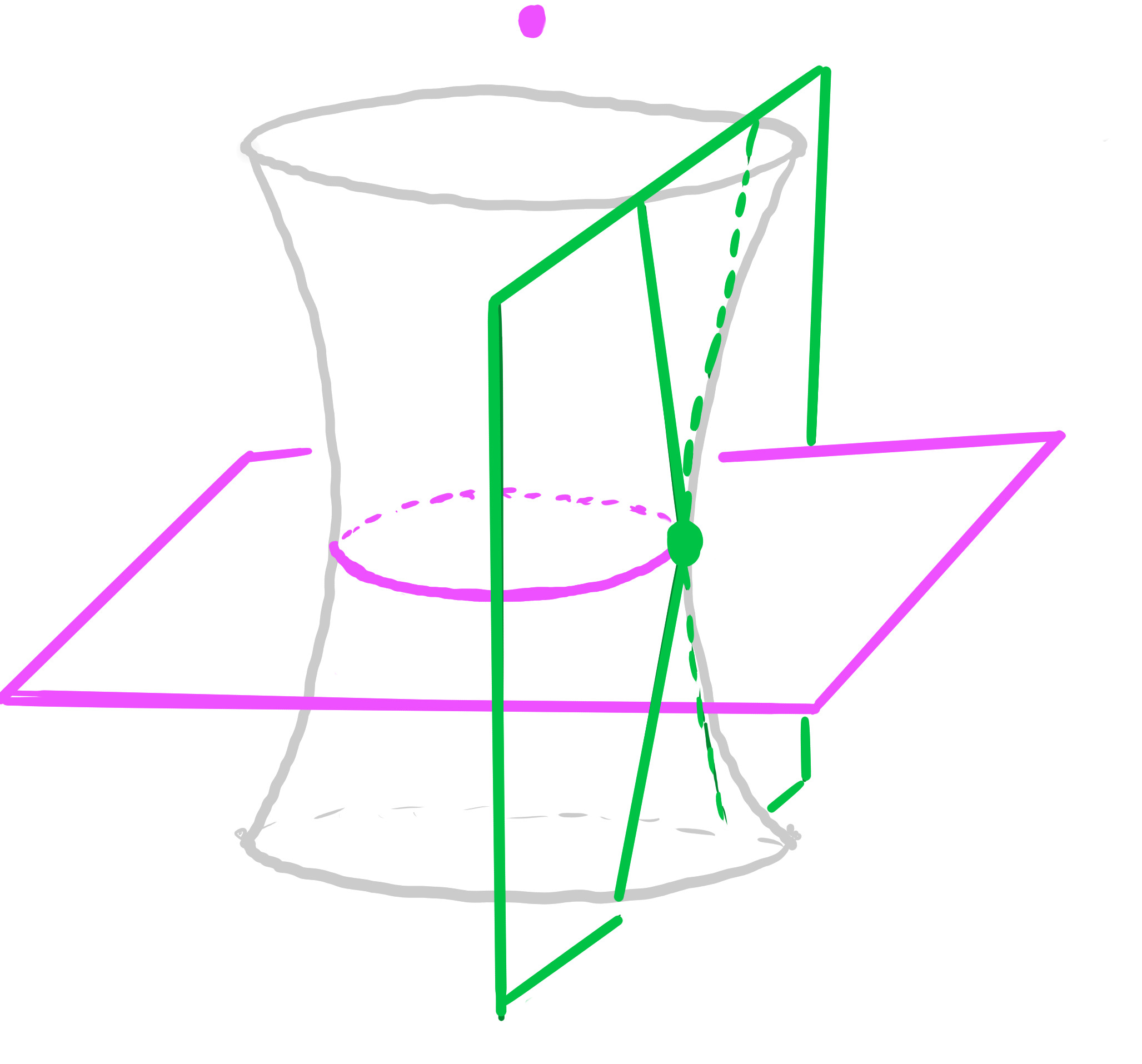}
    \put(2,82){$\lieq \subset \RP^5$}
    \put(48,89){$\color{magenta}\PM$}
    \put(89,46){$\color{magenta}\PM^\pol$}
    \put(33,35){$\color{magenta}\mobq$}
    \put(62,43){$\color{darkgreen}\PB$}
    \put(46,2){$\color{darkgreen}\PB^\pol$}
    \put(52,60){$\color{darkgreen}\blac$}
  \end{overpic}
  \caption{
    The projective model of Lie geometry with Möbius geometry and Laguerre geometry as subgeometries.
  }
  \label{fig:lie-geometry}
\end{figure}

As explained in Section~\ref{sec:moebius}, we embed the 3-dimensional Euclidean space $\eucl$ into $M^\pol$:
\[
  \eucl \subset S_\eucl \subset M^\pol,
\]
where $S_\eucl \subset M^\pol$ is a hyperplane that does not contain $\PB$.
We generalize the central projection $\pi_\eucl$ as
\[
  \pi_{\eucl} : \RP^5 \setminus (\PB \vee \PM) \rightarrow S_\eucl,\qquad
  X \mapsto (X \vee \PB \vee \PM) \cap S_\eucl,
\]
whose restriction to $\PM^\pol$ coincides with the previous definition.
In particular, its restriction to $\mobq$ is the stereographic projection.
We generalize the map $\xi_\sp$ into the set of spheres as
\[
  \xi_\sp: \RP^5 \setminus \{\PM\} \rightarrow \sp, \qquad
  X \mapsto \pi_\eucl( X^\pol \cap \mobq ),
\]
whose restriction to $M^\pol$ again coincides with the previous definition.
Therefore, every point in $\RP^5 \setminus \{\PM\}$ defines a sphere in $\eucl$,
and $\pi_\eucl(X)$ is still the center of $\xi_\sp(X)$ for all $X\in \RP^5  \setminus \{\PM\}$.
For $X \in \RP^5  \setminus \{\PM\}$ and $X' \in \mobq$ the incidence of the point $\pi_\eucl(X')$ lying on the sphere $\xi_\sp(X)$ is described by
\[
  \pi_\eucl(X') \in \xi_\sp(X)
  \quad\Leftrightarrow\quad
  X \pol X'.
\]

Note that if $X \in \PB^\pol$, then $X^\pol$ contains $\PB$ and therefore $\xi_\sp(X)$ is a plane in $\eucl$.
Thus, the restriction of $\xi_\sp$ to $\PB^\pol$ is a map into the set of Euclidean planes $\pl \subset \sp$.
As explained in Section~\ref{sec:laguerre},
we embed the dual 3-dimensional Euclidean space $\eucl^*$ into $\PB^\pol$ as
\[
  \eucl^* \subset S_{\eucl^*} \subset \PB^\pol,
\]
where $S_{\eucl^*} \coloneqq \PB^\pol \cap \PM^\pol$.
For a point $X \in \eucl^*$, the corresponding dual plane in the embedding of the Euclidean space $\eucl$ is given by polarity
\[
  \star X = X^\pol \cap \eucl.
\]
The central projection $\pi_{\eucl^*}$ is generalized to
\[
  \begin{aligned}
    \pi_{\eucl^*}&: \PB^\pol \setminus \{ \PM \} \rightarrow S_{\eucl^*}, \quad X \mapsto  (X \vee \PM) \cap S_{\eucl^*}.
  \end{aligned}
\]
Thus, for a point $X \in \PB^\pol$, we obtain a corresponding plane
\[
  \xi_\sp(X) = \xi_\pl \circ \pi_{\eucl^*}(X) = \star \pi_{\eucl^*}(X),
\]
and with $X' \in \mobq$ the incidence of the point $\pi_\eucl(X')$ lying on the plane $\xi_\sp(X)$ is again described by polarity
\[
  \pi_\eucl(X') \in \star \pi_{\eucl^*}(X)
  \quad\Leftrightarrow\quad
  X \pol X'.
\]

The restriction of $\xi_\sp$ to the Lie quadric is two-to-one, that is a double cover of the space of spheres.
This gives rise to a map
\[
  \xi_\osp : \lieq \rightarrow \osp,
\]
into the space $\osp$ of \emph{oriented Euclidean spheres} (where we include oriented planes and non-oriented points).
Specifically, the restriction of $\xi_\osp$ to $\blac$ is a map into the set of oriented Euclidean planes $\opl \subset \osp$.
In this representation the oriented contact of oriented spheres is given by polarity with respect to the Lie quadric $\lieq$.
\begin{proposition}
  \label{prop:liepolarity}
  Two points $X, X' \in \lieq$ are polar with respect to $\lieq$
  if and only if the two corresponding oriented spheres $\xi_\osp(X)$ and $\xi_\osp(X')$ are in oriented contact.
\end{proposition}
In particular, for a point $X \in \lieq$ the set $\xi_\osp(X^\pol \cap \blac)$
consists of all planes in oriented contact with the sphere $\xi_\osp(X)$.

\begin{coordinates}
  We choose 
  \[
    \begin{aligned}
      \PM &= [0,0,0,0,0,1], \\
      \PB &= [0,0,0,1,1,0], \\
      S_\eucl &= \set{[x] \in \RP^5}{ x_4 = x_6 = 0}, \\
      S_{\eucl^*} &= \set{[x] \in \RP^5}{ x_4 = x_5,\ x_6 = 0}.
    \end{aligned}
  \]
  Then the embedding of the Euclidean space is given by
  \[
    \begin{aligned}
      E^\infty &= S_\eucl \cap \PB^\pol = \set{[x] \in \RP^5}{ x_4 = x_5 = x_6 = 0}, \\
      \eucl &= S_\eucl \setminus E^\infty = \set{[x] \in \RP^5}{ x_4 = x_6 = 0, x_5 \neq 0},	
    \end{aligned}
  \]
  with the projection
  \[
    \pi_{\eucl}([x]) = [x_1, x_2, x_3, 0, x_5 - x_4, 0].
  \]

  For the lift to $\RP^5$ we introduce the basis of $\R^6$
  \[
    e_1,\quad
    e_2,\quad
    e_3,\quad
    e_\infty = \tfrac{1}{2}(e_5 + e_4),\quad
    e_0 = \tfrac{1}{2}(e_5 - e_4),\quad
    e_6.
  \]
  Then an oriented sphere with center $c \in \R^3$ and signed radius $r \in \R$
  corresponds to the point $[x] \in \lieq \setminus B^\pol$ with
  \[
    x = c + e_0 + (\abs{c}^2 - r^2)e_\infty + re_6.
  \]
  Vice versa, from a point $[x] \in \lieq \setminus B^\pol$
  the center and signed radius of the corresponding oriented sphere $\xi_\osp([x])$ are recovered by
  \[
    c = \pi_\eucl([x]), \qquad
    r = \frac{\sca{x,e_6}_\lieq}{2\sca{x, e_\infty}_\lieq}.
  \]
  For two points
  \[
    \begin{aligned}
      x &= c + e_0 + (\abs{c}^2 - r^2) e_\infty + re_6,\\
      x' &= c' + e_0 + (\abs{c'}^2 - r'^2) e_\infty + r'e_6.
    \end{aligned}
  \]
  the polarity $\sca{x,x'}_\lieq = 0$ is equivalent to
  \[
    \frac{r^2 + r'^2 - \abs{c - c}^2}{2rr'} = 1 = \cos 0.
  \]
  The left-hand side is the \emph{signed inversive distance} and thus,
  this condition describes the oriented contact of $\xi_\osp([x])$ and $\xi_\osp([x'])$.
\end{coordinates}

\subsection*{Angled spheres}

We now discuss points that do not lie on the Lie quadric.
The \emph{outside} of the Lie quadric is given by
\[
  \lieq^+ = \set{[x] \in \RP^5}{\sca{x,x}_\lieq > 0}.
\]
For a point $X \in \lieq^+$, the polar hyperplane $X^\pol$ is a 4-dimensional subspace of signature $\texttt{(+++--)}$.
The intersection of  $X^\pol$ with the Lie quadric gives us a geometric interpretation for the point $X$
in terms of the oriented spheres
\[
  \xi_\osp(X^\pol \cap \lieq).
\]
The line $X \vee M$ has signature $\texttt{(+-)}$ and thus intersects $\lieq$ in two points $X^+, X^-$,
which correspond to two coinciding oriented spheres $\xi_\osp(X^+)$ and $\xi_\osp(X^-)$ with opposite orientation,
or one non-oriented sphere $\xi_\sp(X) = \xi_\sp(X^+) = \xi_\sp(X^-)$.

For each $X$ there is some fixed angle $\varphi$, such that for every $X' \in \lieq \cap X^\pol$ the normal vector of $\xi_\osp(X')$ at a point of intersection with $\xi_\sp(X)$ has angle $\varphi$ with the normal vector of $\xi_\osp(X^+)$ at this point. Equivalently, the normal vector has angle $\pi -\varphi$ with the normal vector of $\xi_\osp(X^-)$ at this point.

In particular, the subset $\xi_\osp(X^\pol \cap \mobq)$ consists of all points on the non-oriented sphere $\xi_\sp(X)$,
while the subset $\xi_\osp(X^\pol \cap \blac)$ consists of all oriented planes that intersect $\xi_\sp(X^+)$ in a fixed angle.
In the case that $X \notin \PB^\pol$ the subset $\xi_\osp(X^\pol \cap \blac)$ may equivalently be described as all oriented planes which are in oriented contact with an oriented sphere $\xi_\osp(X)$ concentric with $\xi_\sp(X)$.

We extend the map $\xi_\osp$ correspondingly to a map
\[
  \xi_\osp : \RP^5 \setminus \PB^\pol \rightarrow \osp, \qquad
  X \mapsto \xi_\osp(Y),
\]
where the point $Y$ is the intersection point of $X \vee B$ with $\lieq$ that is not $B$.
Then $\xi_\osp(X)$ is the oriented sphere that is in oriented contact with all oriented planes $\xi_\osp(X^\pol \cap \blac)$.
Thus, we introduce the space of \emph{angled spheres} as pairs of a non-oriented sphere and a concentric oriented sphere of smaller radius
\[
  \asp \coloneqq \text{Angled spheres}(\eucl).
\]
This gives rise to the map
\[
  \xi_\asp : \lieq^+ \setminus \PB^\pol \rightarrow \asp,\qquad
  X \mapsto \left(\xi_\sp(X), \xi_{\osp}(X)\right).
\]

\begin{coordinates}
  A point
  \[
    [x] = \left[ c + e_0 + (\abs{c}^2 - r^2)e_\infty + r\cos\varphi\ e_6 \right] \in \lieq^+
  \]
  with $c \in \R^3$, $r \in \R$, $\varphi \in [0,\pi)$ represents an angled sphere
  with sphere $\xi_\sp([x])$ and angle $\varphi$.

  Indeed, for a point
  \[
    [x'] = \left[ c' + e_0 + (\abs{c'}^2 - r'^2)e_\infty + r' e_6 \right] \in \lieq,
  \]
  which represents an oriented sphere $\xi_\osp([x'])$, the polarity condition $\sca{x,x'}_\lieq = 0$ is equivalent to
  \[
    \frac{r^2 + r'^2 - \abs{c - c'}^2}{2rr'} = \cos\varphi.
  \]
  Thus $\xi_\sp([x'])$ intersects $\xi_\sp([x])$ in the angle $\varphi$.
\end{coordinates}

We may view the set of angled planes $\apl$ as a (degenerate) subset of angled spheres $\asp$, and extend the map $\xi_\asp$ to all of $\lieq^+ \setminus (\PB \vee \PM)$.
Furthermore, for points $X \in \lieq^-$ the oriented sphere $\xi_\osp(X)$ has greater radius than the non-oriented sphere $\xi_\sp(X)$,
or the non-oriented sphere $\xi_\sp(X)$ may become imaginary.
The corresponding pairs can be interpreted as \emph{imaginary angled spheres}.
If we extend the set $\asp$ accordingly, we can extend the map $\xi_\asp$ to a map
\[
  \xi_\asp : \RP^5 \setminus (\PB \vee \PM) \rightarrow \asp.
\]
\begin{remark}
  Angled spheres were already discussed by Blaschke \cite{blaschkevl}, as \emph{sphere-complexes} both in Laguerre geometry
  and in Lie geometry.
  This includes the description of angled spheres in terms of pairs of concentric non-oriented and oriented spheres, see \cite[Figure~44]{blaschkevl}.
  The non-imaginary case of angled spheres is also known as \emph{turbines}, see \cite{kasnerlie, narasingalie}.
\end{remark}

\subsection*{Pencils of angled spheres}

Furthermore, we discuss lines in $\RP^5$ not contained in the Lie quadric.
By the correspondence of points outside the Lie quadric with angled spheres,
lines outside the Lie quadric therefore correspond to \emph{pencils of angled spheres}.
We now give three geometric characterizations of lines in $\RP^5$ with respect to $\lieq$ (see Figure~\ref{fig:liepencil}).
Let $\ell \in \linesof{\RP^5}$ be a line of signature $\texttt{(+-)}$.

The first way to describe $\ell$ is via the two planes
\[
  E_M = (\ell^\pol \cap M^\pol),
  \qquad
  E_B = (\ell^\pol \cap B^\pol).
\]
Vice versa, two planes $E_M \subset M^\pol$, $E_B \subset B^\pol$ correspond to a line $\ell$
if and only if they intersect in a line in $M^\pol \cap B^\pol$.
The plane $E_M$ corresponds to a circle $\circl \coloneqq \pi_\eucl(E_M \cap \mobq)$ in Möbius geometry.
The plane $E_B$ corresponds to an oriented cone $\cone \coloneqq \xi_\opl(E_B \cap \blac)$ in Laguerre geometry.
The circle $\circl$ and the oriented cone $\cone$ are coaxial with axis $\pi_\eucl(\ell)$.
We call this pair a \emph{coaxial circle-cone pair}.
The angled spheres of the pencil consist of all spheres with center on the axis $\pi_\eucl(\ell)$ that contain the circle $\circl$,
where the angle is given by the intersection angle with $\cone$.

The second way to describe $\ell$ is via the two points
\[
  P_M = \ell \cap M^\pol,
  \qquad
  P_B = \ell \cap B^\pol,
\]
that span $\ell$.
The point $P_M$ corresponds to a non-oriented sphere $\xi_\sp(P_M)$ in Möbius geometry.
It is centered at the apex of $\cone$ and contains the circle $\circl$.
The point $P_B$ corresponds to an angled plane $\xi_\apl(P_B)$ in Laguerre geometry.
Its center plane is the plane of the circle $\circl$, while the angle is given by the intersection angle of the center plane with $\cone$.


The third way to describe $\ell$ is via the two points  
\[
  \{ Q_1, Q_2 \} =   \ell \cap \lieq,
\]
that span $\ell$.
These points correspond to two oriented spheres $\xi_\osp(Q_1), \xi_\osp(Q_2)$ in Lie geometry.
These are the two spheres that contain the circle $\circl$ and are in oriented contact with the cone $\cone$.

\begin{figure}[H]
  \includegraphics[width=220pt]{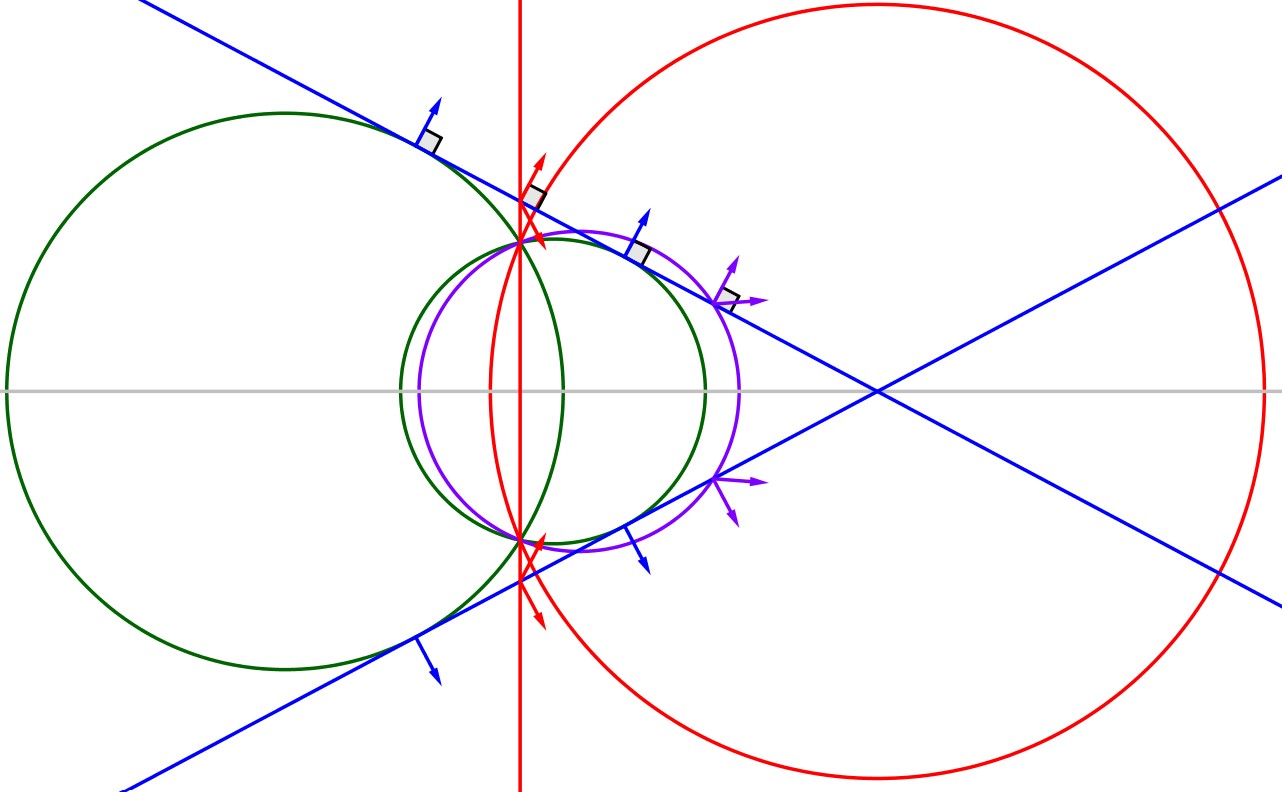}
  \includegraphics[width=220pt]{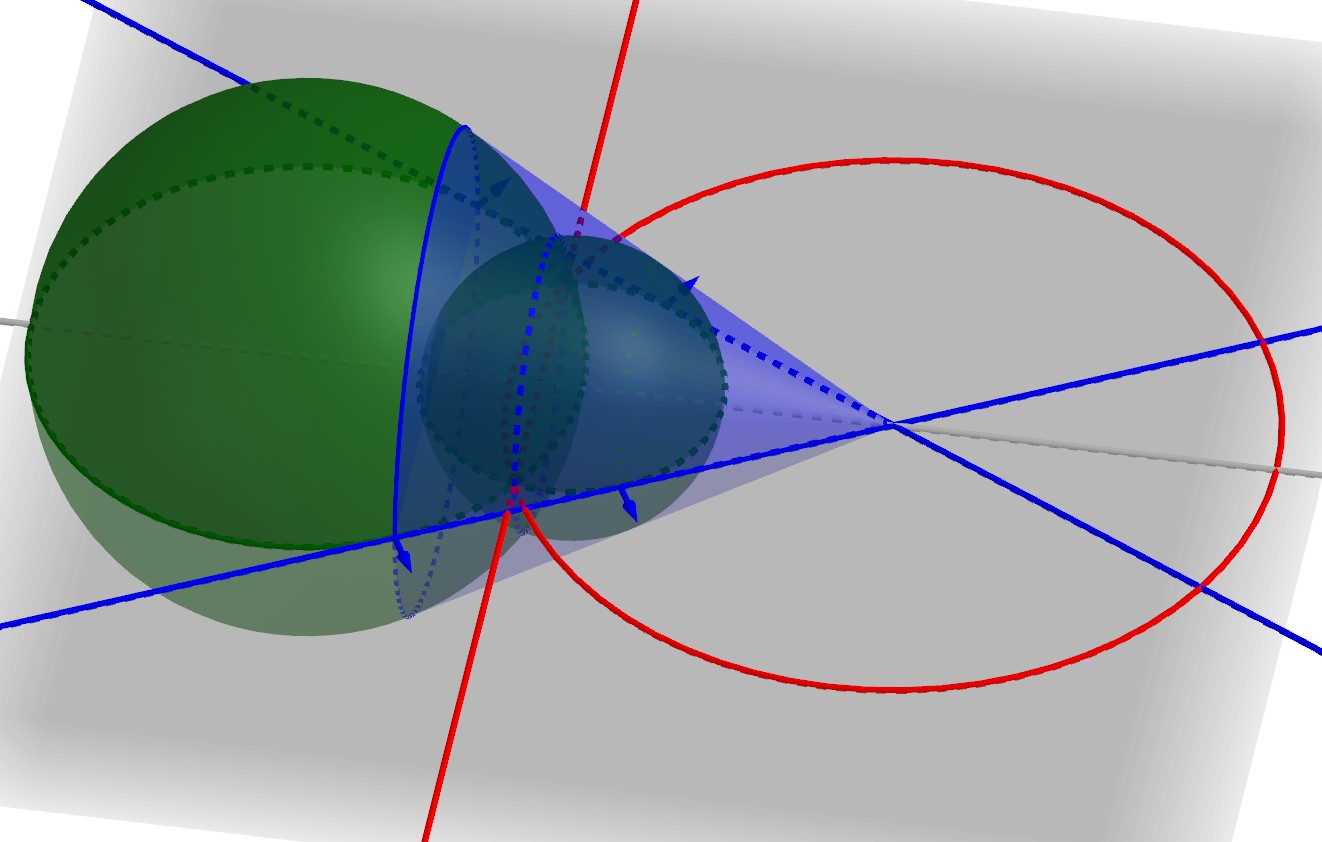}\\
  \includegraphics[width=220pt]{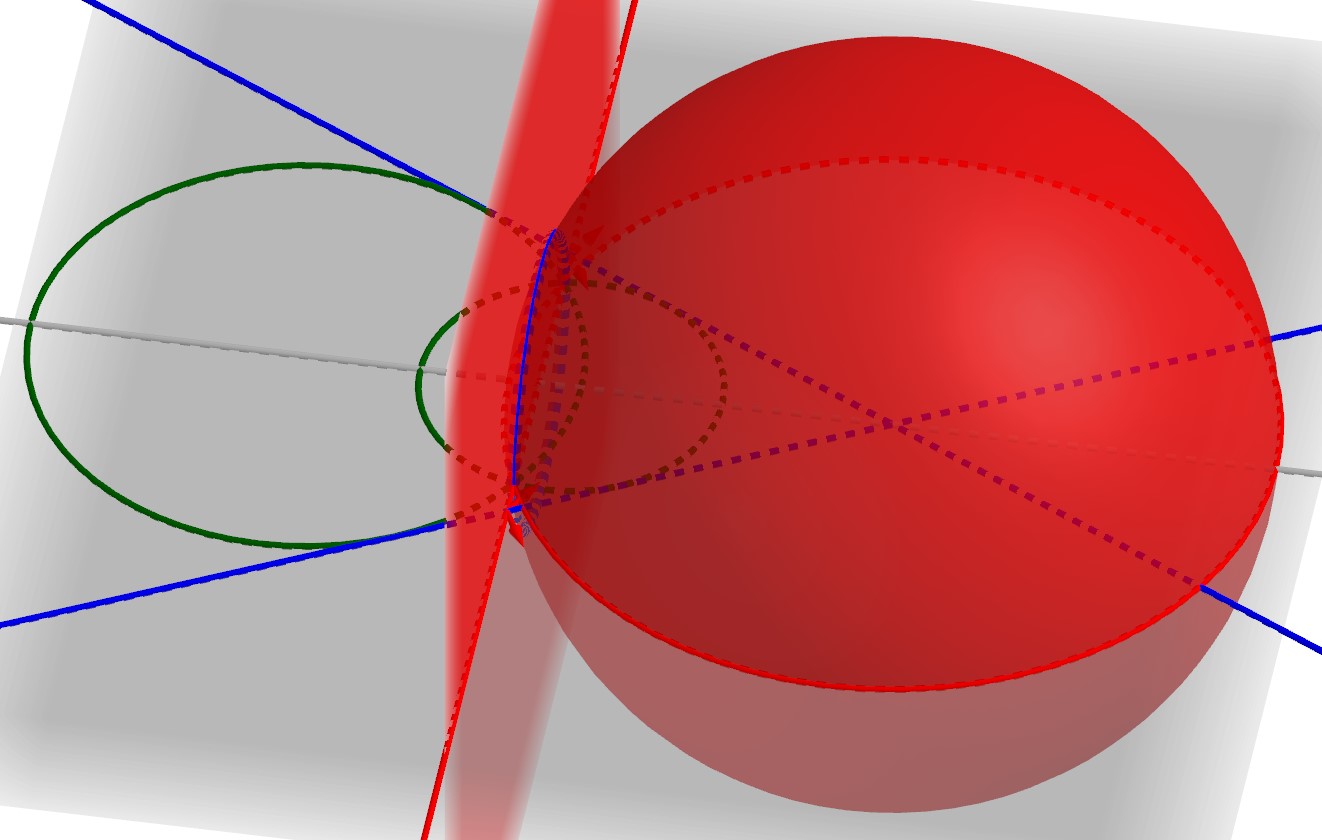}
  \includegraphics[width=220pt]{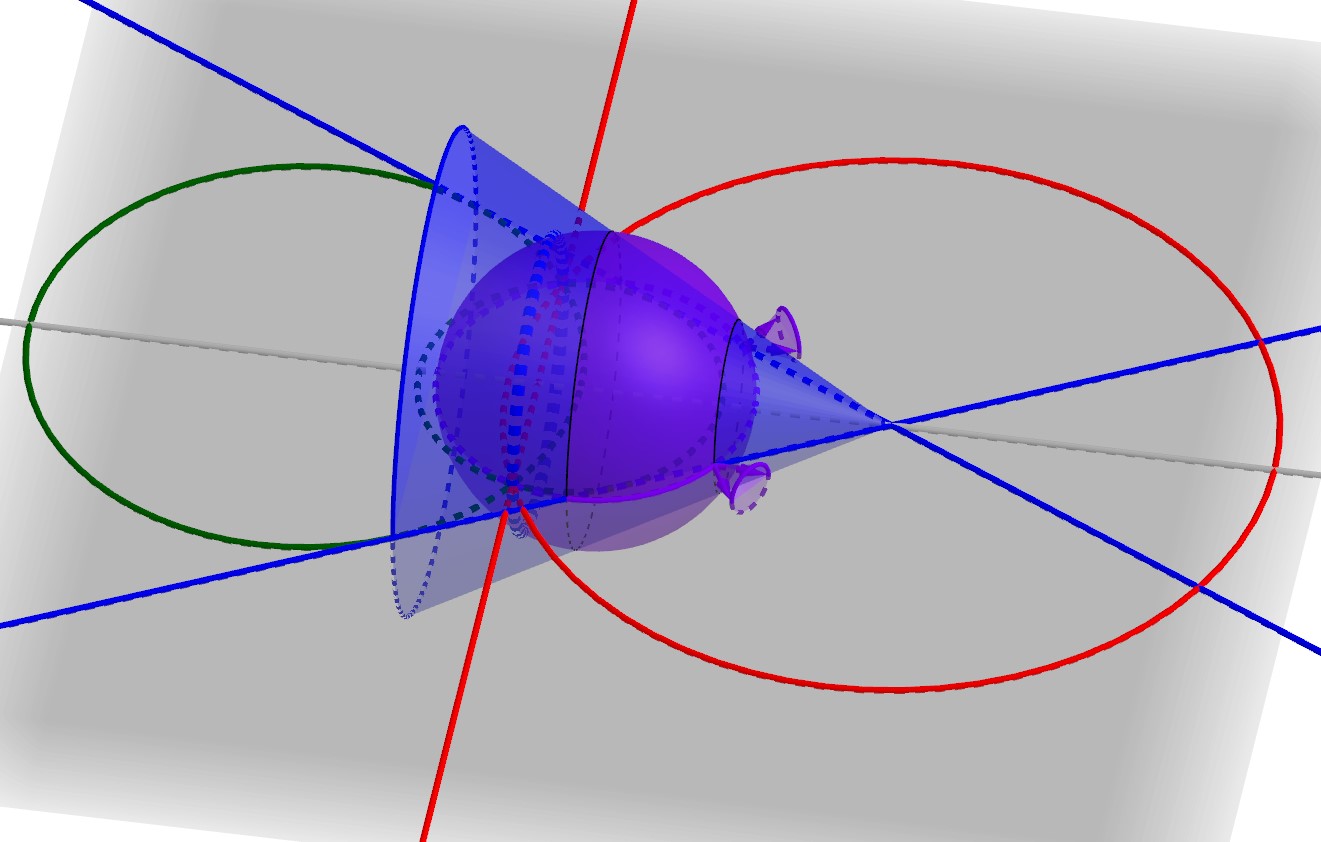}
  \caption{
    Representation of a pencil of angled spheres,
    which corresponds to a line $\ell$ in $\RP^5$.
    In blue the intersections of $\ell^\pol$ with $M^\pol$ and $B^\pol$, corresponding to a circle and an oriented cone, respectively.
    In red the intersections of $\ell$ with $M^\pol$ and $B^\pol$, corresponding to a sphere and an angled plane, respectively.
    In green the intersections of $\ell$ with $\lieq$, corresponding to two oriented spheres.
    In violet an arbitrary point on $\ell$, corresponding to an angled sphere in the pencil described by $\ell$.}
  \label{fig:liepencil}
\end{figure}

\section{Lie lift of principal binets}\label{sec:lielift}

We are now prepared to define the Lie lift of a principal binet.
Recall, that a principal binet comes as a pair of a binet and a bi*net.
Thus, both the Möbius lift and the Laguerre lift are well-defined.
We join both lifts to obtain the Lie lift (see Figure~\ref{fig:lie-lift}).

\begin{definition}[Lie lift]\label{def:lielift}
  Let $b: D \rightarrow \eucl$ be a principal binet.
  A line bicongruence $b_\lieq: D \rightarrow \linesof{\RP^5}$ is a \emph{Lie lift} of $b$ if
  \begin{enumerate}
	  \item the binet $b_\mobq: D \rightarrow M^\pol, \quad d \mapsto b_\lieq(d) \cap M^\pol$ is a Möbius lift of $b$,
	  \item and the binet $b_\blac: D \rightarrow B^\pol, \quad d \mapsto b_\lieq(d) \cap B^\pol$ is a Laguerre lift of $\square b$.\qedhere
  \end{enumerate}
\end{definition}

\begin{figure}[H]
  \centering
  \begin{overpic}[width=0.4\textwidth]{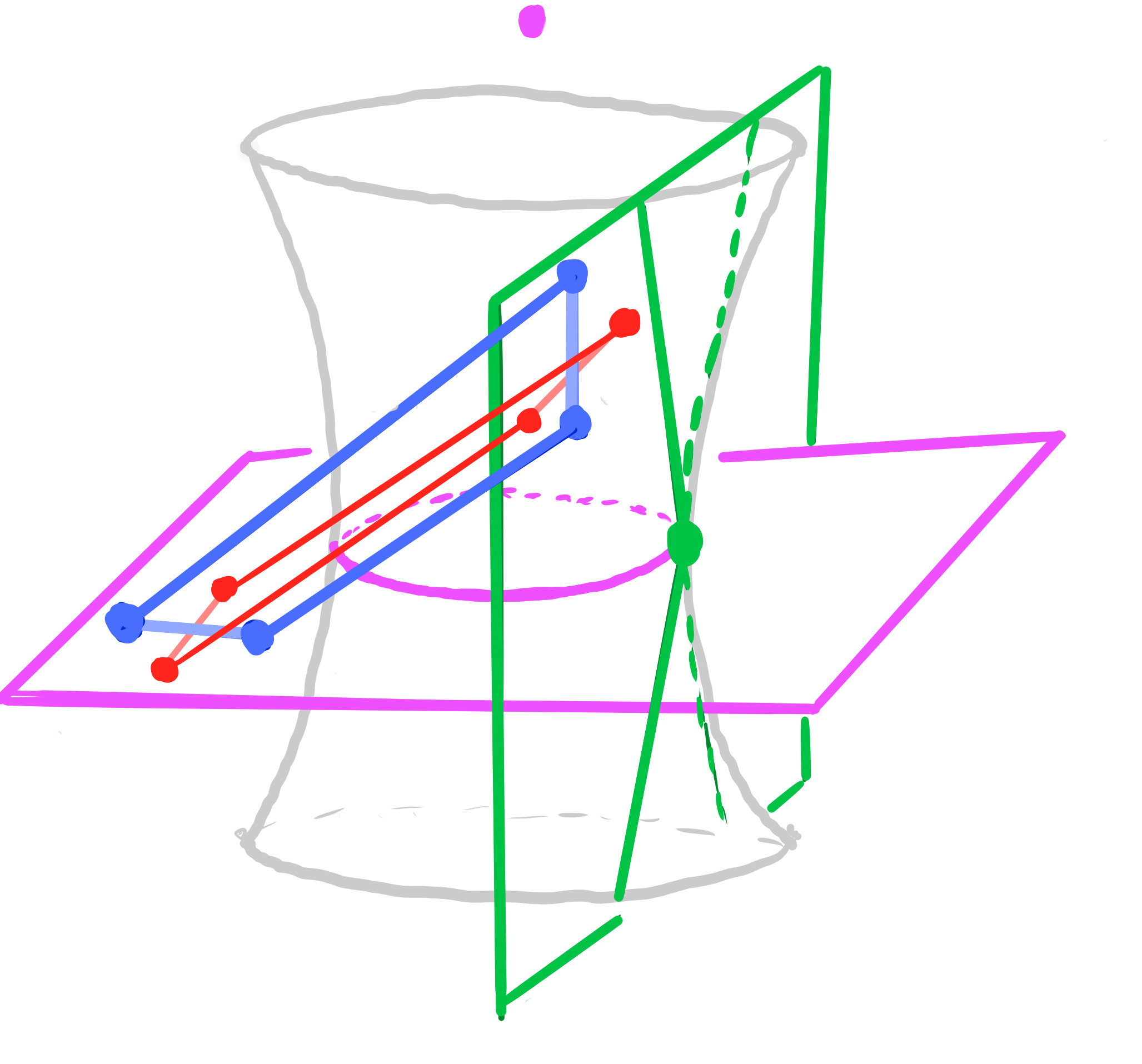}
    \put(2,82){$\lieq \subset \RP^5$}
    \put(48,89){$\color{magenta}\PM$}
    \put(89,46){$\color{magenta}\PM^\pol$}
    \put(33,35){$\color{magenta}\mobq$}
    \put(62,43){$\color{darkgreen}\PB$}
    \put(46,2){$\color{darkgreen}\PB^\pol$}
    \put(51,23){$\color{darkgreen}\blac$}
    \put(5,34){$\color{blue}b_\mobq$}
    \put(52,68){$\color{blue}b_\blac$}
    \put(32,61){$\color{blue}b_\lieq$}
  \end{overpic}
  \caption{
    Lie lift $b_\lieq$ of a principal binet $b$.
    The lines $b_\lieq$ are spanned by points of the Möbius lift $b_\mobq$ and the Laguerre lift $b_\blac$.
  }
  \label{fig:lie-lift}
\end{figure}

\begin{lemma} \label{lem:polarlielift}
  The Lie lift $b_\lieq$ of a principal binet $b$ is a polar bicongruence.
\end{lemma}
\proof{
  Consider two incident $d,d' \in D$. We have
  \begin{align}
	b_\lieq(d) = b_\mobq(d) \vee b_\blac(d), \quad b_\lieq(d') = b_\mobq(d') \vee b_\blac(d'),
  \end{align}
  and by definition of the Möbius and Laguerre lifts, $b_\mobq(d) \pol b_\mobq(d')$ and $b_\blac(d) \pol b_\blac(d')$.
  It remains to show that $b_\mobq(d) \pol b_\blac(d')$ (or equivalently $b_\blac(d) \pol b_\mobq(d')$).
  Let $X \in \mobq \subset M^\pol$ be the inverse stereographic projection of $b(d)$, that is $\pi_\eucl(X) = b(d)$.
  Then the Möbius lift of $b(d)$ satisfies $b_\mobq(d) \in L \coloneqq X \vee B$.
  The point $b(d)$ is contained in the plane $\square b(d')$.
  This is a special case of the relation of oriented contact described in Proposition~\ref{prop:liepolarity},
  and therefore $X \pol b_\blac(d')$.
  Moreover, since $b_\blac(d') \in \blac \subset \PB^\pol$, we have $\PB \pol b_\blac(d')$.
  Together the last two observations imply that $L \pol b_\blac(d')$, and therefore in particular $b_\mobq(d) \pol b_\blac(d')$.\qed
}

\begin{lemma}
  \label{lem:lie-lift-projection}
  Let $b_\lieq : D \rightarrow \RP^4 \setminus \PB^\pol$ be a polar bicongruence with respect to the Lie quadric $\lieq \subset \RP^5$.
  Then there is a principal binet $b$ such that 
  \begin{enumerate}
	  \item the binet $b_\mobq: D \rightarrow M^\pol, \quad d \mapsto b_\lieq(d) \cap M^\pol$ is a Möbius lift of $b$,
	  \item and the binet $b_\blac: D \rightarrow B^\pol, \quad d \mapsto b_\lieq(d) \cap B^\pol$ is a Laguerre lift of $\square b$.\qedhere
  \end{enumerate}  
\end{lemma}

\proof{
	Consider some $d\in D$ and define the line
	\begin{align}
		A(d) \coloneqq (b_\lieq(d) \vee M) \cap M^\pol,
	\end{align}
	which contains $b_\mobq(d)$.
	Since $b_\lieq$ is a polar line congruence we see that $b_\mobq(d')$ is polar to $A(d)$, whenever $d'$ is incident to $d$.
	Let $d_1,d_2,d_3,d_4 \in D$ be incident to $d\in D$, then it follows that the four points $b_\mobq(d_i)$ are polar to $A(d)$. Consequently, the four points $b_\mobq(d_i)$ are contained in a plane. Hence, $b_\mobq$ is a conjugate polar binet, and thus, by Lemma~\ref{lem:principalmobiuslift}, the Möbius lift of some principal binet $b$. Analogously, we find that $b_\blac$ is the Laguerre lift of some principal binet $b'$. It remains to show that $b = b'$. For this we need to show that if $d \sim d'$ then $b(d) \in \square b'(d')$. This is equivalent to $b(d)$ being polar to $b'(d')$. This can be concluded by taking the arguments in the proof of Lemma~\ref{lem:polarlielift} in reverse order. As a result, $b$ coincides with $b'$ and thus the claim is proven.\qed
}

\begin{theorem}\label{th:principallielift}
	There is a Lie lift $b_\lieq$ of a regular conjugate binet $b: D \rightarrow \pl$ if and only if $b$ is a principal binet.
\end{theorem}

\proof{
  If $b$ is not an orthogonal binet, then due to Theorem \ref{th:mobiusortho} the binet $b$ cannot have a Möbius lift,
  and therefore no Lie lift either.
  For the reverse direction, we show that any Möbius lift $b_\mobq$ of a principal binet $b$
  together with any Laguerre lift $b_\blac$ of $\square b$ define a line bicongruence
  \begin{align}
    b_\lieq: D \rightarrow \RP^5,\quad b_\lieq(d) = b_\mobq(d) \vee b_\blac(d).
    \label{eq:lielift}
  \end{align}
  Consider a cross $(v,f,v',f') \in C$ and define the two subspaces of $\RP^5$,
  \[
    E \coloneqq b_\lieq(v) \vee b_\lieq(v'), \qquad
    E' \coloneqq b_\lieq(f) \vee b_\lieq(f').
  \]
  Recall that Lemma~\ref{lem:regularmobiuslift} states that the Möbius lift of a regular orthogonal binet is regular, and therefore $b_\mobq(v) \neq b_\mobq(v')$.
  Since $b_\mobq(v) \in b_\lieq(v)$ and $b_\mobq(v') \in b_\lieq(v')$,
  we find that $b_\lieq(v) \neq b_\lieq(v')$.
  With the same reasoning we see that $b_\lieq(f) \neq b_\lieq(f')$ as well.
  Hence, $\dim E, \dim E' > 1$.
  Moreover, note that the proof of Lemma~\ref{lem:polarlielift} does not need the intersection property of line bicongruences,
  therefore the polarity statement holds for any $b_\lieq$ as defined in Equation~\eqref{eq:lielift}.
  This implies $E' \pol E$.
  Therefore $\dim E + \dim E' \leq 4$.
  Together with the previous inequality, this implies $\dim E = \dim E' = 2$, and thus $b_\lieq(v)$ intersects $b_\lieq(v')$ in a point.
  Analogously, $b_\lieq(f)$ intersects $b_\lieq(f')$ in a point as well.\qed
}

The proof of Theorem \ref{th:principallielift} shows that each principal binet has a 2-parameter family of Lie lifts.
One of the parameters comes from the 1-parameter family of Möbius lifts
and the other parameter comes from the 1-parameter family of Laguerre lifts. 

Let us turn to some geometric properties of the Lie lift (see Figure~\ref{fig:principal-coaxial}).
Each line $b_\lieq(d)$ of the Lie lift may geometrically be interpreted as a pencil of angled spheres
which is described by a coaxial circle-cone pair (see Section~\ref{sec:lie}).
The corresponding circle is given by
\[
  \circl(d) = \pi_\eucl(b_\lieq(d)^\pol \cap \mobq),
\]
which coincides with the circle \eqref{eq:orthocircle} introduced for the Möbius lift of principal binets,
and by Proposition~\ref{prop:circle-axis} is contained in the plane $\square b(d)$.
The corresponding cone is given by
\[
  \cone(d) = \xi_\opl(b_\lieq(d)^\pol \cap \blac),
\]
which coincides with the oriented cone \eqref{eq:principalcone} introduced for the Laguerre lift of principal binets,
and by Proposition~\ref{prop:cone-axis} has apex $b(d)$.
The line
\[
  N(d) = \pi_\eucl(b_\lieq(d)),
\]
contains all centers of the spheres in the angled sphere pencil and is also the common axis of $\circl(d)$ and $\cone(d)$.
Thus, $N(d)$ is the normal line of the principal binet $b$ at $d$, that is $b(d) \in N(d)$ and $N(d) \perp \square b(d)$
(see also Propositions~\ref{prop:circle-axis} and~\ref{prop:cone-axis}),
and we obtain the following proposition.
\begin{proposition}
  Let $b_\lieq$ be the Lie lift of a principal binet $b$.
  Then the projection $N \coloneqq \pi_\eucl \circ b_\lieq$
  is the normal bicongruence of $b$ (see Definition~\ref{def:normalbicongruence}).
\end{proposition}

\begin{figure}[H]
  \centering
  \includegraphics[width=0.6\textwidth]{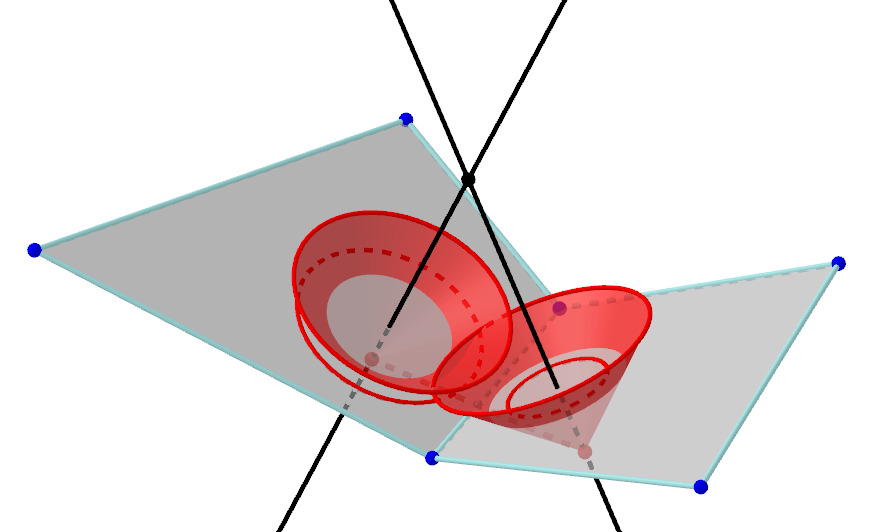}
  \caption{
    Coaxial cirlce-cone pairs of a principal binet.
  }
  \label{fig:principal-coaxial}
\end{figure}

\subsection*{Invariance}
Let $T$ be a Lie transformation and $\tilde T$ the corresponding projective transformation of $\RP^5$.
Similar to the Möbius and Laguerre invariance,
instead of applying the Lie transformation $T$ to a principal binet $b$,
we may apply $\tilde T$ to the Lie lift $b_\lieq$.
The image $\tilde T \circ b_\lieq$ is again a polar bicongruence
and its sections with $\PM^\pol$ and $\PB^\pol$ yield the points and planes of a principal binet.

\section{Principal curvature spheres} \label{sec:curvaturespheres}

We briefly introduce the notion of principal curvature spheres for principal binets (see Figure~\ref{fig:principal-sphere}).

\begin{definition}[Principal curvature spheres]
  Let $b$ be a principal binet with Möbius lift $b_\mobq$.
  Let $(v, f, v', f') \in C$ be a cross of $D$, and let
  \[
    H(v, v') \coloneqq \square b_\mobq(v) \vee \square b_\mobq(v'),\qquad
    H(f, f') \coloneqq \square b_\mobq(f) \vee \square b_\mobq(f'),
  \]
  be the two 3-dimensional subspaces in $\RP^4$ spanned by the two pairs of adjacent faces of the cross.
  Then we call the spheres
  \[
    \cs(v,v') \coloneqq \xi_\sp(H(v,v')^\pol), \qquad
    \cs(f,f') \coloneqq \xi_\sp(H(f,f')^\pol),
  \]
  the \emph{principal curvature spheres}
  of the edges $(v, v')$ and $(f, f')$ respectively.
\end{definition}
The definition of the principal curvature spheres is invariant under Möbius transformations.
It exhibits the following geometric properties.
\begin{proposition}
  \label{prop:curvature-spheres}
  Let $b$ be a principal binet with Möbius lift $b_\mobq$ and orthogonal sphere representation $b_\sp$.
  Let $(v, f, v', f') \in C$ be a cross of $D$.
  Then the following properties hold for the principal curvature sphere $\cs(v, v')$:
  \begin{enumerate}
  \item
    \label{prop:curvature-spheres1}
    $\cs(v, v')$ is orthogonal to the six spheres $b_\sp(\tilde f)$,
    where $\tilde f \in F$ is any face incident to $v$ or $v'$,
  \item
    \label{prop:curvature-spheres2}
    $\cs(v, v')$ contains the two circles $\circl(v)$ and $\circl(v')$, as defined in \eqref{eq:orthocircle},
  \item
    \label{prop:curvature-spheres3}
    the intersection of the lines of the normal congruence $N(v) \cap N(v')$
    is the center of $\cs(v,v')$.
  \end{enumerate}
  Analogous properties hold for the principal curvature sphere $\cs(f, f')$.
\end{proposition}
\begin{proof}\
  \begin{enumerate}
  \item
    The points $b_\mobq(\tilde f)$ are contained in the subspace $H(v, v')$.
  \item
    The planes $\square b_\mobq(v)$ and $\square b_\mobq(v')$ are contained in the subspace $H(v,v')$.
  \item
    The line $N(v)$ is the axis of the circle $\circl(v)$,
    and contains all centers of spheres through $\circl(v)$.
    Since $S(v,v')$ contains $\circl(v)$ and $\circl(v')$, its center is given by the intersection of $N(v)$ and $N(v')$.\qedhere
  \end{enumerate}
\end{proof}

\begin{figure}[H]
  \centering
  \includegraphics[width=0.49\textwidth]{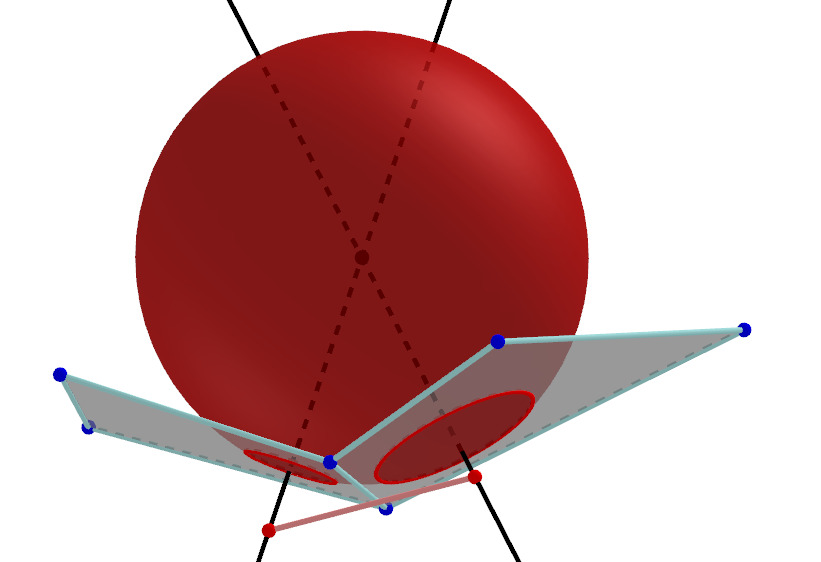}
  \includegraphics[width=0.49\textwidth]{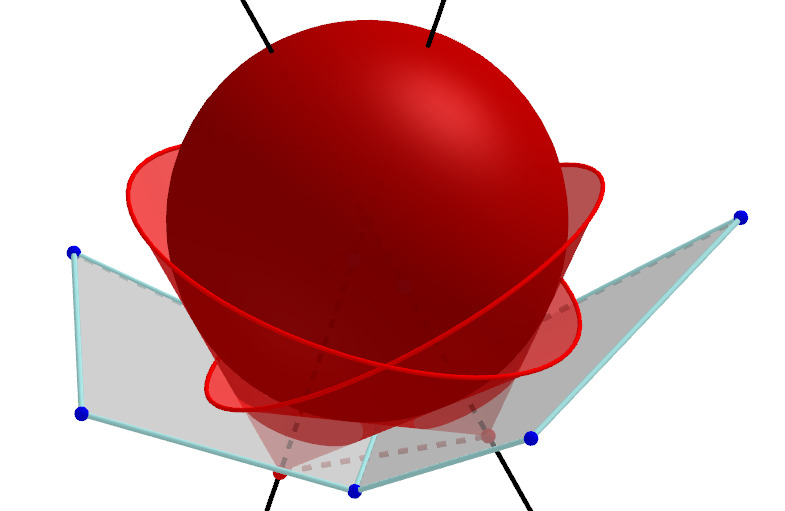}
  \caption{
    Left: (Möbius) Principal curvature sphere. Right: Laguerre curvature sphere.
  }
  \label{fig:principal-sphere}
\end{figure}

Property \ref{prop:curvature-spheres3} of Proposition~\ref{prop:curvature-spheres} is the discrete analog
of the centers of the principal curvature spheres being the focal points of the normal congruence of a surface,
while properties \ref{prop:curvature-spheres1} and \ref{prop:curvature-spheres2} describe the tangentiality to the surface.
\begin{remark}
  A corresponding notion of \emph{mean curvature sphere} may be introduced in the following way.
  Choose an orientation for the two principal curvature spheres $\cs(v,v')$ and $\cs(f,f')$ at a cross
  that represent the orientation of the discrete surface.
  Then the mean curvature sphere at this cross may be defined as the sphere $\tilde \cs$ in the pencil of spheres spanned by $\cs(v,v')$ and $\cs(f,f')$
  such that inversion in $\tilde \cs$ maps $\cs(v,v')$ to $\cs(f,f')$ while reversing their orientation.
\end{remark}

We also define a Lie lift of the principal curvature spheres to angled spheres.
\begin{definition}[Lie lift of principal curvature spheres]
  Let $b$ be a principal binet with Lie lift $b_\lieq$.
  Let $(v, f, v', f') \in C$ be a cross of $D$, and let
  \[
    S(v,v') \coloneqq b_\lieq(v) \cap b_\lieq(v'),\qquad
    S(f,f') \coloneqq b_\lieq(f) \cap b_\lieq(f'),
  \]
  be the two intersection points of adjacent lines of $b_\lieq$.
  Then we call the two corresponding angled spheres
  \[
    \acs(v,v') \coloneqq \xi_\asp(S(v,v')),\qquad
    \acs(f,f') \coloneqq \xi_\asp(S(f,f')),\qquad
  \]
  the \emph{Lie lift of principal curvature spheres} of the edge $(v,v')$ and $(f,f')$ respectively.
\end{definition}
They define a Lie lift of the previously defined principal curvature spheres in the sense
that the non-oriented spheres of the angled spheres coincide with the principal curvature sphere.
\begin{proposition}
  Let $b$ be a principal binet with Möbius lift $b_\mobq$ and corresponding Lie lift $b_\lieq$.
  Let $(v, f, v', f') \in C$ be a cross of $D$.
  Then
  \[
    \xi_\sp(S(v,v')) = \cs(v,v'),\qquad
    \xi_\sp(S(f,f')) = \cs(f,f').\qedhere
  \]
\end{proposition}
\begin{proof}
  It is easy to see that $H(v,v') = S(v,v')^\pol \cap M^\pol$,
  and thus
  \[
    \xi_\sp(S(v,v')) = \pi_\eucl(S(v,v')^\pol \cap \mobq) = \pi_\eucl(H(v,v') \cap \mobq) = \cs(v,v').\qedhere
  \]
\end{proof}
\begin{remark}
  The angled sphere of the Lie lift of a principal curvature sphere
  may be described by a pair of concentric spheres (see Section~\ref{sec:lie}),
  a non-oriented sphere which coincides with the principal curvature sphere $\cs(v,v')$,
  and an oriented sphere $\ocs(v,v') \coloneqq \xi_\osp(S(v,v'))$.
  The oriented sphere $\ocs(v,v')$ has the property
  that it is in oriented contact with two cones $\cone(v)$ and $\cone(v')$ along a circle each.
  This sphere may be viewed as a Laguerre geometric version of the discrete principal curvature sphere (see Figure~\ref{fig:principal-sphere}, right).
\end{remark}

\section{Circular-conical binets} \label{sec:circularconical}

In this section we discuss the special case of circular nets \cite{cdscircular, bobenkocircular}
and conical nets \cite{lpwywconical},
two classical discretizations of curvature line parametrizations.
\begin{definition}[Circular nets]
  \label{def:circular-net}
  A \emph{circular net} is a net $g: V \rightarrow \eucl$ such that the four image points of each quad are contained in a circle.
\end{definition}

Traditionally conical nets are introduced as nets.
In our context it is more convenient to view them as *nets.
Given a conical *net $h$, the traditional conical net is given by $\square^*h$.
\begin{definition}[Conical nets]
  \label{def:conical-net}
  A \emph{conical *net} is a *net $h: V \rightarrow \pl$ such that there exists an orientation of all planes
  for which the four planes of each quad are in oriented contact with a cone (of revolution).
\end{definition}

There is a well-known \emph{reflection construction}
to obtain a circular net from a conical *net and vice-versa (see Figure~\ref{fig:circular-conical-binet}) \cite{bsorganizing, pwconical}.
\begin{enumerate}
	\item
	\label{conical-from-circular}
	A conical *net $h$ is obtained from a circular net $g$ by reflecting an initial plane $h(v_0)$ through the point $g(v_0)$
	about the planes that are spanned by adjacent circle-axes of $g$.
	The composition of the four reflections incident to a face is the identity,
	and thus this construction yields a well-defined plane per vertex.
        These planes constitute the conical *net $h$.
	The four planes corresponding to four vertices incident to a face intersect in a common point on the circle-axis.
	These points constitute the traditional conical net $\square^*h$.
	\item
	\label{circular-from-conical}
	A circular net $g$ is obtained from a conical *net $h$ by reflecting an initial point $g(v_0)$ in the plane $h(v_0)$
	about the planes that are spanned by adjacent cone-axes of $h$.
	The composition of the four reflections incident to a vertex is the identity,
	and thus this construction yields a well-defined point per vertex.
	These points constitute the circular net $g$.
\end{enumerate}
The two constructions are symmetric in the following sense:
A conical *net $h$ can be obtained from a circular net $g$ by construction \ref{conical-from-circular}
if and only if $g$ can be obtained from $h$ by construction \ref{circular-from-conical}.
Indeed, this holds since the set of reflection planes coincide.
\begin{figure}[H]
  \centering
  \includegraphics[width=0.3\textwidth]{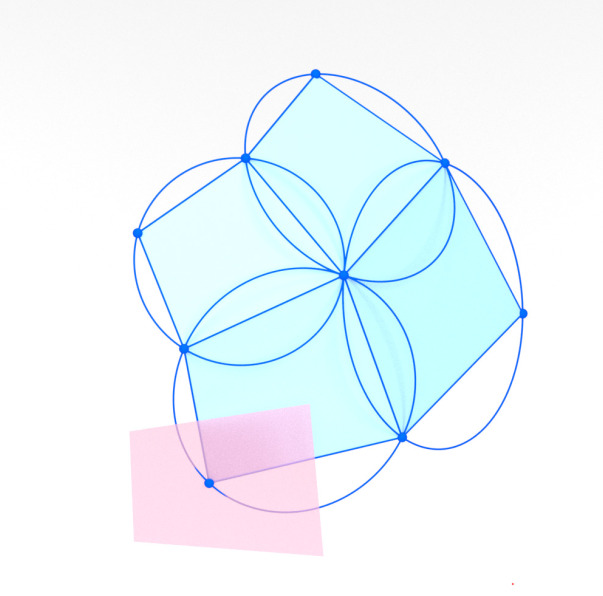}
  \includegraphics[width=0.3\textwidth]{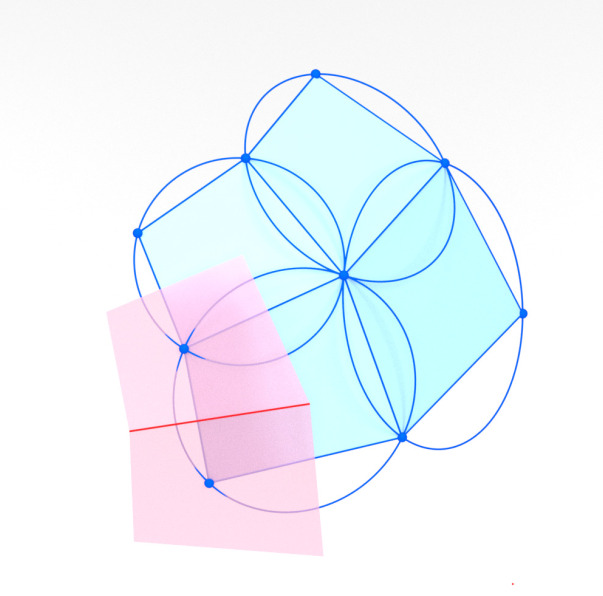}
  \includegraphics[width=0.3\textwidth]{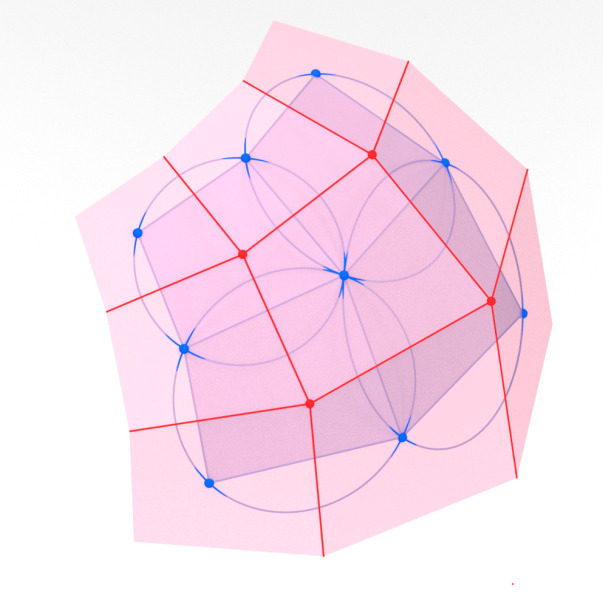}
  \caption{
    Obtaining a conical net from a circular net by the reflection construction.
    The resulting circular-conical binet is an orthogonal binet.
  }
  \label{fig:circular-conical-binet}
\end{figure}

Every pair of circular net $g$ and conical *net $h$ related by the reflection construction defines a principal binet.
Indeed, $g$ is a conjugate net, $h$ is a conjugate *net,
and the lines of the edges of $g$ are orthogonal to the lines of the edges of $\square^* h$.
\begin{definition}[Circular-conical binets]
	A principal binet $b: D \rightarrow \eucl$ is called a \emph{circular-conical binet} if
	\begin{enumerate}
		\item the restriction to $V$ is a circular net,
		\item the restriction of $\square b$ to $V$ is a conical *net,
		\item and these restrictions are related by the reflection construction described above.  \qedhere
	\end{enumerate}
\end{definition}

Let us discuss how the various lifts we introduced specialize in the case of circular-conical binets.
Firstly, there exists a canonical Möbius lift $b_\mobq$ of circular-conical binets $b$
such that its restriction to $V$ lies on $\mobq$ while the faces of its restriction to $F$ are tangential to $\mobq$.
Geometrically this means that the spheres $b_\sp(v)$ of the corresponding orthogonal sphere representation have radius zero for $v \in V$,
and that the circles $\circl(v)$ also have radius zero for $v \in V$.
On the other hand, the circles $\circl(f)$ for $f \in F$ coincide with the circles in the definition of circular nets.

\begin{theorem}
	\label{thm:circular-conical-moebius-lift}
	Let $b: D \rightarrow \eucl$ be a circular-conical binet.
	Then there exists a (canonical) Möbius lift $b_\mobq$ such that the restriction of $b_\mobq$ to $V$ is contained in $\mobq$.
\end{theorem}
\begin{proof}
  For each $f \in F$ the point $b(f)$ lies on the axis of the circle
  that contains the four points of the incident vertices $v_1, v_2, v_3, v_4 \inc f$.
  Thus, there exists a sphere $b_\sp(f)$ with center $b(f)$ that contains this circle,
  and therefore, contains the four points $b(v_1)$, $b(v_2)$, $b(v_3)$, $b(v_4)$.
  For $v \in V$ we choose the spheres $b_\sp(v) = b(v)$ as spheres of radius zero.
  As a consequence, $b_\sp$ is an orthogonal sphere representation that corresponds to a Möbius lift $b_\mobq$ with $b_\mobq(v) \in \mobq$ for all $v \in V$.
\end{proof}
\begin{theorem}
  \label{thm:canonical-moebius-faces}
  The canonical Möbius lift $b_\mobq$ of a circular-conical binet $b$
  satisfies that the faces of its restriction to $F$ are tangent to $\mobq$.
\end{theorem}
\begin{proof}
  For a given vertex $v\in V$ there are four faces $f_1,f_2,f_3,f_4 \in F$ incident to $v$.
  The points $b(f_1), b(f_2), b(f_3), b(f_4)$ lie in a plane which contains the point $b(v)$.
  In the lift the points $b_\mobq(f_1), b_\mobq(f_2), b_\mobq(f_3), b_\mobq(f_4)$ still lie in a plane
  which contains the point $b_\mobq(v)$.
  On the other hand, this plane is contained in $b_\mobq(v)^\pol$ and thus constitutes a tangent plane of $\mobq$.
\end{proof}

Similarly, there exists a canonical Laguerre lift of circular-conical binets
such that its restriction to $V$ lies on $\blac$ while the faces of its restriction to $F$ are tangential to $\blac$.
Geometrically this means that the angled planes $b_\apl(v)$ become oriented planes (angle equal to zero) for $v \in V$
where the corresponding non-oriented plane is given by $\square b(v)$.
Equivalently, the normal circles $b_\nci(v)$ of the corresponding orthogonal circle representation $b_\nci$ have radius zero for $v \in V$,
and thus the cones $\cone(v)$ degenerate to oriented planes for $v \in V$, which coincide with the oriented planes $b_\apl(v)$ of the conical *net.
On the other hand, the cones $\cone(f)$ for $f \in F$ coincide with the cones in the definition of conical *nets.
\begin{theorem}
	\label{thm:circular-conical-laguerre-lift}
	Let $b: D \rightarrow \eucl$ be a circular-conical binet.
	Then there exists a (canonical) Laguerre lift $b_\blac$ such that the restriction of $b_\blac$ to $V$ is contained in $\blac$.
\end{theorem}
\begin{proof}
  For such a lift to exist we need that for $v\in V$ the normal circle $b_\nci(v)$
  in the orthogonal circle representation has radius zero, i.e., is a point.
  Moreover, for $f\in F$ the circle  $b_\nci(f)$ needs to pass through the four points $b_\nci(v_1)$, $b_\nci(v_2)$, $b_\nci(v_3)$, $b_\nci(v_4)$
  of the four vertices adjacent to $f$.
  This circle exists since it corresponds to the normals of the cone of the conical net.
\end{proof}
\begin{theorem}
  The canonical Laguerre lift $b_\blac$ of a circular-conical binet $b$
  satisfies that the faces of its restriction to $F$ are tangent to $\blac$.
\end{theorem}
\begin{proof}
  Analogous to the proof of Theorem~\ref{thm:canonical-moebius-faces}.
\end{proof}

Together the canonical Möbius lift and the canonical Laguerre lift of a circular-conical binet
constitute a canonical Lie lift which satisfies that its restriction to $V$ is contained in $\lieq$.
\begin{theorem}
	\label{thm:circular-conical-lie-lift}
	Let $b: D \rightarrow \eucl$ be a circular-conical binet.
	Then there exists a (canonical) Lie lift $b_\lieq$ such that the restriction of $b_\lieq$ to $V$ is contained in $\lieq$.
\end{theorem}
\begin{proof}
  The canonical Lie lift is simply the vertex-wise span of the canonical Möbius lift $b_\mobq$ and the canonical Laguerre lift $b_\blac$. Therefore, for each $v \in V$ the points $b_\mobq(v)$ and $b_\blac(v)$ are contained in $\lieq$. It remains to show that $b_\blac(v)$ is polar to $b_\mobq(v)$. However, this follows immediately since they represent oriented spheres that are in the same contact element (defined by $b(v)$ and $\square b(v)$).
\end{proof}
A (discrete) line congruence which is contained in $\lieq$ is also called a \emph{(discrete) isotropic line congruence}.
The observation that circular-conical binets correspond to isotropic line congruences was already made in \cite{bsorganizing,ddgbook},
where they are treated as \emph{contact element nets} without the binet interpretation.
Moreover, the converse is also true: every isotropic line congruence is the restriction to $V$ of the canonical Lie lift of a circular-conical binet.

\begin{remark}
  \label{rem:circular-conical-freedom}
  In each of the two reflection constructions \ref{conical-from-circular} and \ref{circular-from-conical}
  there are 2 degrees of freedom for the initial value
  since we require the initial plane to contain a point, or the initial point to lie in a plane, respectively.
  This condition may be dropped, which leads to a reflection construction with 3 degrees of freedom.
  The resulting pair of net and *net still define a principal binet, and
  the canonical Möbius and Laguerre lifts do also still exist.
  However, together those lifts do not define an isotropic line congruence since the canonical Möbius and Laguerre lifts are not polar per vertex anymore.
  Thus, in general the canonical Lie lift does not exist in this case.
\end{remark}

Consider again a circular-conical binet $b$.
For an edge $(v, v')$ with $v \in V$ the Lie lift of the principal curvature sphere $\acs(v,v')$ corresponds to the intersection of
two adjacent isotropic lines $b_\lieq(v)$ and $b_\lieq(v')$.
The corresponding principal curvature sphere $\cs(v,v')$ is the common sphere of the two contact elements at $v$ and $v'$.
This means that it contains the two points $b(v)$ and $b(v')$ and is tangent to the two planes $\square b(v)$ and $\square b(v')$.
The corresponding Laguerre geometric principal curvature sphere $\ocs(v,v')$ has the same radius (and center) in this case.
This principal curvature sphere for circular-conical binets was already considered in \cite[Figure~3.9]{ddgbook}.

For a dual edge $(f, f')$ with $f \in F$ the Lie lift of the principal curvature sphere $\acs(f,f')$ corresponds to the intersection of
the two (generally) non-isotropic lines $b_\lieq(f)$ and $b_\lieq(f')$.
The corresponding principal curvature sphere $\cs(f,f')$ is the sphere containing the two circles $\circl(v)$ and $\circl(v')$ of the circular net.
The corresponding Laguerre geometric principal curvature sphere is a concentric oriented sphere $\ocs(f,f')$
which is in oriented contact with the two cones $\cone(f)$ and $\cone(f')$ of the conical *net along a circle, respectively.
The principal curvature sphere $\cs(f,f')$ also occurs in the literature \cite[Figure~3.2]{ddgbook}.

\section{Spaces of binets} \label{sec:spaces}

In this section, we discuss the spaces of the various special cases of binets.

A conjugate binet is the union of a conjugate net defined on $V$ and a conjugate net defined on $F$. Hence, the space of conjugate binets is the space of conjugate binets squared. For conjugate nets it is well known \cite{dsqnet}, that we may prescribe arbitrary data on the coordinate axes, and then still have 2-dimensional freedom for each point not on the coordinate axes.

Next, we discuss orthogonal binets.
Assume $g: F \rightarrow \eucl$ is an arbitrary net, let us try to find an orthogonal binet $b: D \rightarrow \eucl$ that restricts to $g$ on $F$. We may choose arbitrary (generic) points for $b$ on the coordinate axes of $V$. The remaining points of $b$ on $V$ away from the coordinate axes cannot be chosen arbitrarily anymore.
Instead, we can propagate the initial data quad by quad in the following way.
If we know three points $b(i,j), b(i+1,j), b(i, j+1)$ of a quad in the positive quadrant of $\Z^2$
the condition on $b(i+1,j+1)$ is that the lines
\begin{align}
	b(i+1,j+1) \vee b(i+1,j), \quad b(i+1,j+1) \vee b(i,j+1).
\end{align}
are orthogonal to the corresponding lines of $g$.
Thus, we can choose each of these lines in a plane and $b(i+1,j+1)$ in the intersection of these planes,
which leaves one degree of freedom.
The propagation works analogously for the other quadrants of $\Z^2$.

Thus, the notion of an orthogonal binet is not a very rigid notion.
In particular, an orthogonal binet $b : D \rightarrow \eucl$ remains an orthogonal binet
if one applies a translation to either of its restrictions to $V$ or to $F$ separately. Since one of these restrictions does not determine the other restriction, we view
 a binet -- and in particular an orthogonal binet --
as a pair of nets which \emph{together} constitutes \emph{one} discrete surface.

We may also discuss the space of orthogonal binets via their Möbius lifts. Since Möbius lifts are a special case of polar binets -- and we encounter other special cases as well -- we begin by discussing the space of polar binets in general.

Assume $g: F \rightarrow \RP^n$ is an arbitrary net, and let us try to find a polar binet $b: D \rightarrow \RP^n$ that restricts to $g$ on $F$. The condition for $b$ at $v\in V$ is that
\begin{equation}
  \label{eq:polar-binet-construction}
  b(v) \in b(f_1)^\pol \cap b(f_2)^\pol \cap b(f_3)^\pol \cap b(f_4)^\pol
  = \left(b(f_1) \vee b(f_2) \vee b(f_3) \vee b(f_4)\right)^\pol,
\end{equation}
where $f_1,f_2,f_3,f_4 \in F$ are the four faces incident to $v$.
Generically, $b(f_1), b(f_2), b(f_3), b(f_4)$ span a 4-dimensional subspace.
Thus, if $ n \geq 4$, then there is an at least $(n-4)$-dimensional freedom for each point $b(v)$.
In particular, if $n = 4$, generically, the polar binet $b$ is uniquely determined by $g$.

If $n = 3$, it is not generally possible to find a polar binet $b$ that restricts to $g$.
Instead, this is possible if and only if $b(f_1), b(f_2), b(f_3), b(f_4)$ span a plane,
i.e., if $g$ is a conjugate net (see Section \ref{sec:conjugatebinets}).

In the cases $n=1,2$ the conditions on $g$ become even more restrictive, but since we do not need these cases here, we do not discuss them further.

Now, let us consider Möbius lifts of orthogonal binets. Since a Möbius lift is a polar binet in $\RP^4$, our discussion of the space of polar binets for $n=4$ applies.
This means we can construct a polar binet $b_\mobq$ in $\RP^4$ with respect to the Möbius quadric $\mobq$ by choosing $b_\mobq$ arbitrarily on $F$
and then, generically, $b_\mobq$ is completely determined.
In terms of the corresponding spheres $b_\sp$, each sphere on $V$ is the unique sphere orthogonal to four incident spheres on $F$.
By Lemma \ref{lem:moebius-lift-projection}, the projection $\pi_\eucl \circ b_\mobq$ is an orthogonal binet.
In contrast, we saw that if we fix the values of an orthogonal binet on $F$, there is still one degree of freedom for each value on $V$.

Next, we discuss the space of orthogonal bi*nets and the differences to the space of orthogonal binets. Recall that every orthogonal bi*net defines a normal binet.
A normal binet is a polar net with respect to $\unis$ in $\eucl \subset \RP^3$,
thus our discussion of the space of polar binets applies in the case $n=3$.
This has two immediate consequences.
First, this means that the restriction of a normal binet to $V$ (or $F$) uniquely determines the whole normal binet.
Second, the normal binet is always a conjugate binet, see Lemma~\ref{lem:normalconjugate}.

As a result, it is not possible to choose an orthogonal bi*net arbitrarily on $F$ and then complete it on $V$. Instead, the restriction to $F$ must be chosen such that the restriction of the corresponding normal binet to $F$ is a conjugate net. Consequently, the entire normal binet is uniquely determined, and then there is a 1-dimensional freedom for each plane of the orthogonal bi*net on $V$ (because the normal vectors are predetermined).

The notion of a principal binet is a  more rigid notion than the notion of orthogonal or conjugate binets.
Assume $g: F \rightarrow \eucl$ is an arbitrary conjugate net, and let us try to find a principal binet $b: D \rightarrow \eucl$ that restricts to $g$ on $F$.
We may choose arbitrary (generic) points for $b$ on the coordinate axes of $V$,
and propagate the initial data quad by quad.
The orthogonality condition on $b(i+1,j+1)$ is that the lines
\begin{align}
	b(i+1,j+1) \vee b(i+1,j), \quad b(i+1,j+1) \vee b(i,j+1),
\end{align}
are orthogonal to the corresponding lines of $g$, and the conjugate condition is that 
\begin{align}
	b(i+1,j+1) \in  b(i+1,j) \vee b(i,j) \vee b(i,j+1).
\end{align}
This determines each point in $V$ away from the coordinate axes uniquely. This is in contrast to orthogonal binets, for which there is a 1-dimensional degree of freedom per vertex and conjugate binets for which there is a 2-parameter freedom per vertex.

\begin{remark}
  Furthermore, the space of principal binets is significantly larger than the space of circular-conical binets. In particular, for principal binets the restriction to $F$ may be an arbitrary conjugate binet, which means there are 2 degrees of freedom per face. On the other hand, for circular-conical binets the restriction to $F$ needs to be a conical net, which means there is only 1 degree of freedom per face.
\end{remark}

Let us briefly discuss the difference of the Möbius lift of a principal binet compared to the Möbius lift of an orthogonal binet.
Previously, we discussed the space of Möbius lifts, which are polar binets with respect to $\mobq$ in $\RP^4$.
However, this discussion was based on the genericity assumption that the space
\begin{align}
  b_\mobq(v_1) \vee b_\mobq(v_2) \vee b_\mobq(v_3) \vee b_\mobq(v_4)
\end{align}
is 3-dimensional. On the other hand, in the principal case this space is 2-dimensional due to Theorem~\ref{thm:conjugatemobiuslift}.
Thus, $b_\mobq$ is a conjugate binet and the 2-dimensional space above is given by $\square b_\mobq(f)$.
Hence, the space of Möbius lifts of principal binets is that of binets in $\RP^4$ that are conjugate and polar with respect to $\mobq$. Assume $g: F \rightarrow \RP^4$ is an arbitrary conjugate net, and let us try to find a polar and conjugate binet $p: D \rightarrow \RP^4$ that restricts to $g$ on $F$. The condition for $p$ at $v\in V$ is that $b(v)$ is in the line
\begin{align}
	p(f_1)^\pol \cap p(f_2)^\pol \cap p(f_3)^\pol \cap p(f_4)^\pol,
\end{align}
and also that $p(v)$ is in the four adjacent planes $\square p(f_1), \Box p(f_2), \square p(f_3), \square p(f_4)$. Therefore, we may choose $b(v)$ arbitrarily on the coordinate axes of $\Z^2$, and then for $v\in V$ not on the coordinate axes the value $b(v)$ is already determined. Effectively, the principal condition eliminates the 1-dimensional freedom per vertex that we had in the case of orthogonal binets.

The case of Laguerre lifts of principal binets is analogous to that of Möbius lifts of principal binets, since the Laguerre lift of a principal binet is a conjugate and polar binet with respect to $\blac$ in $\RP^4$. Thus, the restriction of the Laguerre lift to $F$ is a conjugate net. This implies that the normal binet is a conjugate binet, which is therefore not an additional requirement. Consequently, with the same arguments as for the Möbius lift of a principal binet we see that the restriction to $F$ determines the restriction to $V$ uniquely.

Finally, we discuss the space of line bicongruences that are Lie lifts of a principal binet. That means we consider the space of polar line bicongruences $\ell: D \rightarrow \RP^5$. Assume $g: F \rightarrow \RP^5$ is an arbitrary line congruence, and let us try to find a polar line bicongruence $\ell: D \rightarrow \RP^5$ that restricts to $g$ on $F$. The condition is that for every $v\in V$
\begin{align}
  \ell(v) \pol \ell(f_1), \ell(f_2), \ell(f_3), \ell(f_4),
\end{align}
where $f_1,f_2,f_3,f_4 \in F$ are the four faces incident to $v$. Generically, there is only one such line, therefore $g$ uniquely determines $\ell$. This is a straight-forward observation, but also a strong statement. Hence, we phrase it in a theorem.

\begin{theorem}
	Every line congruence $g: F \rightarrow \linesof{\RP^5}$ defines a principal binet $b: D \rightarrow \eucl$ via
	\begin{align}
		b(f) &= \pi_{\eucl}(\PM^\pol \cap g(f)),\\
		b(v) &= \pi_{\eucl}(\PM^\pol \cap g(f_1)^\pol \cap g(f_2)^\pol \cap g(f_3)^\pol \cap g(f_4)^\pol),
	\end{align}
	for all $f\in F, v\in V$ such that $f_1,f_2,f_3,f_4$ are incident to $v$.
\end{theorem}

\section{Multi-dimensional consistency} \label{sec:consistency}

Let us begin with the necessary definitions of combinatorics, which are the generalizations of vertices, edges, faces and crosses of $\Z^N$.
\begin{align*}
	V_N &\coloneqq V(\Z^N) = \Z^N,\\
	E_N &\coloneqq E(\Z^N) \simeq \cup_{i=1}^N \Z^N,\\
	F_N &\coloneqq F(\Z^N) \simeq \cup_{i=1}^N \cup_{j=1+1}^N \Z^N,\\
	D_N &\coloneqq V_N \cup F_N,\\
	C_N &\coloneqq \set{(v,f,v',f')}{v, v' \in V_N, \, f, f' \in F_N, \, v, v' \inc f, f'}.
\end{align*}
Note that for $\Z^2$, up to permutation there is only one cross per edge. In $\Z^N$ there are $\binom{2N-2}{2}$ crosses per edge up to permutation (see Figure~\ref{fig:crosses-types}).

\begin{figure}[H]
  \centering
  \includegraphics[width=0.3\textwidth, trim={0 90 0 20}, clip]{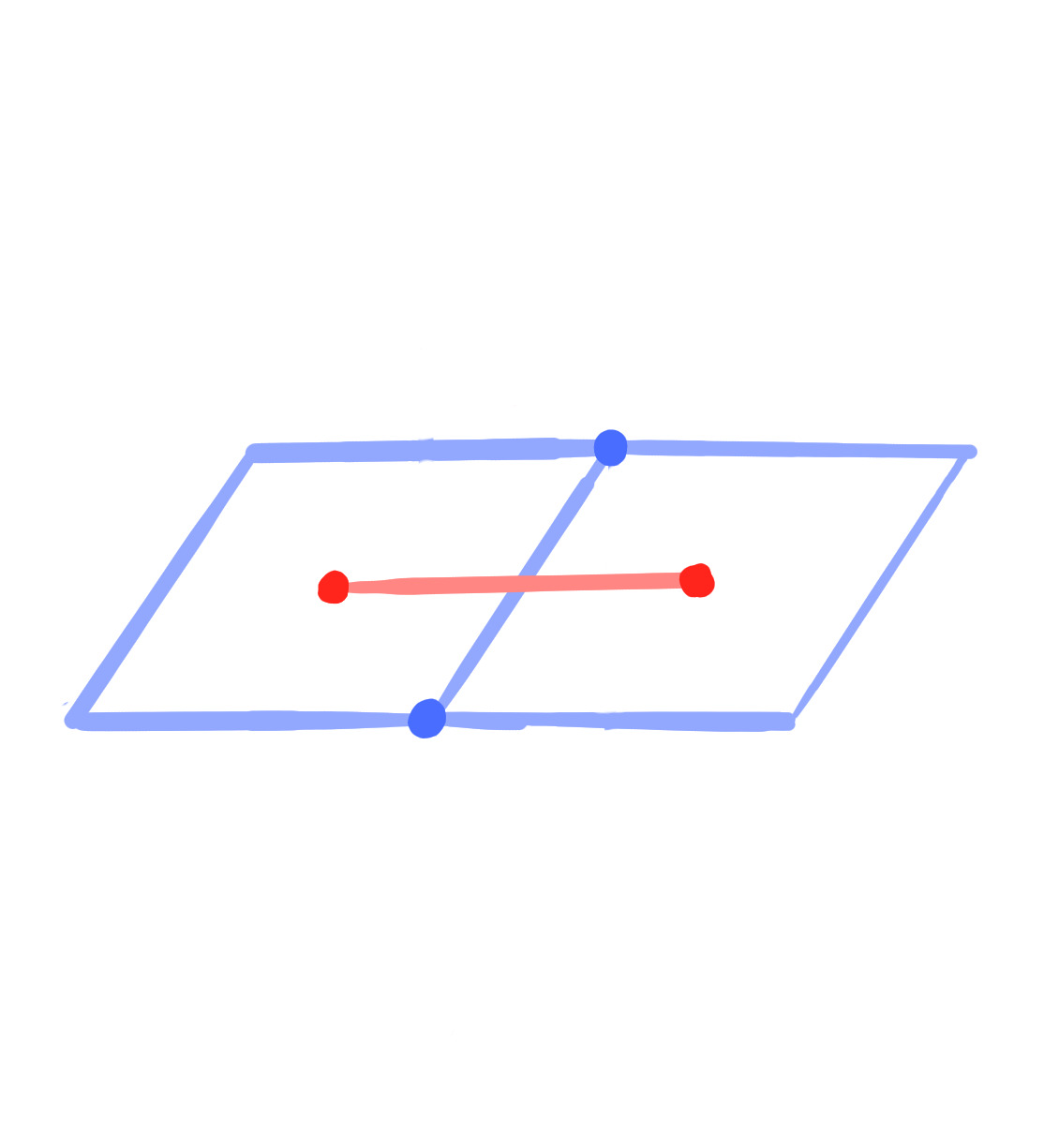}
  \includegraphics[width=0.3\textwidth, trim={0 90 0 20}, clip]{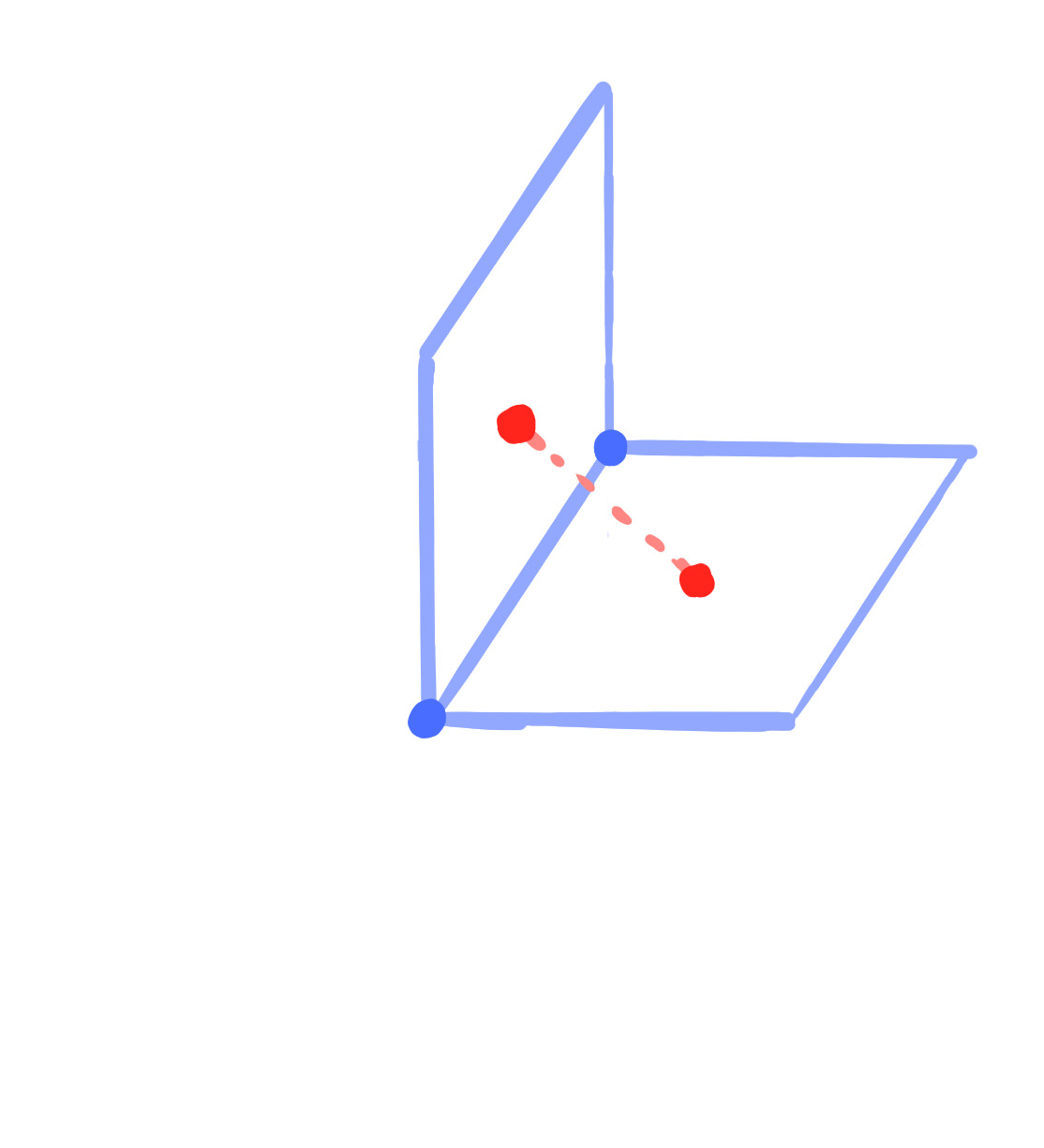}
  \caption{
    The two types of crosses in $\Z^3$.
    This amounts to a total of six crosses per edge of $\Z^3$.}
  \label{fig:crosses-types}
\end{figure}

Since we want to keep this section from becoming overly technical, we will not consider regularity assumptions on nets, *nets, binets and so forth in this section. Instead we assume that the objects involved are generic enough so that the respective constructions make sense.

Multi-dimensional consistency is a form of discrete integrability \cite{ddgbook}. In this section we show the multi-dimensional consistency of principal binets. Multi-dimensional consistency is about generalizing the definitions of objects on $\Z^2$ to $\Z^N$ with $N>2$. To do this, one needs to find a suitable definition for a generalization and then show that this definition does not come with contradictions, that is to show that the space of such generalized objects is non-empty. 

More specifically, a system on $\Z^N$ is a \emph{$k$-dimensional system} if it is determined by $(k-1)$-dimensional initial data, also known as Cauchy data. The initial data is usually given on the coordinate planes, which then uniquely determines the remaining data on $\Z^k$. Moreover, if $(k-1)$-dimensional initial data also uniquely and without contradictions determines solutions on $\Z^m$ for $m > k$, the system is called \emph{multi-dimensionally} consistent.

We begin by generalizing conjugate binets, which we previously defined on $D_2$ in Section \ref{sec:conjugatebinets} (see Figure~\ref{fig:conjugate-face-nets}).

\begin{definition}[Conjugate nets]
  A \emph{conjugate net} is a map $g: V_N \rightarrow \RP^n$ such that for every $f\in F_N$  the span $\bigvee_{v \inc f} g(v)$ is a plane. A \emph{conjugate face-net} $g: F_N \rightarrow \RP^n$ is a map such that for every $v\in V_N$  the span $\bigvee_{f \inc v} g(v)$ is a plane. A \emph{conjugate binet} is a map $b: D_N \rightarrow \RP^n$ such that the restriction to $V_N$ is a conjugate net and the restriction to $F_N$ is a conjugate face-net.
\end{definition}

\begin{figure}[H]
  \centering
  \raisebox{0.15\textwidth}{
    \begin{overpic}[width=0.2\textwidth]{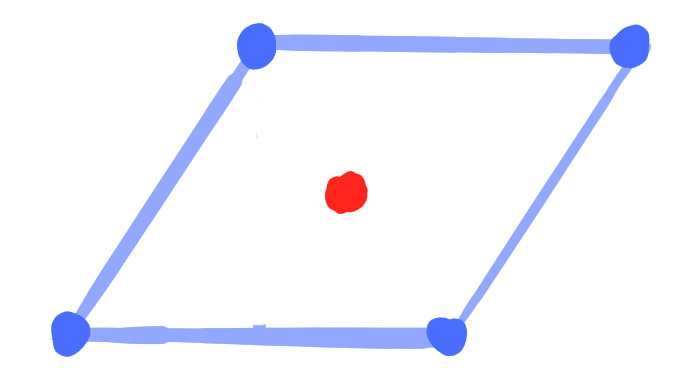}
      \put(53,22){$\color{red}f$}
    \end{overpic}
  }
  \begin{overpic}[width=0.4\textwidth]{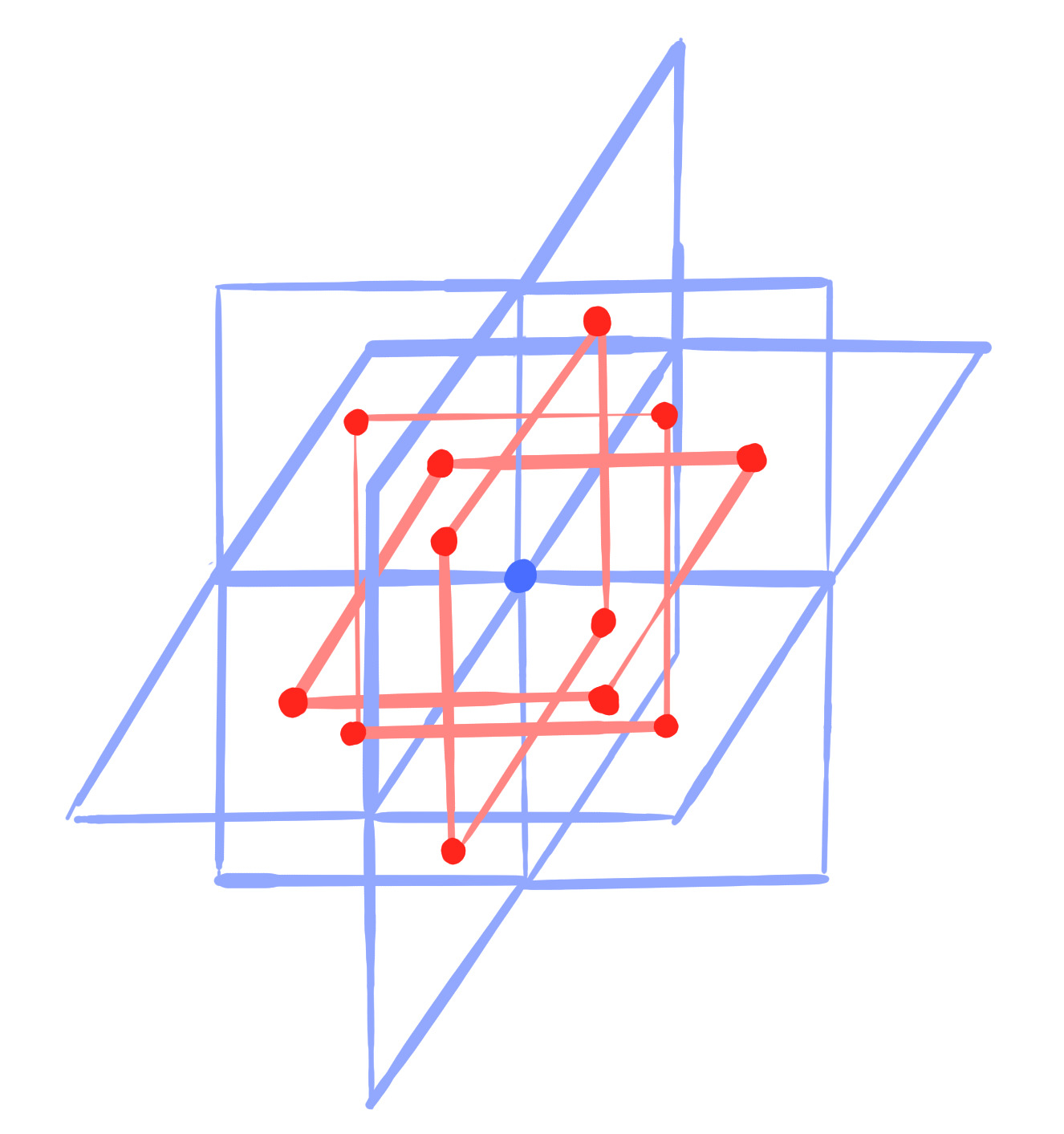}
    \put(49,50){$\color{blue}v$}
  \end{overpic}
  \caption{
    Combinatorics for the conditions of conjugate nets and conjugate face-nets in $\Z^3$
    Left: A face $f \in F_3$ has four incident vertices.
    The corresponding points of a conjugate net are coplanar.
    Right: A vertex $v \in V_3 = \Z^3$ has twelve incident faces.
    The corresponding points of a conjugate face-net are coplanar
  }
  \label{fig:conjugate-face-nets}
\end{figure}

The generalization of conjugate nets to $\Z^N$ is due to Doliwa and Santini \cite{dsqnet}, who also proved the multi-dimensional consistency of conjugate nets, as given by the next theorem.

\begin{theorem} \label{th:conjugatenetsconsistency}
	Conjugate nets are multi-dimensionally consistent 3D-systems.
\end{theorem}

Before we show consistency for conjugate face-nets and conjugate binets, we also introduce conjugate *nets, conjugate face-*nets and conjugate bi*nets.
\begin{definition}[Conjugate *nets]
  A \emph{conjugate *net} is a map $g: V_N \rightarrow \planesof\RP^n$ such that adjacent planes intersect in a line, and such that for every $f\in F_N$  the intersection $\bigcap_{v \inc f} g(v)$ is a point. A \emph{conjugate face-*net} is a map $g: F_N \rightarrow \planesof{\RP^n}$ such that adjacent planes intersect in a line, and such that for every $v\in V_N$  the intersection $\bigcap_{f \inc v} g(v)$ is a point. A \emph{conjugate bi*net} is a map $b: D_N \rightarrow \RP^n$ such that the restriction to $V_N$ is a conjugate *net and the restriction to $F_N$ is a conjugate face-*net.
\end{definition}

If $b$ is a conjugate binet then $\square b$ is the conjugate bi*net consisting of the planes of $b$, analogous to the $D_2$ case given in Section~\ref{sec:conjugate-bi-star-nets}. Vice versa, if $b$ is a conjugate bi*net then $\square^* b$ is the conjugate binet consisting of the intersection points of $b$.

While the definitions of conjugate nets appear to be very symmetric with respect to $V_N \leftrightarrow F_N$, the set of vertices and faces is not isomorphic for $N>2$. Therefore conjugate face-nets are a priori different from conjugate nets. However, there is actually a close relation, that we outline in the following. Let $F_N^{ij}$ be the faces that are spanned by the $i,j$-directions in $\Z^N$, hence $F_N^{ij} \simeq \Z^N$ (see Figure~\ref{fig:face-identification}).
We make this identification explicit via
\begin{align}
	F_N^{ij} \ni f = \{r, r+e_1, r + e_1 +e_2, r+e_2\} \Leftrightarrow r \in \Z^N.
\end{align}

\begin{figure}[H]
  \centering
  \includegraphics[width=0.35\textwidth]{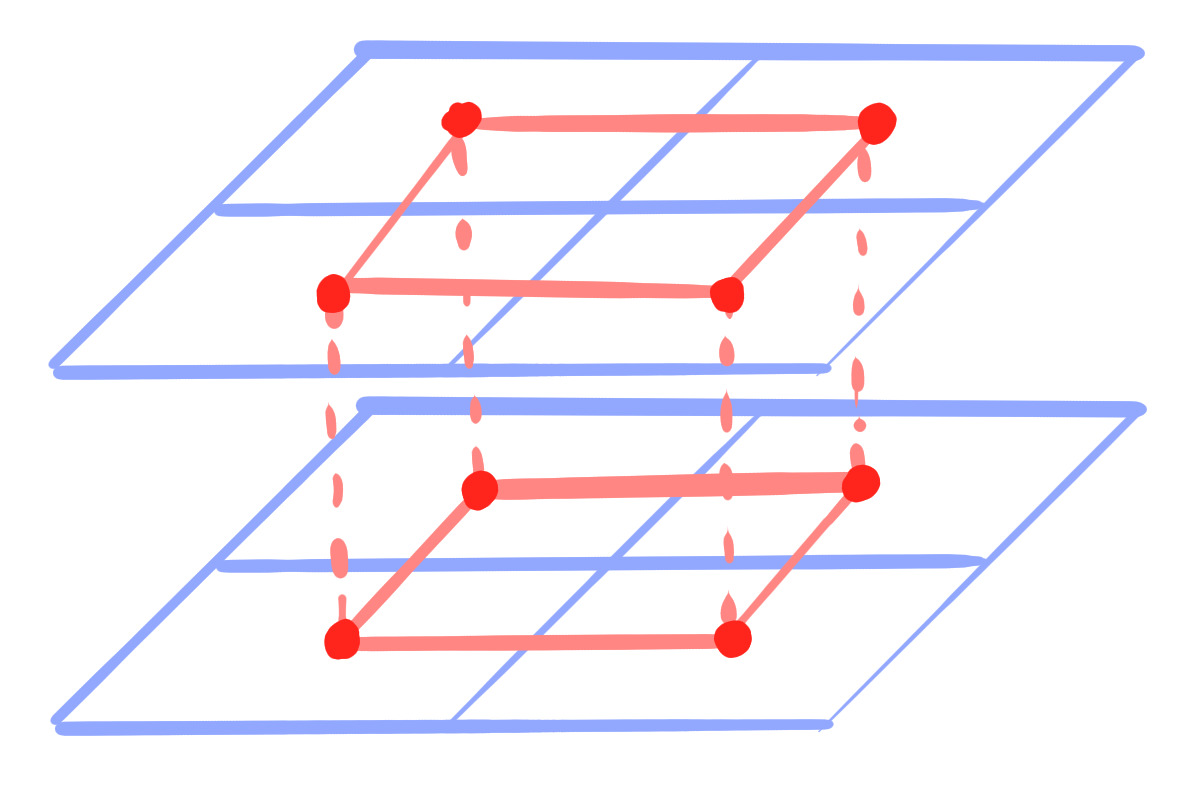}
  \caption{
    Identification of $F_3^{12}$  with $\Z^3$.
    In blue two layers of $12$-faces of $\Z^3$.
    In red $F_3^{12} \simeq \Z^3$.
    The dotted edges show the missing edges for the identification with $\Z^3$.
  }
  \label{fig:face-identification}
\end{figure}

For a map $g$ defined on $F_N$, we denote by $g^{ij}$ its restriction to the $ij$-faces. So for $r\in \Z^n$ we denote by $g^{ij}(r)$ the value of $g$ on the face $\{r, r+e_i, r + e_i +e_j, r+e_j\}$.

\begin{lemma}\label{lem:conjugatefacetoconjugate}
	The restriction $g^{ij}: \Z^N \rightarrow \RP^n$  of a conjugate face-net $g: F_N \rightarrow \RP^n$ to $F_N^{ij} \simeq \Z^N$ is a conjugate net.
\end{lemma}
\proof{
	Without loss of generality let $i=1, j=2$.
	Clearly, the $12$-faces of $g^{12}$ are planar since
	\begin{align}
		g^{12}(r),g^{12}(r+e_1),g^{12}(r+e_1+e_2),g^{12}(r+e_2) \in \square g(r+e_1+e_2) 
	\end{align}
	for all $r\in \Z^N$. Next, we consider $1k$-faces for $k > 2$. Without loss of generality, we assume $k=3$. For some $r\in \Z^N$, consider the two lines
	\begin{align}
		L &\coloneqq\square g(r+e_1) \cap \square g(r+e_1+e_2),\\
		L_3 &\coloneqq\square g(r+e_1+e_3) \cap \square g(r+e_1+e_2+e_3).
	\end{align}
	The two lines $L(r), L_3(r)$ intersect in $g^{23}(r)$ and therefore span a plane. The $13$-faces of $g^{12}$ are planar since
	\begin{align}
		g^{12}(r),g^{12}(r+e_1),g^{12}(r+e_1+e_3),g^{12}(r+e_3) \in L(r) \vee L_3(r).
	\end{align}
	The argument for the planarity of $23$-faces works analogously. It remains to check the planarity of $kl$-faces for $2 < k < l$, without loss of generality we assume $k=3, l = 4$. This requires a bit of dimension juggling. Note that for any $0<m\leq N, r\in \Z^N$ the span of two adjacent planes $\square g(r) \vee \square g(r+e_m)$ is a 3-space. Similarly, the four planes of $\square g$ around a quad span a 4-space, the eight planes around a cube span a 5-space, and the sixteen planes around a 4-cube span a 6-space.  Consider the eight 3-spaces
	\begin{align}
		A_{abc} &\coloneqq \square g(r+a \cdot e_1+b \cdot e_2 + c \cdot e_4) \vee \square g(r+a \cdot e_1+b \cdot e_2 + e_3 + c \cdot e_4),
	\end{align}
	where $a,b,c\in \{0,1\}$, or equivalently $(a,b,c) \in \Z_2^3$, and where we suppress the $r$ in the notation. Thus we view the eight 3-spaces $A_{abc}$ as associated to the vertices of a 3-cube. Let $q^{12} = (p,p_1,p_{12},p_2)$ be the vectors of a quad in the 3-cube $\Z_2^3$. The intersection $A_p \cap A_{p_1}$ is a plane, since the two spaces belong to a quad of the 4-cube. Consequently, $A_p \cap A_{p_1} \cap A_{p_2}$ is a line. For example, we have that 
	\begin{align}
		g^{12}(r) \vee g^{12}(r+e_3) = A_{000} \cap A_{100} \cap A_{010},\\
		g^{12}(r+e_4) \vee g^{12}(r+e_3+e_4) = A_{001} \cap A_{101} \cap A_{011}.
	\end{align}	
	Moreover, we also have that 
	\begin{align}
		g^{12}(r) \vee g^{12}(r+e_3) = A_{000} \cap A_{100} \cap A_{010} \cap A_{110},\\
		g^{12}(r+e_4) \vee g^{12}(r+e_3+e_4) = A_{001} \cap A_{101} \cap A_{011} \cap A_{111}.
	\end{align}	
	The same argument holds generally, therefore we have that 
	\begin{align}
		B_{q^{12}} \coloneqq A_p \cap A_{p_1} \cap A_{p_2} \cap A_{p_{12}},
	\end{align}	
	 is a line as well. Next, let $q^{13} \coloneqq (p,p_1,p_{13},p_3)$ be the vectors of another quad in the 3-cube $\Z_2^3$, which shares an edge with $q^{12}$. Since both lines $B_{q^{12}}$, $B_{q^{13}}$ are contained in the plane $A_p \cap A_{p_1}$, the two lines intersect in a point. Moreover, the intersection 
	 \begin{align}
	 	A_p \cap A_{p_1} \cap A_{p_2} \cap A_{p_3} = \bigcap_{m=1}^3 (A_p \cap A_{p_1}),
	 \end{align}
	 is a point, since it is the intersection of three planes in the 3-space $A_p$. As a result, the intersection $B_{q^{12}} \cap B_{q^{13}} \cap B_{q^{23}}$ is also a point. In other words, the intersections of the lines associated to three pairwise adjacent quads of the 3-cube intersect in a point. By iterating this statement, we obtain that the intersection of the lines of any two quads of the 3-cube intersect in a point. In particular, we obtain that the intersection
	 \begin{align}
	 	(g^{12}(r) \vee g^{12}(r+e_3)) \cap (g^{12}(r+e_4) \vee g^{12}(r+e_3+e_4))
	 \end{align}
	 is a point. This is equivalent to the planarity of the $34$-faces, which completes the proof.\qed
}

In the next step, we show the converse of the previous lemma, namely that every conjugate net defines a conjugate face-net.

\begin{lemma}\label{lem:conjugatetoconjugateface}
  For every conjugate net $g: F^{ij}_N \rightarrow \RP^n$ there exists a unique conjugate face-net $\lceil g \rceil: F_N \rightarrow \RP^n$,
  such that $\lceil g \rceil$ restricts to $g$ on $F^{ij}_N$. 
\end{lemma}

\proof{
	Note that $g: F^{ij}_N \rightarrow \RP^n$ already defines a *net $\square g: V_N \rightarrow \planesof{\RP^n}$. The completion of $g$ is given by
	\begin{align}
		\lceil g \rceil (f) \coloneqq \bigcap_{v \inc f} \square g(v).
	\end{align}
	We have to show that the four planes that define $\lceil g \rceil(f)$ intersect in a point. If that is true, then it is clear that $\lceil g \rceil$ is a conjugate face-net, since we construct it from planes. Without loss of generality, assume $i=1, j=2, k=3, l=4$. For a $12$-face $f$ the intersection point $\lceil g \rceil(f)$ is just $g(f)$. For a $13$-face $f$ at position $r$, the intersection point is a so called focal point, that is
	\begin{align}
		\lceil g^{13} \rceil(r) = (g(r) \vee g(r + e_2)) \cap  (g(r + e_3) \vee g(r+ e_2 + e_3)).
	\end{align}
	Since these two lines are in a common plane spanned by $g$, the intersection point exists. The argument is the same for $23$-faces. It remains to check $34$-faces. In this case we intersect four planes that are in a 4-dimensional cube (both combinatorially in $\Z^N$ and geometrically in $\RP^n$). Each plane is also the intersection of two 3-cubes in the 4-cube. As such the intersection of the four planes is the intersection of four 3-cubes in a 4-dimensional space, which is a point.\qed
}

As a result of Lemma \ref{lem:conjugatetoconjugateface}, there is indeed a bijection between conjugate nets and conjugate face-nets.

\begin{theorem} \label{th:conjugatefacenetsconsistency}
	Conjugate face-nets are multi-dimensionally consistent 3D-systems.
\end{theorem}
\proof{
	Due to Lemma \ref{lem:conjugatetoconjugateface}, conjugate face-nets are in bijection with conjugate nets. Moreover, due to Theorem \ref{th:conjugatenetsconsistency}, conjugate nets are multi-dimensionally consistent. Therefore, conjugate face-nets are also multi-dimensionally consistent.\qed
}

\begin{theorem}
	Conjugate binets are multi-dimensionally consistent 3D-systems.
\end{theorem}
\proof{
	This is an immediate consequence of Theorem \ref{th:conjugatenetsconsistency} and Theorem \ref{th:conjugatefacenetsconsistency}.\qed
}

\begin{remark}
	As mentioned in the beginning, we generally assume objects to be generic enough, without being precise. In particular, the proofs of Lemma~\ref{lem:conjugatefacetoconjugate} and Lemma~\ref{lem:conjugatetoconjugateface} assume that the ambient dimension is large enough. If this is not the case, it is always possible to consider a lift to an ambient space with large enough dimension. Here a lift means that the central projection of the lift is the original configuration. The statement is then proved in that higher dimensional space and projected back, which does not affect the validity of the claim. This is a common technique in these situations, which was also already applied for conjugate nets \cite[Exercise 2.2]{ddgbook}. Note that for this technique to work, a condition is that the configuration is determined by the same initial data before and after the lift, which is the case here (a counterexample would be line complexes, see \cite{bslinecomplexes}).
\end{remark}

Another concept of discrete integrability is the concept of \emph{consistent reductions}. A reduction is an additional constraint on an $N$-dimensional system. The reduction is called consistent if it suffices to prescribe the constraint on the initial data, and then the constraint is satisfied on all of $\Z^N$. 

In our case, we are interested in polar conjugate binets in $\RP^4$ as a consistent reduction of conjugate binets in $\RP^4$. Recall that a binet $b: D_2 \rightarrow \RP^n$ is polar if for all incident $d,d'\in D_2$ holds that $b(d)$ is polar to $b(d')$. We use the same definition in higher dimensions $N > 2$, except that we replace $D_2$ with $D_N$.

\begin{theorem} \label{th:polarconsistency}
  Polar conjugate binets in $\RP^4$ are a consistent reduction of conjugate binets in $\RP^4$.
\end{theorem}
\proof{
  An equivalent way to characterize the polarity of polar conjugate binets is to require that $\square b(d) \subset b(d)^\pol$ for all $d\in D_N$.
  Therefore, we build polar conjugate binets from initial data by constructing both $b$ and $\square b$ simultaneously.
  We have to check that we can ``complete the cube'' (see Figure~\ref{fig:consistent-reduction}).
  Let
  \begin{align}
    v,v_1,v_2,v_3,v_{12},v_{23},v_{13},v_{123},
  \end{align}
  be the eight vertices of a cube and
  \begin{align}
    f^{12},f^{23},f^{13},f_3^{12},f_1^{23},f_2^{13},
  \end{align}
  be the six faces of the same cube (see Figure~\ref{fig:consistent-reduction}).
  Assume we know $b$ on all vertices of the cube except $v_{123}$ and on the three faces $f^{12},f^{23},f^{13}$.
  Clearly,
  \begin{align}
    \square b(f^{12}_3) &= b(v_{23}) \vee b(v_{3}) \vee b(v_{13}),\\
    \square b(f^{23}_1) &= b(v_{13}) \vee b(v_{1}) \vee b(v_{12}),\\
    \square b(f^{13}_2) &= b(v_{12}) \vee b(v_{2}) \vee b(v_{23}).
  \end{align}
  Analogously,
  \begin{align}
    b(f^{12}_3) &= \square b(v_{23}) \cap  \square b(v_{3}) \cap \square b(v_{13}),\\
    b(f^{23}_1) &= \square b(v_{13}) \cap \square b(v_{1}) \cap \square b(v_{12}),\\
    b(f^{13}_2) &= \square b(v_{12}) \cap \square b(v_{2}) \cap \square b(v_{23}),
  \end{align}
  which holds since adjacent planes intersect in lines. We need to check the three new polarity conditions, that is that $\square b(f^{12}_3)$ is polar to $b(f^{12}_3)$, that $\square b(f^{23}_1)$ is polar to $b(f^{23}_1)$ and that $\square b(f^{13}_2)$ is polar to $b(f^{13}_2)$. However, these are almost trivial. For example, $b(v_{23}) \in \square b(v_{23})^\pol$, $b(v_{3}) \in \square b(v_{3})^\pol$, $b(v_{13}) \in \square b(v_{13})^\pol$ and therefore
  \begin{align}
    \square b(f_3^{12}) &= b(v_{23}) \vee b(v_{3}) \vee b(v_{13})\\
			&\subset \square b(v_{23})^\pol \vee \square b(v_{3})^\pol \vee \square b(v_{13})^\pol\\
			&= (\square b(v_{23}) \cap  \square b(v_{3}) \cap \square b(v_{13}))^\pol\\
			&= b(f^{12}_3)^\pol.
  \end{align}
  Analogously one may check the other two polarity conditions. Finally, we obtain
  \begin{align}
    \square b(v_{123}) &= b(f_3^{12}) \vee b(f_2^{13}) \vee b(f_1^{23}),\\
    b(v_{123}) &= \square b(f_3^{12}) \cap  \square b(f_2^{13}) \cap \square b(f_1^{23}).
  \end{align}
  Again, we need to check polarity, which we do analogously to before:
  \begin{align}
    \square b(v_{123}) &= b(f_3^{12}) \vee b(f_2^{13}) \vee b(f_1^{23})\\
                       &\subset \square b(f_3^{12})^\pol \vee \square b(f_2^{13})^\pol \vee \square b(f_1^{23})^\pol\\
                       &= (\square b(f_3^{12}) \cap  \square b(f_2^{13}) \cap \square b(f_1^{23}))^\pol\\
                       &= b(v_{123})^\pol.
  \end{align}
  This concludes the proof.\qed
}
\begin{figure}[H]
  \centering
  \begin{overpic}[width=0.27\textwidth]{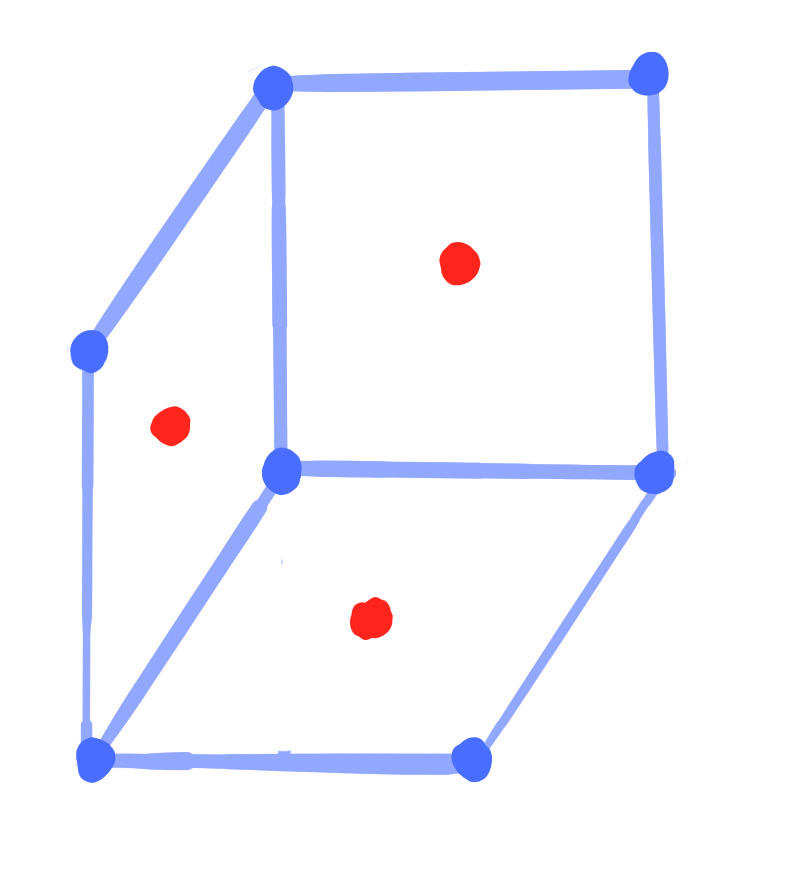}
    \put(34,49){$\color{blue}v$}
    \put(4,6){$\color{blue}v_1$}
    \put(78,45){$\color{blue}v_2$}
    \put(25,95){$\color{blue}v_3$}
    \put(55,6){$\color{blue}v_{12}$}
    \put(76,95){$\color{blue}v_{23}$}
    \put(-3,61){$\color{blue}v_{13}$}
    \put(45,28){$\color{red}f^{12}$}
    \put(55,68){$\color{red}f^{23}$}
    \put(15,42){$\color{red}f^{13}$}
  \end{overpic}
  \hspace{0.1\textwidth}
  \begin{overpic}[width=0.27\textwidth]{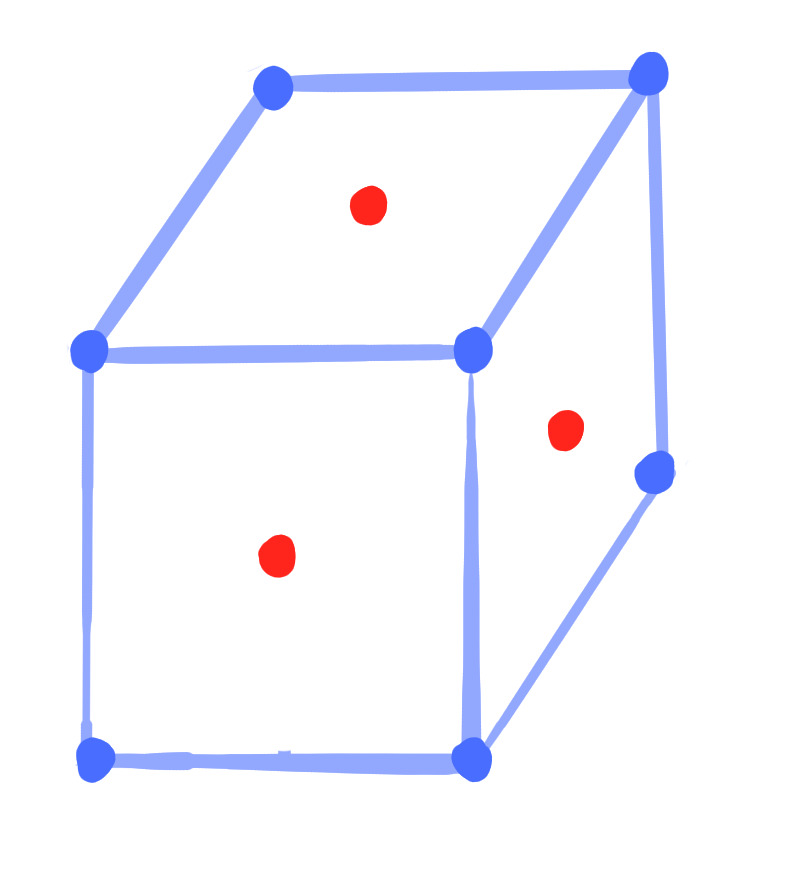}
    \put(57,59){$\color{blue}v_{123}$}
    \put(4,6){$\color{blue}v_1$}
    \put(78,45){$\color{blue}v_2$}
    \put(25,95){$\color{blue}v_3$}
    \put(55,6){$\color{blue}v_{12}$}
    \put(76,95){$\color{blue}v_{23}$}
    \put(-3,61){$\color{blue}v_{13}$}
    \put(44,75){$\color{red}f^{12}_3$}
    \put(34,35){$\color{red}f^{23}_1$}
    \put(59,41){$\color{red}f^{13}_2$}
  \end{overpic}
  \caption{Combinatorics for the consistency around the the cube of a polar binet.}
  \label{fig:consistent-reduction}
\end{figure}
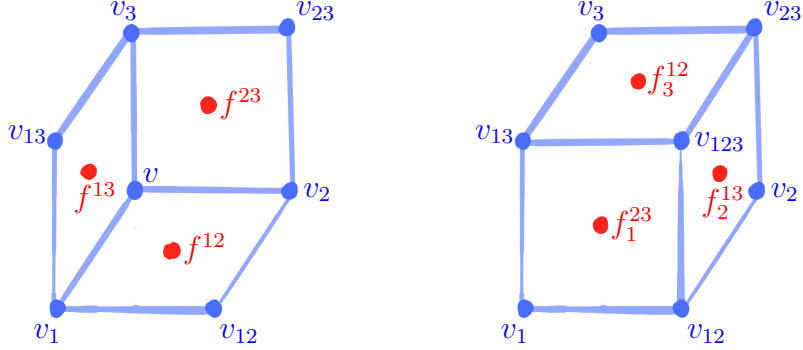

\begin{theorem}
	Principal binets are a consistent reduction of conjugate binets.
\end{theorem}
\proof{
	Let $b$ be a principal binet defined on 2-dimensional initial data. We may choose a Möbius lift $b_\mobq$ on the initial data. The Möbius lift is a polar conjugate binet in $\RP^4$, therefore Theorem \ref{th:polarconsistency} applies. Thus we may extend the Möbius lift from the initial data to all of $D_N$. Then $\pi_\eucl(b_\mobq)$ is a principal binet on all of $D_N$.\qed
}


Previously, we have seen that every conjugate binet $c$ can be obtained as a completion of two conjugate nets $g,h$, that is such that $c = (g, \lceil h \rceil)$. However, since $g$ and $h$ are the same types of nets, one may ask what happens if we complete $g$ instead of $h$. To answer this, introduce the set $W_N$ that is isomorphic to $F_N$, but such that $W^{12}_N$ is identified with $V_N$. So for the moment, let us think of $g$ as defined on the restriction $W^{12}_N = V_N$ of the set $W_N \simeq F_N$. Then for all $v \in W_N$ we define
\begin{align}
	\lceil g\rceil(v) \coloneqq \bigcap_{f \in F_N^{12}, f\inc v} \square g(f).
\end{align}

\begin{theorem}
	Let $g,h: \Z^n \rightarrow \eucl$ be two conjugate nets. Then the binet $(g,\lceil h\rceil): V_N \cup F_N \rightarrow \eucl$ is a principal binet if and only if $(\lceil g\rceil, h): W_N \cup F^{12}_N \rightarrow \eucl$ is a principal binet.
\end{theorem}

\proof{
	Let $p \coloneqq (g,h)$ be the pair of conjugate binets, and $c \coloneqq (g,\lceil h\rceil)$ and $c' \coloneqq (\lceil g\rceil, h)$ be the two completions of $p$. If $c$ is a principal binet, then its Möbius lift $c_\mobq$ is a polar conjugate binet.  The restriction of $c_\mobq$ to $V_N \cup F^{12}_N$ is a pair of conjugate nets $p_\mobq = (g_\mobq, h_\mobq)$. It is straight-forward to show that $c_\mobq = (g_\mobq, \lceil h_\mobq \rceil)$. Moreover, $p_\mobq$ has the property that for all $v \in V$ holds $\square h_\mobq(v) \subset g(v)^\pol$. Analogously for all $f \in F$ holds $\square g_\mobq(f) \subset h(f)^\pol$. Next, we claim that $c'_\mobq \coloneqq (\lceil g_\mobq \rceil,  h_\mobq)$ is a Möbius lift of $c'$. Due to Lemma \ref{lem:conjugatetoconjugateface}, $c'_\mobq$ is a conjugate binet. It is clear that $\pi_\eucl(c'_\mobq) = c'$. Therefore it remains to show that $c'$ is a polar binet. The completion of $g_\mobq$ is defined via
	\begin{align}
		\lceil g_\mobq\rceil(v)   = \bigcap_{f \in F_N^{12}, f\inc v} \square g_\mobq(f),
	\end{align}
	for all $v \in W_N$. The point $\lceil g_\mobq\rceil (v)$ is polar to $h_\mobq(f)$ whenever $v$ is adjacent to $f$, because
	\begin{align}
		\lceil g_\mobq\rceil(v) = \bigcap_{f' \in F_N^{12}, f\inc v} \square g_\mobq(f') \subset \square g_\mobq(f) \subset h_\mobq(f)^\pol.
	\end{align}
	Hence, $c'_\mobq$ is indeed a Möbius lift of $c'$, which means $c'$ is a principal binet.\qed
}

Therefore, one could argue that a pair of conjugate nets $g,h$ which have a polar Möbius lift in the sense of the proof of the previous theorem are actually defining a principal binet. In this way, the relation between $g$ and $h$ becomes completely symmetric. However, as a pair they are not symmetrically defined on $\Z^N$, since in this way the $12$-directions are distinguished. On the other hand, if we think of binets on $\Z^N$ as transformations of discrete surfaces, it does make sense to have a distinguished $12$-direction.

\begin{remark}
	There is an alternative route to obtain the multi-dimensional consistency of principal binets in $\Z^N$ directly via the Lie lift. We provide a short discussion in this remark. In the $\Z^2$ case, the Lie lift $b_\lieq$ of a principal binet $b$ is by definition a line bicongruence (see Definition \ref{def:lielift}). There is a generalization of line congruences to $\Z^N$ called line complexes, which was introduced by Bobenko and Schief \cite{bslinecomplexes}. A \emph{line complex} is a map $\ell: V_N \rightarrow \mbox{Lines}(\RP^n)$ such that adjacent lines intersect. They also showed that line complexes are multi-dimensionally consistent. It turns out that the restriction $g_\lieq$ of a Lie lift $b_\lieq$ to $V_N$ is indeed a line complex. In fact, any line complex $g_\lieq$ in $\RP^5$ is the restriction of a Lie lift. Moreover, the restriction $g_\lieq$ determines the whole Lie lift $b_\lieq$ uniquely, since
	\begin{align}
		b_\lieq(f) = \cap_{v\sim f} (g_\lieq(v))^\perp = \left(\bigvee_{v\sim V} g_\lieq(v) \right)^\perp
	\end{align}
	is a unique line, as it is the polar complement of a 3-space in $\RP^5$. Note that the restriction of $b_\lieq(f)$ to $F_N$ is \emph{not} a line complex. Instead, the restriction to $F_N$ is another generalization of a line complex, where lines are defined on faces of $\Z^N$ such that the lines of all faces around an edge span a plane. It is clear that this generalization is multi-dimensionally consistent in $\RP^5$ due to the polarity construction from line complexes. We claim this generalization is also multi-dimensionally consistent in any $\RP^N$ with $N > 5$, and may be modified to be multi-dimensionally consistent also for $N = 3,4$ analogously to \cite{bslinecomplexes}, but since we do not need this here, we do not provide a proof.
\end{remark}

\begin{remark}
  \newcommand{\mathbbx}{\mbox{\boldmath $x$}}
  One of the motivating examples in the introduction (Example~\ref{itm:exampleconfocal}) are discrete confocal quadrics \cite{bsstconfocali, bsstconfocalii} (see Figure~\ref{fig:discrete-confocal}, left).
  Consider $(\Z^3)^*$ -- the dual of $\Z^3$ -- as the set
  \begin{align}
    (\Z^3)^* \simeq \Z^3 + \frac12(1,1,1),
  \end{align}
  so that every point in $(\Z^3)^*$ corresponds to a unit-cube of $\Z^3$.
  Discrete confocal quadrics are defined as a pair of a map $\mathbbx$ defined on $\Z^3$ and a map $\mathbbx^*$ defined on $(\Z^3)^*$ with some constraints. In particular, for every edge $(v,v')$ and dual edge $(c,c')$ such that $v,v' \in c\cap c'$ the corresponding lines are orthogonal, that is
  \begin{align}
    (\mathbbx(v) \vee \mathbbx(v')) \perp (\mathbbx^*(c) \vee \mathbbx^*(c')).
  \end{align}
  In this sense, discrete confocal quadrics are orthogonal nets.
  In fact, discrete confocal quadrics are a discretization of \emph{orthogonal coordinate systems} -- orthogonal parametrizations of $\eucl$.
  Moreover, it turns out that as a consequence of this orthogonality condition the images of all quads of both $\mathbbx$ and $\mathbbx^*$ are planar,
  hence discrete confocal quadrics are also conjugate nets.
    In fact, if we view one $\Z^2$ slice together with an adjacent $(\Z^2)^*$ slice
  of a discrete confocal quadric the definition coincides with our definition of a principal binet (see Figure~\ref{fig:discrete-confocal}, right) \cite{hstellipsoid}.
  However, in $\Z^3$ the combinatorics of discrete confocal quadrics are clearly different from the combinatorics that we use for principal binets. 
  At the moment of writing, it is unclear if there is any relation between the two discretizations.  
   If there is not, then principal binets defined on $D_N$ may provide a different discretization of orthogonal coordinate systems. However, it would appear that one should not consider all points of a principal binet to be part of the same orthogonal coordinate system, possibly only the restrictions to two conjugate nets on $V_N \cup F^{12}_N$ as discussed before.
\end{remark}
\begin{figure}[H]
  \centering
  \includegraphics[width=0.49\textwidth]{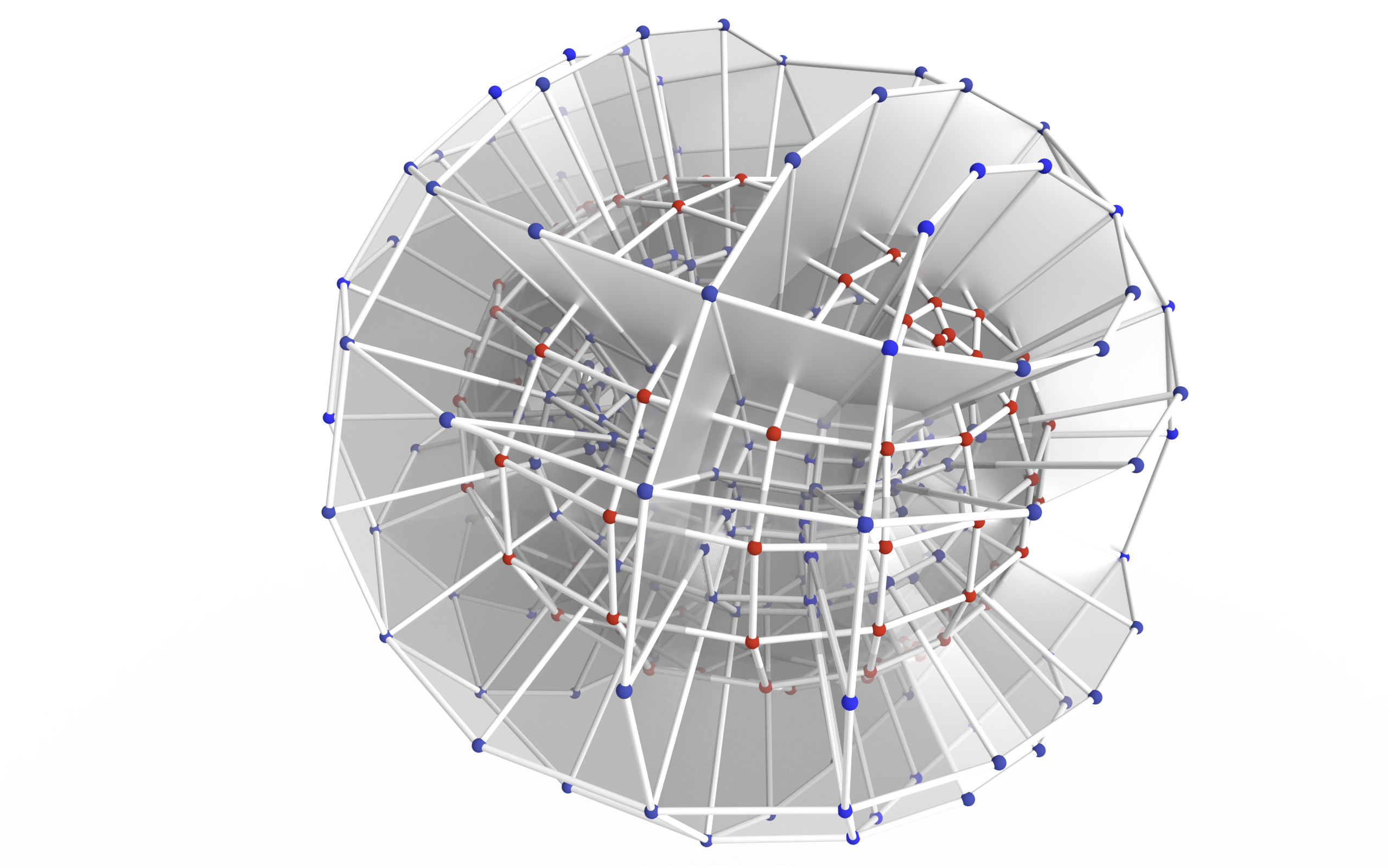}
  \includegraphics[width=0.49\textwidth]{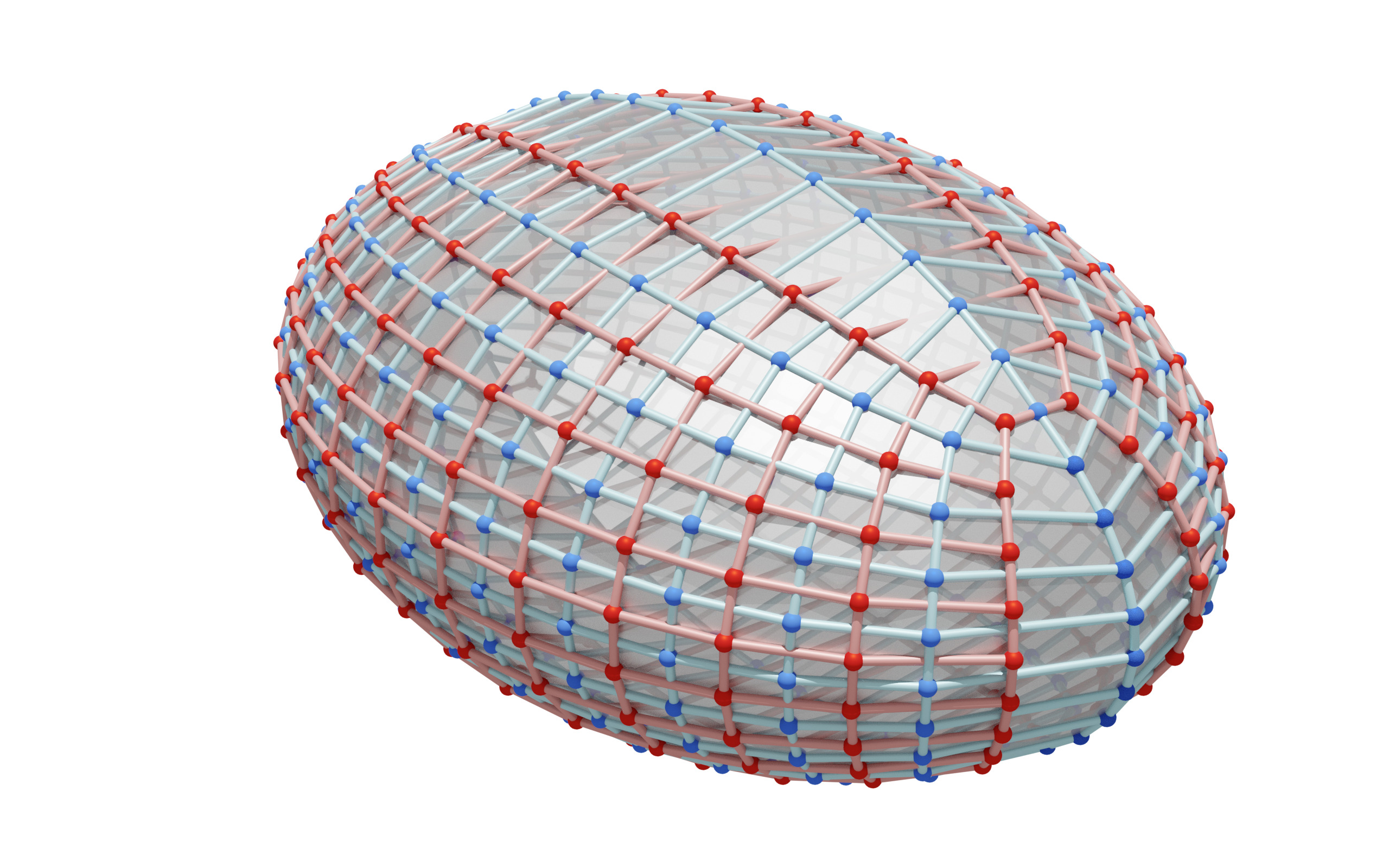}
  \caption{
    Left: Discrete confocal quadrics (blue vertices on $\Z^3$ and red vertices on $(\Z^3)^*$).
    Right: Discrete ellipsoid as a 2-dimensional slice of discrete confocal quadrics (red vertices on $\Z^2$ and blue vertices on $(\Z^2)^*$).
  }
  \label{fig:discrete-confocal}
\end{figure}

\newpage
\appendix

\section{Appendix: Overview of notation}\label{sec:cheatsheet}

A cheat sheet for our plentiful notation.

{\renewcommand{\arraystretch}{1.2}
\begin{tabular}{|l|l|}
	\hline
	 $V \simeq \Z^2$ &  vertices \\
	\hline
	$E = E(\Z^2)$ & edges \\
	\hline
	$F \simeq (\Z^2)^*$ & faces \\
	\hline
	$D = V\cup F$ & double \\
	\hline
	$C \simeq E$ & crosses \\
	\hline
\end{tabular}}
{\renewcommand{\arraystretch}{1.2}
\begin{tabular}{|l|l|}
	\hline
	 $\pl$ &  planes \\
	\hline
	 $\opl$ & oriented planes \\
	\hline
	 $\apl$ & angled planes \\
	\hline
	 $\sp$ &  spheres \\
	\hline
	 $\osp$ & oriented spheres \\
	\hline
	 $\asp$ & angled spheres \\
	\hline
\end{tabular}}
{\renewcommand{\arraystretch}{1.2}
\begin{tabular}{|l|l|}
	\hline
	 $\lieq$ &  Lie quadric \\
	\hline
	 $\mobq$ & Möbius quadric \\
	\hline
	 $\blac$ & Blaschke cylinder \\
	\hline
	 $\eucl$ &  Euclidean space \\
	\hline
	 $\eucl^*$ & Euclidean dual space \\
	\hline
	 $\euclq$ & absolute conic \\
	\hline
\end{tabular}}

\begin{minipage}{4cm}
	Subsets of $\RP^5$ \\and their \\(defining) relations.
\end{minipage}
{\renewcommand{\arraystretch}{1.2}\renewcommand{\tabcolsep}{0.32cm}
\begin{tabular}{|l|l|l|}
	\hline
	 $B \in \lieq$ &  $M \in \lieq^+$ & $B \perp M$ \\
	\hline
	 $\blac = B^\perp \cap \lieq$ & $\mobq = M^\perp \cap \lieq$ & $\ell_\blac=M\vee B$  \\
	\hline
	 $S_\eucl \subset M^\perp$ & $E^\infty = S_\eucl \cap B^\perp$ & $\eucl = S_\eucl \setminus E^\infty$ \\
	\hline
	 $S_{\eucl^*} = M^\perp \cap B^\perp$ &  $\eucl^* = S_{\eucl^*} \setminus \{B\} $ & $\euclq = E^\infty \cap \mobq$ \\
	\hline
	 $S_\unis \subset B^\perp$ & $S_\unis \ni M$ & $\unis = \blac \cap S_\unis$  \\
	\hline
\end{tabular}}

\begin{minipage}{4cm}
	Projections and \\identifications.
\end{minipage}
{\renewcommand{\arraystretch}{1.2}
\begin{tabular}{|ll|}
	\hline
	$\pi_{\eucl} : \RP^5 \setminus \ell_\blac \rightarrow S_\eucl,$ & $X \mapsto (X \vee \ell_\blac) \cap S_\eucl$ \\
  	\hline
	$\xi_\sp: \RP^5 \setminus \{\PM\} \rightarrow \sp,$ & $X \mapsto \pi_\eucl( X^\pol \cap \mobq )$ \\
	\hline
	$\pi_{\eucl^*}: B^\perp \setminus \{ \PM \} \rightarrow S_{\eucl^*},$ & $X \mapsto  (X \vee \PM) \cap S_{\eucl^*}$ \\
	\hline
	$\xi_\pl : B^\perp \setminus \ell_\blac \rightarrow \pl,$ & $X \mapsto \star\pi_{\eucl^*}(X) \subset \eucl$ \\
	\hline
	$\pi_{\unis} : B^\perp \setminus \{\PB\} \rightarrow S_\unis,$ & $ X \mapsto (X \vee \PB) \cap S_\unis$ \\
	\hline
	$\xi_\opl: \blac \setminus \{\PB\} \rightarrow \opl \simeq \pl \times \unis,$ & $X \mapsto (\xi_\pl(X), \pi_{\unis}(X))$  \\
	\hline
	$\xi_\nci : B^\perp \setminus \ell_\blac \rightarrow \circlesof{\unis},$ & $X \mapsto \pi_{\unis}(X)^\perp \cap \unis$ \\
	\hline
	$\xi_\apl : B^\perp \setminus \ell_\blac \rightarrow \apl,$ & $X \mapsto \left(\xi_\pl(X), \xi_{\nci}(X)\right)$ \\
	\hline
	$\xi_\osp: \RP^5 \setminus \ell_\blac \rightarrow \osp,$ & $\xi_\osp(X) \cong \xi_\opl(X^\perp \cap \blac) $ \\
	\hline
	$\xi_\asp: \RP^5 \setminus \ell_\blac \rightarrow \asp \simeq \sp \times \osp,$ & $X \mapsto \left(\xi_\sp(X), \xi_{\osp}(X)\right)$ \\
	\hline
\end{tabular}}

\begin{minipage}{2.95cm}
	Binets, lifts and \\representations.
\end{minipage}
{\renewcommand{\arraystretch}{1.2}
\begin{tabular}{|l|l|}
	\hline
	 $b_\mobq: D \rightarrow M^\perp$ &  Möbius lift \\
	\hline
	$b_\sp = \xi_\sp\circ b_\mobq: D \rightarrow \sp$ & orthogonal sphere rep. \\
	\hline
	$b_\blac: D \rightarrow B^\perp$ & Laguerre lift \\
	\hline
	$b_\nci = \xi_{\nci} \circ b_\blac: D \rightarrow \circlesof{\unis}$ & orthogonal circle rep. \\
	\hline
	$b_\apl = \xi_{\apl} \circ b_\blac: D \rightarrow \apl$ & angled plane rep. \\
	\hline
	$b_\lieq = b_\blac \vee b_\mobq: D \rightarrow \linesof{\RP^5}$ & Lie lift \\
	\hline
	$n: D \rightarrow S_\unis$ & normal binet \\
	\hline
	$N: D \rightarrow \linesof{\eucl}$ & normal congruence \\
	\hline
	$\circl: D \rightarrow \circlesof{\eucl}$ & $\circl(d) \cong \pi_\eucl (\square b_\mobq(d) \cap \mobq)$ \\
	\hline
	$\cone: D \rightarrow \conesof{\eucl}$ & $\cone (d) \cong \xi_\opl (\square b_\blac(d) \cap \blac)$ \\
	\hline	
	$\square b: D \rightarrow \pl$ & planes of conj.~binet $b$ \\
	\hline
	$\square^* b: D \rightarrow \pl$ & points of conj.~bi*net $b$ \\
	\hline	
\end{tabular}}

\bibliographystyle{alpha}

\bibliography{references}

\end{document}